\documentclass{article}
\usepackage{graphicx} 
\usepackage{subcaption}  
\usepackage{geometry}
\usepackage{color}
\usepackage{amsthm}
\usepackage{thmtools,thm-restate}
\usepackage{amsfonts}
\usepackage{amsmath}
\usepackage{float}
\usepackage{hyperref}
\usepackage[capitalize,nameinlink]{cleveref}
\usepackage[dvipsnames]{xcolor}
\usepackage{mathrsfs}
\usepackage{amsmath}
\usepackage{bbm}
\hypersetup{
	colorlinks=true,
	pdfpagemode=UseNone,
    citecolor=OliveGreen,
    linkcolor=NavyBlue,
    urlcolor=black,
	pdfstartview=FitW
}
\DeclareUnicodeCharacter{221E}{\ensuremath{\infty}}
\usepackage[linesnumbered, ruled, noend]{algorithm2e}
\usepackage{mathtools}
\theoremstyle{definition}
\newtheorem{thm}{Theorem}
\newtheorem{defn}{Definition}
\newtheorem{lem}{Lemma}
\newtheorem{cor}{Corollary}

\newtheorem{prop}{Proposition}

\newcommand{\dg}[1]{{\small \color{brown}[David: #1]}}
\newcommand{\yatin}[1]{{\small \color{blue}[Yatin: #1]}}
\newcommand{\fran}[1]{{\small \color{red}[Fran: #1]}}

\newcommand{\N}{\mathbb{N}}
\newcommand{\wt}[1]{\widetilde{#1}}
\newcommand{\R}{\mathbb{R}}
\newcommand{\E}{\mathbb{E}}

\newcommand{\calG}{\mathcal{G}}
\newcommand{\1}{\mathbbm{1}}
\newcommand{\pequiv}{\stackrel{P}{\simeq}}
\newcommand{\calF}{\mathcal{F}}
\newcommand{\calS}{\mathcal{S}}

\newcommand{\Var}{\mathsf{Var}}
\newcommand{\GOE}{\mathrm{GOE}}
\newcommand{\calC}{\mathcal{C}}
\newcommand{\calA}{\mathcal{A}}
\newcommand{\calR}{\mathcal{R}}
\newcommand{\eps}{\varepsilon}
\newcommand{\op}{\mathrm{op}}

\newcommand{\norm}[1]{\left\lVert#1\right\rVert}
\newcommand{\Unif}{\mathrm{Unif}}
\newcommand{\sign}{\mathrm{sign}}
\newcommand{\brac}[1]{\left\langle #1 \right\rangle}
\newcommand{\bfsigma}{\boldsymbol{\sigma}}
\newcommand{\bfh}{\boldsymbol{h}}
\newcommand{\bfu}{\boldsymbol{u}}
\newcommand{\bfz}{\boldsymbol{z}}
\newcommand{\bfg}{\boldsymbol{g}}
\newcommand{\bfzeta}{\boldsymbol{\zeta}}
\newcommand{\bfxi}{\boldsymbol{\xi}}

\newcommand{\ignore}[1]{\relax}

\title{Sequential Dynamics in Ising Spin Glasses}
\author{Yatin Dandi\thanks{Statistical Physics of Computation Laboratory, École polytechnique fédérale de Lausanne (EPFL) CH-1015 Lausanne. Email: yatin.dandi@epfl.ch.}, David Gamarnik\thanks{Sloan School of Management, 
Operations Research Center and Institute of Data,  Systems and Society (IDSS), MIT. Email: gamarnik@mit.edu. }, Francisco Pernice\thanks{CSAIL and LIDS, MIT. Email: fpernice@mit.edu.}, Lenka Zdeborová\thanks{Statistical Physics of Computation Laboratory, École polytechnique fédérale de Lausanne (EPFL) CH-1015 Lausanne. Email: lenka.zdeborova@epfl.ch.}}
\date{}

\begin{document}

\maketitle

\begin{abstract}
We present the first exact asymptotic characterization of sequential dynamics for a broad class of local update algorithms on the Sherrington-Kirkpatrick (SK) model with Ising spins. Focusing on dynamics implemented via systematic scan—encompassing Glauber updates at any temperature—we analyze the regime where the number of spin updates scales linearly with system size. Our main result provides a description of the spin-field trajectories as the unique solution to a system of integro-difference equations derived via Dynamical Mean Field Theory (DMFT) applied to a novel block approximation. This framework captures the time evolution of macroscopic observables such as energy and overlap, and is numerically tractable. Our equations serve as a discrete-spin sequential-update analogue of the celebrated Cugliandolo-Kurchan equations for spherical spin glasses, resolving a long-standing gap in the theory of Ising spin glass dynamics. Beyond their intrinsic theoretical interest, our results establish a foundation for analyzing a wide variety of asynchronous dynamics on the hypercube and offer new avenues for studying algorithmic limitations of local heuristics in disordered systems.

\end{abstract}

\tableofcontents

\section{Introduction}
We consider a discrete sequential dynamics on the \emph{Sherrington-Kirpatrick (SK) model} in a linear time scale. 
The SK model is the prototypical example of a model of a \emph{disordered system}. It was proposed in 1975 \cite{SK75} as a simplified, mean-field version of the Edwards-Anderson model \cite{EA75}, whose original motivation in physics was to understand the intriguing properties of certain dilute magnetic alloys. The mathematical definition is as follows. Let $J$ be a $\GOE(N)$ matrix, i.e., an $N\times N$ matrix with centered, jointly Gaussian entries with covariance $\E J_{ij}J_{kl} = (\delta_{ik}\delta_{jl} + \delta_{il}\delta_{jk})/N.$ The \emph{Sherrington-Kirpatrick Hamiltonian}, also referred to below as the \emph{energy},\footnote{In physics the Hamiltonian and the energy are usually defined with the opposite sign.} is the random function $H:\{\pm 1\}^N \to \R$ given by 
\begin{equation}\label{def:ham}
    H(\bfsigma) = \frac{1}{2}\brac{\bfsigma, J\bfsigma}.
\end{equation}
Then, given an inverse temperature parameter $\beta \in [0,\infty),$ associated with the SK model is its Gibbs measure: 
the (random) measure on $\{\pm1\}^n$ given by
\begin{align}\label{eq:Gibbs-measure}
    \mu_\beta(\bfsigma) &= \frac{1}{Z_\beta} \exp(\beta H(\bfsigma)),
\end{align}
where $Z_\beta$ is the normalization constant so that $\sum_{\bfsigma} \mu(\bfsigma)=1,$ called the \emph{partition function.} 

A few years after this model was proposed, in a series of groundbreaking papers \cite{ParisiRSB1, ParisiRSB2, ParisiRSB3}, Parisi non-rigorously found an exact formula for the limiting \emph{free energy} $\lim_N \frac{1}{N}\E \log Z_\beta.$ The subsequent interpretation of Parisi's solution led to a remarkably complete non-rigorous understanding of the structure of the model at equilibrium (i.e., of the measure $\mu_\beta).$ Many of the most important aspects of this picture have now been proved; chiefly among them, the Parisi formula for the free energy \cite{Guerra03, Talagrand06} and the ultrametricity conjecture \cite{Panchenko13}.

While a great deal of progress has been achieved in recent decades in understanding the equilibrium properties of this model~\cite{panchenko2013sherrington}, far less is known about its dynamical aspects under natural dynamics on the configuration space $\{\pm 1\}^N$. This is the focus of our paper. The dynamics we analyze are as follows. First, we fix an ordering of the $N$ spins, say $\bfsigma_1,\dots, \bfsigma_N.$ Then, we initialize at a uniformly-random configuration and perform $T$ in-order scans through the $N$ spins. In each scan, we iterate through the spins $\bfsigma_1,\dots,\bfsigma_N$ one at a time, updating each one according to a pre-specified (randomized) update rule that takes as input the spin's local field. We allow for any suitably regular update rule, including as a special case Glauber updates at any temperature.

To analyze this algorithm, we consider as an approximation a version of the dynamics with $\delta$-blocks, for $\delta \in (0,1]$: the spin set $[N]$ is split into
$1/\delta$ groups of size $\delta N$
each, and updates of spins within each group are
done simultaneously, while transitions between groups are done in a scanning fashion as before. 
Our main result provides a full asymptotic description of 
the resulting dynamics when $T$ is a 
constant independent of $N$, in the double limit where first $N\to\infty$ and then $\delta\to 0$. 
We conjecture that the resulting equations agree with those of the original sequential dynamics without block updates. We will formally address this conjecture in upcoming work, but for now we verify it experimentally in \Cref{sec:numerics}. In upcoming work we also derive the equations for Glauber dynamics, where in each step a uniformly randomly chosen spin gets updated.

The equations we derive for the $\delta$-block
dynamics specify the asymptotic joint
distribution of $\{\sigma_x^t\}_{t=0}^T,\{h^t_x\}_{t=1}^T$ for $x\in [0,1],$ where $\sigma_x^t$ is the spin value of the $xN$-th spin at the end
of scan $t$, and $h^t_x$ the corresponding local field
experienced by this spin. 
This joint distribution up to scan $T$ is described via a system of 
$O(T^2)$-dimensional differential equations we derive, from which many useful observables can be computed, such
as the energy of the system at the end of scan
$T$. While the system of differential equations is somewhat involved and grows with $T$, it is easy
to implement numerically, which we do. We report our 
numerical findings in  \Cref{sec:numerics}.
To the best of our knowledge, such a detailed description of the spin glass dynamics at a linear time scale was unknown in both the physics and mathematics literatures. 
Our results are of relevance to several prior lines of research. We highlight the two most relevant ones.

\subsection{Prior work in physics and probability 
literature}
The first line of investigation relates to the long
history of studies in physics
and probability theory of time dependent
aspects of spin glasses: what does the spin glass system
look like at a given time $t$ when it evolves according
to a natural dynamics such as Glauber dynamics? These
questions have been raised in physics already several decades
ago, see e.g. \cite{SZ82, CK93}. The matter is  of particular interest because many of the intriguing features observed in the materials that motivated the study of these models in the first place were dynamical in nature, like the so-called \emph{aging} phenomenon (see \cite{BaDG01} for a mathematical treatment). For the Sherrington-Kirpatrick model, the most natural dynamics are Markov chains on the state space $\lbrace\pm 1\rbrace^N$, with \emph{Glauber dynamics} being the most widely studied version. However, exact asymptotic analysis of
dynamics of this kind, with discrete asynchronous updates,  has so far been out of reach for physicists and mathematicians alike. 

Other spin glasses are more tractable for the study of dynamical properties. Cugliandolo and Kurchan \cite{CK93} studied a model related to SK, known as the \emph{pure spherical $p$-spin model}, where the quadratic form \eqref{def:ham} is replaced by a homogeneous degree-$p$ polynomial with Gaussian coefficients, and the state space is relaxed from the hypercube to the sphere in $\R^N.$ In this model, the natural dynamics is the Langevin diffusion on the sphere. 
Cugliandolo and Kurchan non-rigorously derived an exact asymptotic description of these Langevin-type dynamics for the spherical $p$-spin models in the \emph{out of equilibrium} regime, where one runs the dynamics for a time independent of the system size $N$ ---the continuous analogue of running Glauber dynamics for linearly in $N$ many updates. These \emph{Cugliandolo-Kurchan equations} were proven rigorously by Ben Arous, Dembo and Guionnet \cite{BaDG06} after earlier work by the same authors \cite{BaG95, BaG97, Gionnet97,  BaDG01}. The asymptotic energy achieved by the
dynamics in this time scale was rigorously proven only recently~\cite{sellke2024threshold}, confirming a prediction of~\cite{CK93}.

Some large deviations results for time dependent
observables of the SK model are derived
in~\cite{grunwald1996sanov}, in principle giving an asymptotic description of the dynamics in the out of equilibrium regime. However, this characterization is in terms of the full empirical measure over histories of pairs of spins and local fields, an object which has proven too rich to extract useful information from. By contrast, like in the case of \cite{CK93}, our equations can be written in terms of the correlation and response functions only, which are finite-dimensional objects. 

A description of the dynamics
corresponding to the Glauber
dynamics is also reported  in~\cite{maclaurin2021emergent}, but it appears to be correct only for 
sufficiently high temperature~\cite{Lauren-private}.

The discrete (non-spherical) SK model itself can be treated via Langevin dynamics as well \cite{CK-SK94} by relaxing the hypercube $\{\pm 1\}$ to the real line, and adding a potential which penalizes the dynamics for being far from $\{\pm 1\}$. However, its relevance to the genuinely
discrete dynamics is not entirely clear.

The dynamical behavior directly on the hypercube was also treated in the physics literature, but so far only for the case of synchronous or parallel dynamics where every spin is updated at every time step \cite{eissfeller1992new,eissfeller1994mean,VBKZ24}. The properties of the parallel dynamics are, however, very different from those of sequential dynamics where the spins are changing one by one. For example, in the latter case the energy is non-decreasing in time, whereas in parallel dynamics energy can be non-monotone, and limiting cycles can appear in large times; see examples of this in \cite{VBKZ24}. We remark that the parallel dynamics studied in \cite{VBKZ24} corresponds to setting $\delta = 1$, i.e., simultaneously updating a single block of all spins, and hence our analysis also provides a proof of the limiting equations established therein.

In summary, the dynamics corresponding to sequential updates of discrete
spins (as opposed to relaxations) 
 has remained an outstanding open problem in both the physics and mathematics of disordered systems. In particular, no analogue of the asymptotically exact Cugliandolo-Kurchan equations is known in that case. This is the problem
 which we address in the present work. 

Due to the absence of an exact characterization of the dynamics, physicists have resorted instead to simulations, which have demonstrated that the behavior of the dynamics depends strongly on the details of the spin update protocol. 
In particular, Parisi \cite{Parisi-experimental}  ran simulations on three natural dynamics at \emph{zero temperature}. He called them the \emph{greedy algorithm}, the \emph{reluctant algorithm} and the \emph{sequential algorithm}. In the greedy algorithm, one finds at each step the single spin-flip which would result in the greatest energy increase, and performs it. In the reluctant dynamics, one finds the single spin-flip which would result in the minimal positive energy change, and performs it. The sequential algorithm is the algorithm described at the beginning of the section and which we analyze in this paper. We will also call it  \emph{systematic scan} following \cite{dyer2008dobrushin,dyer2006systematic}.
Recall that this is the dynamics where 
one starts with a uniformly-random configuration $\bfsigma^0\in \{\pm 1\}^N$ and repeatedly scans through the spins from left to right, flipping sequentially each one 
according to a prescribed rule, for example via a zero-temperature Glauber update. 
The literature also refers to this type of update as ordered sequential update \cite{fouladvand1999multi}, or systematic update \cite{peres2013can}. Variants of these protocols were studied in many follow-up numerical papers, e.g. \cite{bussolari2003energy,contucci2005interpolating,contucci2005finding,VBKZ24}, but an exact asymptotic characterization of any of the asynchronous ones remained  open prior to our work.

\subsection{Prior works in theoretical computer science
literature}
There exists an extensive theoretical computer 
science literature regarding the power
and limitations of MCMC based methods such as Glauber
dynamics in the context of various statistical mechanics 
models, including spin glasses~\cite{sinclair1989approximate}. 
The primary focus of the 
computer science literature is understanding the conditions
under which the chain converges (mixes) to stationarity polynomially fast, as this provides a polynomial time algorithm
for approximately sampling from the stationary distribution. Fast and slow mixing of Glauber
dynamics has been studied recently also in the context
of spin glasses. In particular, 
fast mixing is known to hold
at sufficiently high temperature~\cite{eldan2022spectral}.
Slow mixing 
is known to hold when the model exhibits a certain
clustering (shattering) phenomenon, which $p$-spin models
provably exhibit for large enough $p$ and a certain
temperature 
range~\cite{gamarnik2023shattering},
\cite{alaoui2024near}.
Fast sampling algorithms based on methods other 
than MCMC are  also  
known~\cite{alaoui2023sampling, huang2024sampling}, as well as simulated annealing \cite{huang2024weak}.


The techniques for establishing
slow (super-polynomial or exponentially
large) mixing times
are well-developed and typically rely on 
Poincare-Cheeger inequalities and their variants.
There is an important caveat however associated with this slow mixing
approach in general: it verifies
slow mixing only with respect
to the \emph{worst-case} starting states. Namely, 
when applicable it implies existence of \emph{some}
states from which time to mix is large. This
does not rule out the possibility of fast mixing
from \emph{typical} (random) starting 
states. Understanding mixing from random 
starting states is more physically relevant and is a challenging question,
currently poorly understood. 
In addition to the references provided in the 
previous section related to aging in spherical glasses,
some  recent
steps in this direction 
include~\cite{chen2025almost, gamarnik2024hardness} and \cite{huang2024weak}.

The computer science literature has rarely focused on what can be said about Markov chains in regimes where mixing is slow. An important recent exception is \cite{liu2024locally}, where the authors propose a framework of \emph{locally stationary distributions}. These are distributions which Markov chains always reach fast (even in the absence of mixing) and from which useful algorithmic consequences can be derived. For example, the authors show that Glauber dynamics can achieve constant correlation with the planted signal in the spiked matrix model and the stochastic block model for sufficiently high signal to noise ratio, even if the Markov chain fails to mix. By contrast, in our work, instead of deriving approximate guarantees about the state of the Markov chain at large times, we give an \emph{exact} description of the state of the chain at \emph{all} times from the point of view of macroscopic observables in the linear time scale. Through this lens, we anticipate that our work will prove useful to extract algorithmic consequences for a wide variety of models, though we do not focus on this aspect in the present work.

\subsection{The proof technique}
Our method of analysis relies
on the iterative Gaussian conditioning framework developed in \cite{bolthausen2014iterative, bayati2011dynamics}. 
This is a widely applicable method 
of analysis of 
Gaussian processes when conditional on constantly
many linear projections. 
In particular, it is the basis of the asymptotic analysis of the family of 
Approximate Message Passing (AMP)
algorithms originating 
from the TAP equations ~\cite{thouless1977solution,bolthausen2014iterative}. AMP algorithms have enjoyed a
 wide range of 
applications both in  spin glass 
theory~\cite{montanari2024optimization} 
and beyond, and in particular in the field of  
statistical  
inference~\cite{feng2022unifying}. Recently, the applicability of the Gaussian conditioning technique has been extended to a general class of first order algorithms on Gaussian data in \cite{celentano2021high,gerbelot2024rigorous}, recovering the equations described in statistical physics under Dynamical Mean Field Theory 
(DMFT)~\cite{Sompolinsky_88}. Earlier rigorous proofs of such DMFT equations were based on large deviations over paths \cite{BaDG06,GBA03,BaG95,BaG97,BaDG01}. The algorithms in   \cite{celentano2021high,gerbelot2024rigorous} involve simultaneously updating all the coordinates, such as Langevin dynamics and  gradient descent. In contrast, the sequential spin-by-spin updates in our setting result in limiting equations possessing novel structure such as the inter-spin interactions accumulating into a system of ODEs.

Roughly speaking, the method is based on the
following observation: given a GOE matrix $A$, 
its distribution when conditioned
on projections on finitely many vectors
$v_1,\ldots,v_p$ can be nicely decomposed into two terms. The first term corresponds to the values of the 
projections,  and the second term
corresponds to a fresh copy $\tilde A$ of $A$
independent of these projections, see 
\Cref{prop:gen_gauss_cond} in 
\Cref{sec:finite-delta-dynamics}. This
fact allows for convenient and succinct 
representation of the states of the system
after constantly many updates of various kinds,
such as the ones underlying AMP.

Unfortunately, this key property underlying
Gaussian conditioning does not appear to be of use
for the analysis of sequential algorithms
studied in this paper, since those
correspond to the number of projections scaling
with $N$. This is why we introduce an intermediate
$\delta$-block dynamics which keeps the number of 
Gaussian projections bounded. As mentioned earlier, in future work
we will establish that, as $\delta$ 
approaches zero, the $\delta$-block dynamics
approaches the dynamics of the sequential 
algorithm. For now this is numerically verified in \Cref{sec:numerics}. The 
introduction and the exact analysis of the 
$\delta$-block dynamics is the main technical novelty
of our paper, which we anticipate to be of use
beyond the scope of this work.

As alluded to earlier, the $\delta$-block dynamics is defined as follows. We split the spin set $[N]$ into $1/\delta$ 
consecutive blocks $B(1)=\{1,\ldots,\delta N\},
B(2)=\{\delta N+1,\ldots,2\delta N\},\ldots$,
where the integrality of $1/\delta$ and $\delta N$
is assumed for convenience. We then repeatedly perform sequential scans through the blocks $B(1),\dots,B(1/\delta)$ in this order, updating all spins within each block simultaneously. Specifically, suppose we are in scan $t$ and we arrive at block $j \in [1/\delta]$ in configuration $\bfsigma^{t,j-1}\in \{\pm 1\}^N$, having already updated blocks $B(1),\dots,B(j-1)$ within pass $t.$ We first compute the local fields $\bfh^t_i = (J\bfsigma^{t,j-1})_i$ for all $i\in B(j)$, and then we update all spins $\bfsigma^{t,j}_i$ for $i\in B(j)$ according to the corresponding local field $\bfh_i^t,$ leaving all other spins unchanged: $\bfsigma_i^{t,j}=\bfsigma_i^{t,j-1}$ for $i\notin B(j).$

We prove that for each $T,j,k$ the 
spins $\{\bfsigma_i^{s,j}\}_{s=0}^T$ and fields $\{\bfh^s_i\}_{s=1}^T$ are approximately i.i.d. across $i\in B(k)$ with a common joint
distribution which we describe. We detail
this distribution as an incremental 
stochastic process defined recursively over
$T,j,k$. Finally, fixing
any $x\in [0,1]$, we show that the distribution of this stochastic process evaluated at block $j=xN$ converges
to the solution of a system of differential equations 
evaluated at $x$ as 
$\delta\to 0$. 
Our proof technique is rather flexible and
accommodates broad classes of update rules,
including Glauber updates at any temperature.

\subsection{Notational conventions}
We now establish some notation conventions that will be used throughout the paper. First, we use boldface font to denote vectors or matrices with dimension growing with $N$, which we call \emph{high dimensional}. On the other hand, we will use non-boldface font to denote vectors or matrices with dimension that is constant in $N,$ which we call \emph{low dimensional}. These non-boldface quantities will be of two kinds. First, we will have quantities (either deterministic or random) which are entirely $N$-independent, like those that are the solution to our asymptotic equations (see \eqref{eq:h-joint-law-alone}-\eqref{eq:f-initial} below). We will refer to these as the \emph{Effective Dynamics}, since they simulate in low dimension the (high-dimensional) dynamics that we analyze, in the limit $N\to\infty$. Second, there will be quantities which are low-dimensional functions of boldface quantities, so they depend on $N$ despite having dimension that is constant in $N$. These second quantities will be denoted with a hat. For instance, below we will often use the $(T+1)\times (T+1)$ overlap matrix of spin vectors across times $t=0,\dots,T$, which will be denoted $\hat{C}.$ Finally, in \Cref{sec:finite-delta-dynamics}, in order to compare the limiting dynamics to the finite-$N$ dynamics, it will be useful for us to take a low-dimensional stochastic process and form high-dimensional vectors whose entries are i.i.d. draws from distributions determined by the low-dimensional process. We will denote these ``i.i.d. processes'' with a tilde, e.g. $\wt{\bfh}, \wt{\bfsigma}.$

In this paper, we will deal with matrices with different indexing conventions. Some matrices will have indexing of their rows and columns beginning with zero, others beginning with one, and yet others will have rows beginning with one and columns beginning with zero. We will be explicit about the appropriate indexing convention when each matrix is defined, but in all cases the convention is chosen to agree with the physical interpretation. For a matrix $M$ with zero-indexed rows and columns and $t\in \N,$ we let $M^{(t)} = \{M_{ss'}\}_{s,s'=0}^t$, and for a matrix with one-indexed rows and columns, we define $M^{(t)} = \{M_{ss'}\}_{s,s'=1}^t$. Even when the matrix $M$ is $(t+1)\times(t+1)$ (resp. $t\times t$), we will sometimes write $M^{(t)}$ to remind the reader of this fact. Moreover, for indices $s\leq s',t\leq t'$, we use $M_{s:s',t:t'}$ to refer to the block of $M$ formed by rows $s$ to $s'$ (inclusive) and columns $t$ to $t'$ (inclusive), and let $M_{s:s', t}:=M_{s:s', t:t}$ and $M_{s, t:t'}:=M_{s:s, t:t'}.$  Similarly, if we have a sequence of vectors or scalars $v^0,\dots, v^t$ or $v^1,\dots,v^t,$ we let $v^{(t)}$ denote the tuple $(v^{0}, \dots, v^t)$ or $(v^1, \dots, v^t)$, respectively.

\subsection{Paper organization}
The paper is organized as follows. In \Cref{sec:technical-overview} we present our result and some of its consequences. A formal statement of our main theorem is found in \Cref{sec:main-result}. In \Cref{sec:useful-lemmas} we gather several definitions and useful lemmas that are used throughout our proofs. \Cref{sec:finite-delta-dynamics}, \Cref{sec:delta-to-0-limit} and \Cref{sec:main-thm-proof} make up the proof of our main result. In particular, in \Cref{sec:finite-delta-dynamics}, we analyze
the aforementioned recursive Gaussian process
which we formally introduce in \Cref{sec:technical-overview} and which serves as a limiting approximation of the $\delta$-block dynamics. 
 Then, in \Cref{sec:delta-to-0-limit} and \Cref{sec:main-thm-proof}, we obtain the description of our dynamics as a limit of these modified dynamics. Finally, in \Cref{sec:numerics}, we numerically solve our equations, compare to simulations, and draw some conceptual conclusions.

\section{Main result}\label{sec:technical-overview}
In this section, we formally define the problem and describe some consequences of our result. For a formal statement of our main theorem in its full generality, see \Cref{sec:main-result}.

Our main result gives an asymptotic characterization of the dynamics on $\{\pm 1\}^N$ of a certain algorithm associated with the Hamiltonian given by \Cref{def:ham}. The algorithm we analyze is closely related to a very well known and widely used algorithm known as the \emph{systematic scan} (see, e.g., \cite{dyer2006systematic, dyer2008dobrushin}), which we describe first in order to contrast it with our algorithm.

The systematic scan is simplest in its \emph{greedy}, or \emph{zero-temperature}, version, which is the one Parisi \cite{Parisi-experimental} studied: given a number of passes $T$, we start from a uniformly-random configuration $\bfsigma^0 \in \{\pm 1\}^N,$ and we scan through all the indices $1,\dots,N$ from left to right (in a fixed order) $T$ times, sequentially flipping each bit whenever that results in an increase in energy (the value of $H$). In the general systematic scan, instead of deterministically flipping each bit if that results in an increase in energy, we flip it with a probability described as a function $p:\R \to [0,1]$ of the energy change that would result from the flip. If, for a given constant $\beta>0,$ the function $p=p_\beta$ is chosen as $p(x)=1/(1+e^{-2\beta x}),$ this is known as the \emph{systematic Glauber dynamics at inverse temperature $\beta$}, and in this case it can be checked that this algorithm gives a reversible Markov chain with stationary measure given by \eqref{eq:Gibbs-measure}.

Given a configuration $\bfsigma \in \{\pm 1\}^N,$ the change in energy from flipping a given spin $i$ is known, up to a sign, as the \emph{local field} at $i.$ Formally, given $i \in [N]$ and a configuration $\bfsigma \in \{\pm 1\}^N,$ we define the local field at $i$, denoted $\bfh_i$, as follows:
   \begin{equation}\label{eq:field}
       \bfh_i =\bfh_i(\bfsigma) = (J\bfsigma)_i.
   \end{equation}
Hence, the general systematic scan can be phrased in the following way which will be convenient for us below. First, we fix a function $c:\Omega \times \R \to \{\pm 1\}$, called an \emph{update rule}, where $\Omega$ is a probability space, and for $x\in \R,$ we sometimes use the notation $c(x)$ to denote the random variable $\omega \mapsto c(\omega, x).$ To avoid degenerate cases, we assume three things about $c$: (1) if $G$ is a standard Gaussian random variable, we have $0<\Pr[c(G)=1]<1,$ (2) we take $\Omega = [0,1]$, equipped with the Lebesgue measure, and (3) the boundary $\partial (c^{-1}(1))$ has Lebesgue measure zero in $[0,1]\times \R.$ Then, we start from a random $\bfsigma^0 \in \{0,1\}^N$, and we sequentially scan through all the spins from left to right. If we reach spin $i$ on configuration $\bfsigma$, we replace $\bfsigma_i \gets c(\bfh_i)$, where $\bfh_i=\bfh_i(\bfsigma),$ and where the randomness of $c$ is independent across different updates. If we take $c$ such that $\Pr[c(x)=1]=  p_\beta(x)$ for all $x,$ we recover the systematic Glauber dynamics.

The algorithm that we analyze in this work is a modification of the systematic scan to \emph{block updates.} To define it, we fix a constant $\delta \in [0,1]$ and divide $[N]$ into $1/\delta$ equally-sized ``blocks" $B(j)=\{i\in [N]: (j-1)\delta N < i\leq j \delta N\}$, $j=1,\dots, 1/\delta  $. Then, starting from a uniformly-random configuration $\bfsigma^0 \sim \Unif\{\pm 1\}^N$, we sequentially scan through the blocks from left to right. When we reach a given block on configuration $\bfsigma \in \{\pm 1\}^N$, we compute the local fields of all the spins within the block in parallel, and then we update all the spins in the block simultaneously according to these ``pre-computed'' local fields. This defines a sequence of vectors $\bfsigma^{t,j} \in \{\pm 1\}^N$ for $t\in \N$ and $j\in [1/\delta]$ which correspond to the state of the algorithm immediately after updating the $j$'th block in the $t$'th pass. These vectors are formally defined below. It will be convenient in the analysis for us to have access to the ``random seeds'' used in the updates to $c.$ Hence, we assume that $\{U_i^t:i\in [N], t\geq 0\}$ are i.i.d. $\Unif[0,1]$ random variables sampled in advance.
\begin{algorithm}[H]
\caption{Block Systematic Scan}
\label{alg:t-passes}
\SetAlgoLined
$\bfsigma^{0,0} \leftarrow \bfsigma^0 \sim \Unif\{\pm 1\}^N$

\For{$t \gets 1$ \KwTo $T$}{
  \For{$j \gets 1$ \KwTo $1/\delta$}{
    \For{$i \in B(j)$}{
      $\bfh^t_i \gets (J \bfsigma^{t,j-1})_i$
    }
    \For{$i \in [N]$}{
      $\bfsigma^{t,j}_i \leftarrow \begin{cases}
        c(U_i^t, \bfh^t_i) & i \in B(j)\\
        \bfsigma^{t,j-1}_i & \text{otherwise}
      \end{cases}$
    }
  }
  $\bfsigma^{t+1,0} \gets \bfsigma^{t,\frac{1}{\delta}}$ \tcp*[r]{initialization for pass $t+1$.}
}
\end{algorithm}
We note that, as the notation suggests, the vector $\bfh^t$ does not depend on $j.$ Indeed, each block $\{\bfh^t_i : i \in B(j)\}$ of $\bfh^t$ gathers the local fields of the sites in that block immediately before the block is updated. Hence, these blocks of values are computed once within each pass and do not evolve beyond that. Finally, for each $1\leq t\leq T,$ we let $\bfsigma^t := \bfsigma^{t,\frac{1}{\delta}}$ be the configuration at the end of pass $t.$

Our main result is an exact characterization of Algorithm~\ref{alg:t-passes} (in a sense specified below), in the double limit where first $N\to\infty$ and then $\delta \to 0$, as the unique solution to a system of explicit integro-difference equations. We conjecture that this agrees with the limiting dynamics of the systematic scan itself (without block updates). As mentioned in the introduction, we will address this conjecture in an upcoming work, as part of a more general framework to prove universality results for out of equilibrium discrete-update dynamics. For now, we verify this conjecture experimentally in \Cref{sec:numerics}.

We first describe what our results say about Algorithm~\ref{alg:t-passes}. For the purposes of this section, we will describe a weaker version of our main theorem (\Cref{thm:main}), but one that is easier to describe succinctly. Suppose we fix $x\in [0,1]$ and a sequence of indices $i_N \in [N]$ such that $i_N/N \to x \in [0,1]$ as $N\to\infty.$ Let $\nu_{N, i_N, \delta, T}$ be the joint law of the random variables $\bfsigma^t_{i_N},t=0,\dots,T$ and $\bfh_{i_N}^t, t=1,\dots,T$ from Algorithm~\ref{alg:t-passes}, which we call the \emph{spin-field history at} $i_N$. A consequence of \Cref{thm:main} is that the weak limit $\nu_{x,T}:= \lim_{\delta \to 0}\lim_{N\to\infty}\nu_{N,i_N, \delta, T}$ exists. Moreover, the spin-field histories become ``weakly decoupled'' across $i=1,\dots,N$, in the following sense. If we let 
\[
\widehat{\nu}_{N,\delta, T} = \frac{1}{N}\sum_{i=1}^N \delta_{(\{\bfsigma_i^t\}_{t=0}^T,\{\bfh_i^t\}_{t=1}^T)}
\]
be the empirical measure of spin-field histories, then 
\[
\lim_{\delta \to 0}\lim_{N\to\infty}\widehat{\nu}_{N,\delta,T} = \int_0^1 \nu_{x,T}dx,
\]
where the convergence is weakly in probability. More precisely, for every (pseudo-Lipschitz, see \Cref{def:pseudo-lipshitz-fucntion}) test function $\varphi,$ we have that $\int \varphi d\widehat{\nu}_{N,\delta,T}$ converges in probability to a deterministic limit which depends on $\delta$, and as $\delta \to 0$, this limit converges in $\R$ to $\int_0^1 \int \varphi d\nu_{x,T}dx.$ Essentially every macroscopic observable of interest is a function of the empirical measure over spin-field histories. Examples include the energy, the magnetization, the two-point overlap (the normalized inner product between the configuration of the dynamics at two different times), etc. Hence, our result shows that all of these quantities concentrate, and their deterministic asymptotic values can be computed from $\nu_{x,T}.$ 

Most importantly, our result gives an explicit description of this limit law $\nu_{x,T}$, which we now describe by specifying a procedure to obtain a sample $\{\sigma_x^t\}_{t=0}^T, \{h_x^t\}_{t=1}^T$ from $\nu_{x,T}.$ To do this, we will first describe a generic procedure which takes as input a PSD matrix $K^{(T)} \in \R^{T\times T}$, with indexing starting from one, and a sequence of vectors $v^t \in \R^t, t=1,\dots,T,$ and outputs a joint sample $\{\sigma^t\}_{t=0}^T, \{h^t\}_{t=1}^T$. Then, the sample $\{\sigma_x^t\}_{t=0}^T,\{h_x^t\}_{t=1}^T$ is obtained by choosing $K^{(T)} = \Sigma^{(T)}(x)$ and $v^t = f^t(x), t=1,\dots,T$, where $\Sigma^{(T)}(x)$ and $f^t(x),t=1,\dots,T$ are specified below by a system of ODE.

We now describe the sampling procedure given $K^{(T)}$ and $v^{(T)}:= (v^1,\dots, v^T).$  First, we sample $\sigma^0 \sim \Unif\{\pm 1\}.$ Then, independently of $\sigma^0$, we sample $T$ jointly Gaussian random variables $G^1,\dots, G^T$ of zero mean and covariance matrix $K^{(T)}$. Finally, we recursively define, for $t=1,\dots,T,$
\begin{equation}\label{eq:h-joint-law-alone}
\begin{cases}
        h^t = G^{t} + \brac{v^t, \sigma^{(t-1)}}\\
        \sigma^t = c(h^{t}),
\end{cases}
\end{equation}
where for $s\in \N,$ recall that $\sigma^{(s)} = (\sigma^0,\dots, \sigma^s) \in \{\pm 1\}^{s+1},$ and where the randomness of $c$ is independent across $t=1,\dots,T.$ Below, we use the notation $\E_{K^{(T)}, v^{(T)}}$ to denote the expectation operator with respect to the joint law given by \eqref{eq:h-joint-law-alone}.

To give a full description of the limiting law $\nu_{x,T},$ it remains to specify the quantities $\Sigma^{(T)}(x)\in \R^{T\times T}$ and $f^t(x) \in \R^t$ for $t=1,\dots, T.$ These, viewed as functions of $x\in [0,1],$ are the unique solution to a system of equations which we now specify. 
We first state them as a sequence of ordinary differential equations, since in this form it becomes evident that there is a unique solution. Later we shall write them in the equivalent integro-difference form, which is closer to the physical interpretation. To describe the equations, we first need to define three matrix-valued functions parameterized by the time horizon $T\in \N,$ which we suppress in the notation to avoid clutter. These are $\calA, \calC : \R^{T\times T} \times \left(\times_{s=1}^T \R^s\right) \to \R^{(T+1)\times (T+1)}$ and $\calR : \R^{T\times T} \times \left(\times_{s=1}^T \R^s\right) \to \R^{T\times (T+1)}$. We will index the rows and columns of the output of $\calA$ and $\calC$ from 0 to $T$ to keep it consistent with the physical interpretation as the pass number, with the zeroth pass representing the initial condition. Similarly, the rows of the output of $\calR$ will be indexed from $1$ to $T,$ while the columns will be indexed from $0$ to $T.$ We define $\calC$ and $\calR$ as follows:
\begin{align}
    \calC_{s,t}(K^{(T)}, v^{(T)}) &= \E_{K^{(T)}, v^{(T)}}[\sigma^s \sigma^t] &&0\leq s,t\leq T\label{eq:calC}\\
    \calR_{s,t}(K^{(T)}, v^{(T)}) &= \E_{K^{(T)}, v^{(T)}}[G^s \sigma^t]  &&1\leq s\leq T,0\leq t\leq T.\label{eq:calR}
\end{align}
Finally, the value of $\calA(K^{(T)}, v^{(T)})$ is defined recursively as the unique upper-triangular matrix $A \in \R^{(T+1)\times (T+1)}$ with zeros on the diagonal which satisfies the recursion
\begin{align}\label{eq:calA}
    A_{0:t-1,t} &= \left(A^{(t-1)} - (K^{(t)})^{-1} R_{1:t, 0:t-1}\right)(C^{(t-1)})^{-1}C_{0:t-1,t} + (K^{(t)})^{-1}R_{1:t, t},
\end{align}
where $C = \calC(K^{(T)}, v^{(T)})$ and  $R = \calR(K^{(T)}, v^{(T)})$. Above, we are assuming that the input to $\calA$ is such that $K^{(t)}$ and $C^{(t-1)}$ are invertible for $t=1,\dots,T$. If this is not the case, we define $\calA(K^{(T)}, v^{(T)})$ to be the all-zero matrix, although as we show in \Cref{prop:finite-delta-regularity}, this case never arises in the solution trajectories.

Now, the sequence of ODEs which specify $\Sigma^{(T)}(x)$ and $f^t(x),t=1,\dots,T$ is as follows
\begin{align}
    \frac{d}{dx}\Sigma^{(T)}(x) &= \calC_{1:T, 1:T}(\Sigma^{(T)}(x), f^{(T)}(x))  - \calC_{0:T-1, :T-1}(\Sigma^{(T)}(x), f^{(T)}(x)) \label{eq:Sigma-ODE}\\
    \frac{d}{dx}f^t(x) &= \calA_{0:t-1, t}(\Sigma^{(T)}(x), f^{(T)}(x))  - \begin{bmatrix}
        0 \\
        \calA_{0:t-2, t-1}(\Sigma^{(T)}(x), f^{(T)}(x))
    \end{bmatrix}  && t=1,\dots, T,\label{eq:f-ODE}
\end{align}
with the initial condition 
\begin{align}
    \Sigma^{(T)}(0) &= \int_0^1 \calC_{0:T-1,0:T-1}(\Sigma^{(T-1)}(y), f^{(T-1)}(y))dy \label{eq:sigma-initial}\\
    f^{t}(0) &= \int_0^1 \begin{bmatrix}
        0 \\
        \calA_{0:t-2, t-1}(\Sigma^{(T-1)}(y), f^{(T-1)}(y))
    \end{bmatrix}dy && t=1,\dots, T.\label{eq:f-initial}
\end{align}
This is a \emph{sequence} of ODEs because it must be solved iteratively in $T=1,2,\dots$, at each step feeding the solution of the previous step as the initial condition. The base case of the iteration is at $T=1,$ where the initial condition is $\Sigma^{(1)}(0)=[1]$ and $f^1(0)=0$ (we use the convention $\calA_{0:-1,0}:=0$, and $\calC_{0,0}=1$ is true by definition \eqref{eq:calC}).

The same equations can also be written in integro-difference form. This formulation is closer to the physics, but makes it less transparent that there is a unique solution. Below, for the equations to be closed, one has to interpret $\{\sigma_x^t\}_{t=0}^T,\{h_x^t\}_{t=1}^T$ as being sampled, for each $x,$ from \eqref{eq:h-joint-law-alone} with $K^{(T)}=\Sigma^{(T)}(x)$ and $v^t = f^t(x)$ for $t=1,\dots,T.$ Moreover, we write $C(x) := \calC(\Sigma^{(T)}(x), f^{(T)}(x))$, $R(x) := \calR(\Sigma^{(T)}(x), f^{(T)}(x))$ and $A(x) := \calA(\Sigma^{(T)}(x), f^{(T)}(x))$ to make the equations more compact. Then we have
\begin{align}
        \Sigma^{(T)}(x) &= \int_0^x C_{1:T, 1:T}(y)dy + \int_x^1 C_{0:T-1, 0:T-1}(y)dy \label{eq:Sigma}\\
    f^t(x) &= \int_0^x A_{0:t-1, t}(y)dy + \int_x^1 \begin{bmatrix}
        0 \\
        A_{0:t-2, t-1}(y)
    \end{bmatrix}dy &&t=1,\dots, T\label{eq:f}\\ 
    C_{st}(x) &= \E[\sigma_x^s\sigma_x^t] &&s,t=0,\dots, T\label{eq:C-eqn-asymp}\\
    R_{st}(x) &= \E[G_x^s\sigma_x^t] &&s,t=0,\dots, T\label{eq:R-eqn-asymp}\\
    A_{0:t-1, t}(y) &= \left( A^{(t-1)}(y) - \Sigma^{(t)}(y)^{-1} R_{1:t,0:t-1}(y) \right) C^{(t-1)}(y)^{-1} C_{0:t-1,t}(y)\\
    &\qquad + \Sigma^{(t)}(y)^{-1} R_{1:t, t}(y) \nonumber&& t=1,\dots,T \\
     A_{st}(y)&=0 && \text{if } s\geq t \text{ or $t=0$}.\label{eq:A-upper-diag}
\end{align}

Two remarks are in order. First, note that if the $(t,x)$-pairs, for $(t,x)\in \N \times \R,$ are ordered lexicographically, then the equations are causal with respect to this ordering. Second, the functions $C$ and $R$ above are exactly what is known in the literature on out of equilibrium dynamics of disordered systems as the \emph{correlation} and \emph{response} functions. While we chose to write our equations in the dual form with $\Sigma^{(T)}$ and $f^t, t=1,\dots,T$ as the variables (what appears on the left-hand side of the ODE formulation \eqref{eq:Sigma-ODE}, \eqref{eq:f-ODE}), one can equivalently write them with $C$ and $R$ as the variables. Then this would be exactly in agreement with the choice of order parameters of Cugliandolo-Kurchan in the spherical setting.

In summary, equations \eqref{eq:h-joint-law-alone}-\eqref{eq:f-initial}, or equivalently \eqref{eq:h-joint-law-alone} and \eqref{eq:Sigma}-\eqref{eq:A-upper-diag}, completely specify the limiting law $\nu_{x,T}$ of the spin-field histories, and by solving these equations, we can compute the limiting value of any observable which is a function of the empirical measure of spin-field histories. In \Cref{sec:numerics} we numerically solve these equations, compute some observables of interest, and compare to experimental data obtained through simulations of the systematic scan dynamics.





\subsection{Main Theorem}\label{sec:main-result}
We now give the formal statement of our main theorem. To describe the weak convergence in our main result, we introduce the following class of test functions.
\begin{defn}[Pseudo-Lipschitz function]\label{def:pseudo-lipshitz-fucntion}
We say that a function $f \colon \mathbb{R}^m \to \mathbb{R}$ is \emph{pseudo-Lipschitz of order $k$} if there exists a constant $L$ such that for any $x,y \in \mathbb{R}^m$, we have
\begin{align*}
    \lvert{f(x) - f(y)\rvert} \leq 
    L (
        1 
        + \|{x}\|^{k} 
        + \|{y}\|^{k}
    )
    \|{x - y}\| \,.
\end{align*}
\end{defn}
We now state our main result.
    
\begin{thm}\label{thm:main}
Fix $x \in [0,1]$ and $T>0.$
For each $\delta \in (0,1)$, let $j=\lfloor \frac{x}{\delta}\rfloor \in \mathbb{N}$. Let 
$\bfsigma^t=\bfsigma^t(\delta), t=0,\dots,T$ and $\bfh^t=\bfh^t(\delta), t=1,\dots,T$ be the vectors produced by Algorithm~\ref{alg:t-passes}. Then, for all $\varphi:\{\pm 1\}^{T+1} \times \mathbb{R}^T \rightarrow \mathbb{R}$, pseudo-Lipschitz of finite order in the second argument, we have:
\begin{align*}
 \lim_{\delta \to 0}\lim_{N\to\infty}\frac{1}{\delta N} \sum_{i= j \delta N}^{(j+1)\delta N -1}\varphi(\bfsigma^{(T)}_i, \bfh^{(T)}_i)    = \E[\varphi(\sigma_x^{(T)},h_x^{(T)})],
\end{align*}
    where the first limit is in probability, the second is in $\R$, and the expectation on the right hand side is with respect to the distribution explicitly described by the equations \eqref{eq:h-joint-law-alone}-\eqref{eq:f-initial} above.
\end{thm}
We obtain as an immediate consequence the convergence of the full empirical measure of spin-field histories.
\begin{cor}\label{cor:main}
    Let $\{\bfsigma^t\}_{t=0}^T, \{\bfh^t\}_{t=1}^T$ be as in \Cref{thm:main}. Then we have 
    \begin{align*}
     \lim_{\delta \to 0}\lim_{N\to\infty}\frac{1}{N} \sum_{i= 1}^N \varphi(\bfsigma^{(T)}_i, \bfh^{(T)}_i)    = \int_0^1 \E[\varphi(\sigma_x^{(T)},h_x^{(T)})]dx,
    \end{align*}
    where the first limit is in probability, the second in $\R$, and the expectations on the right hand side are with respect to the distribution explicitly described by the equations \eqref{eq:h-joint-law-alone}-\eqref{eq:f-initial} above.
\end{cor}
\begin{proof}
    The corollary follows from \Cref{thm:main} and the Bounded Convergence Theorem, since
    \[
    \sup_{x\in [0,1]}\E[\varphi(\sigma_x^{(T)},h_x^{(T)})]<\infty
    \]
    as $\sup_{t\in [T]} \sup_{x\in [0,1]}\norm{f^t(x)}_2 < \infty$ by \Cref{prop:finite-delta-regularity} and, for all $x\in [0,1],t\in [T],$ we have $\Var(G^t|G^{(t-1)})\leq \Sigma^{(T)}_{t,t}(x)=1$, so $\{h_x^t:x\in [0,1], t\in [T]\}$ are uniformly subgaussian.
\end{proof}

\subsection{Energy achieved by the dynamics}\label{sec:computing-the-energy}

As an example application of \Cref{thm:main}, we show how it can be used to compute the energy achieved by the dynamics of Algorithm~\ref{alg:t-passes}. We then specialize this application to the simple case $T=1$, where the formulas become more explicit. We refer the reader to \Cref{sec:numerics} for plots of numerical computations of the energy as well as other quantities of interest. 

Let $\{\bfsigma^t\}_{t=0}^T, \{\bfh^t\}_{t=1}^T$ be as in \Cref{thm:main}. Note that the energy change induced by updating spin $i$ in pass $t$ is exactly given by $(\bfsigma^t_i - \bfsigma^{t-1}_i) \bfh_i^t.$ Hence, we have
\begin{align*}
    \frac{1}{N}(H(\bfsigma^T) - H(\bfsigma^0))&= \frac{1}{N}\sum_{i=1}^N\sum_{t=1}^T(\bfsigma^t_i - \bfsigma^{t-1}_i) \bfh_i^t.
\end{align*}
Choosing $\varphi(\{u^t\}_{t=0}^T,\{v^t\}_{t=1}^T) = \sum_{t=1}^T(u^t - u^{t-1})v^t$ and applying \Cref{cor:main}, we obtain that the energy attained by Algorithm~\ref{alg:t-passes} in the double limit $N\to \infty$ and $\delta\to 0$ is exactly
\begin{align*}
    \int_0^1 \E\left[\sum_{t=1}^T (\sigma_x^t - \sigma_x^{t-1})h_x^t\right]dx &= \int_0^1\left(\sum_{t=1}^T R_{t,t}(x) - R_{t,t-1}(x) + \brac{f^t(x), C_{0:t-1,t}(x) - C_{0:t-1,t-1}(x)}\right) dx 
\end{align*}

Next we assume $T=1.$ Denote by $c$ the Glauber update rule, so that $\Pr[c(x)=1] = 1/(1+e^{-2\beta x}).$ For $T=1$, we have
\[
\Sigma_{1,1}(x) = \int_0^x C_{1,1}(x)dx + \int_x^1 C_{0,0}(x)dx = 1,
\]
where we have used that $C_{s,s}(x) = \E (\sigma_x^s)^2=1$ for all $s$ and $x.$ Moreover, since $A$ is upper triangular with zeros in the diagonal, we have $A_{0,0}=0$ and thus
\[
A_{0,1}(y) = -R_{1,0}(y)C_{0,1}(y) + R_{1,1}(y) = R_{1,1}(y),
\]
where we have used that $R_{s,0}(x)=0$ for all $s \geq 1$ since $\sigma^0_x$ is independent of $G_x^s.$ Now, denoting $f(x):=f^1(x)$ and $G:=G_x^1\sim N(0,1)$, we have
\begin{align*}
    R_{1,1}(x) &= \E \sigma^1_x G \\
    &= \E \tanh\left(\beta\left(G + f(x)\sigma^0_x\right)\right) G \\
    &= \beta \E \tanh'\left(\beta\left(G + f(x)\sigma^0_x\right)\right)\\
    &= \beta \E \tanh'\left(\beta\left(G + f(x)\right)\right)\\
    &= \beta \Big(1-\E \tanh^2\left(\beta\left(G + f(x)\right)\right)\Big),
\end{align*}
where we have used that $\tanh'$ is even and Stein's Lemma (see \Cref{lem:Stein} below). Hence, by \eqref{eq:f-ODE}, $f$ is the solution to the differential equation
\begin{align*}
    f'(x) &= \beta \Big(1-\E \tanh^2\left(\beta\left(G + f(x)\right)\right)\Big)
\end{align*}
with the initial condition $f(0)=0.$ Finally, we have 
\[
C_{0,1}(x) = \E \tanh\left(\beta\left(f(x) + G \right)\right).
\]
From this, we can compute the energy $E_1$ at the end of the first pass:
\begin{align*}
    E_1 &=  \int_0^1 \Big( R_{1,1}(x) - R_{1,0}(x)+ f(x)(C_{0,1}(x) - C_{0,0}(x))\Big) dy \\
    &= \int_0^1 \Big( R_{1,1}(x)+ f(x)\Big(\E \tanh\left(\beta\left(f(x) + G \right)\right) - 1\Big)\Big) dy \\
    &= f(1) + \int_0^1 f(x)\Big(\E \tanh\left(\beta\left(f(x) + G \right)\right) - 1\Big)dx,
\end{align*}
where we have used that $R_{1,1}(x) = f'(x).$

\subsection{Organization of the proof}\label{sec:proof-outline}

We organize the proof of \Cref{thm:main} as follows. In \Cref{sec:finite-delta-dynamics}, we analyze Algorithm~\ref{alg:t-passes} for a fixed value of $\delta > 0.$ We show, via the Gaussian conditioning formalism pioneered in \cite{bolthausen2014iterative, bayati2011dynamics}, that the dynamics of Algorithm~\ref{alg:t-passes} converge as $N\to\infty$ to the solution of a system of difference equations which are a $\delta>0$ analogue of \eqref{eq:Sigma}-\eqref{eq:A-upper-diag}. In \Cref{sec:delta-to-0-limit}, we show that the solutions to these $\delta>0$ difference equations converge, as $\delta \to 0$, to the solution of \eqref{eq:Sigma}-\eqref{eq:A-upper-diag}. Along the way, we also prove that \eqref{eq:h-joint-law-alone}-\eqref{eq:f-initial} has a global unique solution in $[0,1]$. Finally, in \Cref{sec:main-thm-proof} we conclude the proof of \Cref{thm:main}.

\section{Preliminaries}\label{sec:useful-lemmas}

Let $\bfsigma^{t,j}, \bfh^t\in \R^N$ denote the spins, fields in Algorithm \ref{alg:t-passes}. To simplify our subsequent analysis, it will be convenient to  isolate spins and contributions to the fields arising from different blocks in $\bfsigma$. To this end, we introduce the following additional notation. First, for a block $j\in [1/\delta],$ we let $D^j\in \R^{N\times N}$ be the diagonal matrix which projects onto the $j$'th block, i.e., 
\[
(D^j)_{k\ell} = \1\{k=\ell \in B(j)\}.\]
Then we define $\bfsigma^{t\vert j}:= D^j\bfsigma^t.$ Thus, $\bfsigma^{t\vert j}_i \in \mathbb{R}^N$ denotes the spin configuration at pass $t$ after the update to block $j$, but with the values of all blocks except the $j$'th  set to zero. We additionally denote by $\bfh^{t,k \rightarrow j}:= D^j J \bfsigma^{t|k}$ the field contributed by block $k$ at pass $t$ to the field at block $j$. Note that we have, for all $i \in B(j),$
\begin{equation}\label{eq:field_decomp}
    \bfh^t_i = \sum_{k<j} \bfh^{t,k \rightarrow j}_i  + \sum_{k\geq j} \bfh^{t-1,k \rightarrow j}_i.
\end{equation}
Finally, using the notation
\begin{align*}
    \bfsigma^{(t)|j}&:= (\bfsigma^{0|j}, \dots, \bfsigma^{t|j}) \in \{\pm 1\}^{N\times(t+1)} \\
    \bfh^{(t),k\to j}&:= (\bfh^{0,k\to j}, \dots, \bfh^{t,k\to j}) \in \R^{N\times(t+1)},
\end{align*}
we define the overlap matrices
\begin{align}
    \hat{C}^{(t),j}&= \frac{1}{\delta N} (\bfsigma^{(t)\vert j})^\top \bfsigma^{(t)\vert j} \in \mathbb{R}^{(t+1)\times (t+1)}\label{eq:hatc} \\
    \hat{Q}^{(t),k \rightarrow j} &= \frac{1}{\delta^{2} N} (\bfh^{(t),k \rightarrow j})^\top \bfsigma^{(t)\vert j}\in \mathbb{R}^{(t+1)\times (t+1)}.\label{eq:hatq}
\end{align}

Our proof heavily relies on the use of projection matrices, the notation for which we specify next.
For any $d \in \mathbb{N}$ and a linear subspace $U \subseteq \mathbb{R}^d$ of dimension $k \leq d$ we denote by $P_U$ the projection linear operator $\mathbb{R}^d \rightarrow \mathbb{R}^d$ onto $U.$ For any set of vectors $\bfu^1,\cdots, \bfu^p$ such that $\operatorname{span}(\bfu^1,\cdots, \bfu^p) = U$, $P_U$ admits the following explicit representation:
\begin{align}\label{eq:projection-pseudoinverse}
    P_U = \bfu^{(p)}((\bfu^{(p)})^\top \bfu^{(p)})^\dagger (\bfu^{(p)})^\top,
\end{align}
where $\bfu^{(p)}:=(\bfu^1,\dots,\bfu^p)\in \R^{N\times p}$ and $(\cdot )^\dagger$ denotes the Moore-Penrose pseudo-inverse. When $(\bfu^{(p)})^\top \bfu^{(p)}$ is full-rank, we have
\begin{equation}\label{eq:pudef}
    P_U = \bfu^{(p)}((\bfu^{(p)})^\top \bfu^{(p)})^{-1} (\bfu^{(p)})^\top.
\end{equation}
We now introduce a notion of asymptotic equivalence which will be used throughout our proofs.
\begin{defn}
    For $p\geq 1$ a constant and two random matrices $X,Y \in \mathbb{R}^{N\times p}$, we use the notation
    \begin{align*}
        X \stackrel{P}{\simeq} Y
    \end{align*}
    to denote
    \begin{align*}
        \frac{1}{\sqrt{N}}\norm{X-Y}_F \xrightarrow[N \rightarrow \infty]{P} 0,
    \end{align*}
    where $\norm{\cdot}_F$ denotes the Frobenious norm and $\xrightarrow[N \rightarrow \infty]{P}$ denotes convergence in probability.
\end{defn}
The above notion is closed under multiplication and addition, in the following sense.
\begin{lem}\label{lem:add_mult}
    For any $X,Y \in \mathbb{R}^{N \times p}$ with $X \stackrel{P}{\simeq} Y$ and $u,v \in \mathbb{R}^p$ with $\norm{u- v} \xrightarrow[N \rightarrow \infty]{P} 0$, we have
    \begin{align*}
        Xu \stackrel{P}{\simeq} Yv.
    \end{align*}
\end{lem}
\begin{proof}
The result is a direct consequence of Slutsky's theorem and the following inequality
  \begin{align*}
       \norm{Xu-Yv}_2 \leq  \norm{X-Y}_F\norm{u-v}_2.
  \end{align*}  
\end{proof}

We now list several additional useful lemmas which will be used throughout.
\begin{lem}[Existence of tempered distributional derivatives]\label{lem:distributional-derivative}
    Let $f:\R\to \R$ a Lebesgue measurable function. Suppose $f$ is of at most polynomial growth: there exists $k\in \N$  and $C>0$ such that $|f(x)|\leq C(1+|x|)^k$ for a.e. $x \in \R.$ Let $\calS$ be the Schwartz space of all functions $\phi \in C^\infty(\R)$ such that, for all $k,m\in \N,$ the function $x^m\frac{d^k}{dx^k}\phi(x)$ is bounded. Then the functional 
    \[
    T(\phi) = -\int f(x)\phi'(x)dx.
    \]
    is continuous as a map $T:\calS \to \R$ (in the suitable topology). Below, we abuse notation and formally write
    \[
    \int f'(x)\phi(x)dx := T(\phi).
    \]
\end{lem}
We refer the reader to standard textbooks for a proof as well as general discussion of tempered distributions, e.g. \cite[Section 9]{folland1999real}. We now state and prove a version of Stein's Lemma for distributional derivatives suitable for our setting.
\begin{lem}[Stein's Lemma]\label{lem:Stein}
    Let $f:\R^m\to\R^n$ be a Lebesgue measurable function of at most polynomial growth in the sense above. We interpret the output of $f$ as a row vector, and let $Z=(Z_1,\dots,Z_m)$ be a centered Gaussian row vector with non-singular covariance $K \in \R^{m\times m},$ so that $Z^\top f(Z) \in \R^{m\times n}.$ Then
    \begin{align*}
        \E[Z^\top f(Z)] &= K \E[\partial_Z f(Z)],
    \end{align*}
    where $\partial_Z f(Z)$ denotes the (distributional) Jacobian $(\partial_Z f(Z))_{ij} = \partial_{Z_i} f_j(Z)$.
\end{lem}
\begin{proof}
    By definition of the distributional derivative, we have
    \begin{align*}
        \E Z^\top f(Z) &= \int \frac{1}{(2\pi)^{m/2}|K|^{1/2}} \exp\left(-\frac{1}{2}x K^{-1}x^\top\right) x^\top f(x) dx \\
        &= K \int \frac{1}{(2\pi)^{m/2}|K|^{1/2}} \exp\left(-\frac{1}{2}x K^{-1}x^\top\right) K^{-1}x^\top  f(x) dx \\
        &= K \int \frac{1}{(2\pi)^{m/2}|K|^{1/2}} \exp\left(-\frac{1}{2}x K^{-1}x^\top\right) \partial_x f(x) dx \\
        &= K \E[\partial_Z f(Z)],
    \end{align*}
   where in the third step we use the fact that $\partial_x \exp(-\frac{1}{2}x K^{-1}x^\top) = \exp(-\frac{1}{2}x K^{-1}x^\top)K^{-1}x^\top$ and then integration by parts, which works for distributional derivatives by definition.
\end{proof}
\begin{lem}[Vanishing projections on GOE matrices] 
\label{lem:proj_gauss}
Suppose that $W \sim \text{GOE}(N)$ and let $p,q \in \mathbb{N}$ be independent of $N$. Let $U,V \in \mathbb{R}^{N \times p}, \mathbb{R}^{N \times q}$ denote any sequence of matrices independent of $W$ such that $\E \norm{U}_F, \E\norm{V}_F =o(N^{3/4})$ as $N \rightarrow \infty$. Then
\begin{align*}
    \frac{1}{N}\norm{U^\top W V}_F \xrightarrow[N \rightarrow \infty]{P} 0.
\end{align*}
\end{lem}
\begin{proof}
    This follows from the fact that 
    \begin{align*}
        \E \left(\frac{1}{N}\norm{U^TWV}_F\right)^2 &= \frac{1}{N^3} \norm{U}_F^2 \norm{V}_F^2.
    \end{align*}
\end{proof}

\begin{lem}[Law of projection on GOE]\label{lem:GOE_proj_law}
Suppose that $W \sim \text{GOE}(N)$. Then, for any deterministic vector $\bfu  \in \mathbb{R}^{N}$, we have
\begin{align*}
    W\bfu \sim \mathcal{N}(\mathbf{0}, (\norm{\bfu}^2 I+\bfu\bfu^\top)/N)
\end{align*}
\end{lem}
\begin{proof}
    The fact that $\E W \bfu$ is immediate. Moreover, we have
    \begin{align*}
        N\cdot \E \left( (W\bfu)(W\bfu)^\top\right)_{ij} &= N\cdot\E \sum_{k,\ell}W_{ik}\bfu_k\bfu_\ell W_{j\ell} \\
        &= \delta_{ij} \left(2\bfu_i^2 + \sum_{k\neq i}\bfu_k^2\right) + (1-\delta_{ij}) \bfu_j\bfu_i \\
        &= (\norm{\bfu}^2I + \bfu\bfu^\top)_{ij},
    \end{align*}
    as desired.
\end{proof}

\begin{lem}[Gaussian conditioning]\label{lem:GOE-ind}
Suppose that $W \sim \text{GOE}(N)$. Let $\bfu^1, \cdots, \bfu^p \in \mathbb{R}^N$ denote a fixed sequence of vectors, and let $V \in \mathbb{R}^{N \times p}$ denote a matrix such that $V^\top\bfu^i = 0$ for all $ i \in [p]$. Then, the random variables $V^\top W V$ and $\{W\bfu^i\}_{i \in [p]}$ are independent.
\end{lem}
\begin{proof}
    We have
    \begin{align*}
        \E (V^\top W V)_{k\ell}(W\bfu^i)_j &= \E \sum_{srt}V_{sk}W_{sr}V_{r\ell}W_{jt}\bfu_t^i \\
        &= \frac{1}{N}\sum_r V_{jk}V_{r\ell}\bfu^i_r + \frac{1}{N}\sum_s V_{sk}V_{j\ell}\bfu_s^i \\
        &= \frac{V_{jk}}{N} (V^\top \bfu^i)_\ell + \frac{V_{j\ell}}{N} (V^\top \bfu^i)_k \\
        &= 0,
    \end{align*}
    as desired.
\end{proof}

\begin{lem}\label{lem:p-conv-empirical-conv}
    Let $X,Y \in \mathbb{R}^{N \times k}$ be two (sequences of) random variables with $X \stackrel{P}{\simeq} Y$ and such that for every $k>0$ there exists a constant $B_k>0$ such that
    \[
    \lim_{N\to \infty}\Pr\left(\frac{1}{N}\sum_{i=1}^N\norm{X_i}^k \geq B_k\right) = 0
    \] and
    \[
    \lim_{N\to \infty}\Pr\left(\frac{1}{N}\sum_{i=1}^N\norm{Y_i}^k \geq B_k\right) = 0.
    \]
    Let $\widehat{\nu}_X, \widehat{\nu}_Y$ be the measures in $\R^k$ given by  $\widehat{\nu}_X= \frac{1}{N}\sum_{i=1}^N \delta_{X_i}$ and $\widehat{\nu}_Y= \frac{1}{N}\sum_{i=1}^N \delta_{Y_i}$. Then, for every pseudo-Lipschitz function $\varphi,$ we have 
    \[
    \left|\int \varphi d\widehat{\nu}_X - \int \varphi d\widehat{\nu}_Y \right| \xrightarrow[N \rightarrow \infty]{P}  0.
    \]
\end{lem}
\begin{proof}
    We have 
    \begin{align*}
        \left|\int \varphi d\widehat{\nu}_X - \int \varphi d\widehat{\nu}_Y \right| &\leq \frac{1}{N} \sum_{i=1}^N| \varphi(X_i) - \varphi(Y_i)| \\
        &\leq \frac{1}{N} \sum_{i=1}^N L(1+ \norm{X_i}^k + \norm{Y_i}^k)\norm{X_i-Y_i}\\
        &\leq \sqrt{\frac{L^2}{N} \sum_{i=1}^N (1+ \norm{X_i}^k +\norm{Y_i}^k)^2} \sqrt{\frac{1}{N} \sum_{i=1}^N\norm{X_i-Y_i}^2} \\
        &\leq \sqrt{\frac{3L^2}{N} \sum_{i=1}^N 1+ \norm{X_i}^{2k} +\norm{Y_i}^{2k}}\sqrt{\frac{1}{N} \sum_{i=1}^N\norm{X_i-Y_i}^2}.
    \end{align*}
    Hence, for every $\eps>0,$ we have
    \begin{align*}
        \Pr\left( \left|\int \varphi d\widehat{\nu}_X - \int \varphi d\widehat{\nu}_Y \right|\geq \eps\right) &\leq \Pr\left(\frac{1}{N}\sum_{i=1}^N\norm{X_i}^{2k} \geq B_{2k}\right) + \Pr\left(\frac{1}{N}\sum_{i=1}^N\norm{Y_i}^{2k} \geq B_{2k}\right) \\
        &\qquad + \Pr\left(\frac{1}{N} \sum_{i=1}^N\norm{X_i-Y_i}^2 \geq \frac{ \eps^2}{3L^2(1 + 2B_{2k})}\right).
    \end{align*}    
    Each of the three terms vanish by assumption, concluding the proof.
\end{proof}
Finally, we prove a technical result that will be useful when we construct couplings between solutions to our equations and the dynamics of Algorithm~\ref{alg:t-passes}. It essentially says that, to construct such a coupling, it suffices to couple the fields, since the spins become automatically coupled as a consequence.
\begin{lem}\label{lem:spins}
    Let $c:[0,1] \times \R \to \{\pm 1\}$ be an update rule such that the boundary $D = \partial (c^{-1}(+1))$ has Lebesgue measure zero. Let $(\widetilde{\bfh}, \widetilde{\bfsigma}) \in (\R^{T} \times \{\pm 1\}^{T+1})^N$ have rows sampled i.i.d. from \Cref{eq:h-joint-law-alone} with some $v^{(T)}$ and a non-singular $K^{(T)}.$ Recall that $(\bfh, \bfsigma) \in (\R^{T} \times \{\pm 1\}^{T+1})^N$ is the process outputted by Algorithm~\ref{alg:t-passes}. Assume the following:
    \begin{itemize}
        \item The randomness of $c$ has been coupled between $(\widetilde{\bfh}, \widetilde{\bfsigma})$ and $(\bfh, \bfsigma)$ as follows. There exist i.i.d. $\Unif[0,1]$ random variables $\{U_i^t: i\in [N], t\in [T]\}$ such that for all $i\in [N]$ and $t \in [T],$ we have $\widetilde{\bfsigma}_i^t = c(U_i^t,\widetilde{\bfh}_i^t)$ and $\bfsigma_i^t = c(U_i^t,\bfh_i^t)$, where $U_i^t$ is independent of $\widetilde{\bfh}_i^t$ and $\bfh_i^t$.
        \item There exists a block $j \in [1/\delta]$ such that $\bfsigma^{0,j} = \tilde{\bfsigma}^{0,j}$ and $\frac{1}{N}\norm{\bfh^{t,j} - \widetilde{\bfh}^{t,j}}_2^2\to 0$ as $N\to\infty$ in probability for all $t\in [T].$
    \end{itemize}
    Then, we have
    \[
        \frac{1}{N}\#\{i\in B(j) : \exists s\leq T \text{ such that }\bfsigma^s_i \neq \tilde{\bfsigma}^s_i \} \to 0
    \]
    as $N\to\infty$ in probability.
\end{lem}
\begin{proof}
    By a union bound, it suffices to fix $t \in [T]$ and show that $\frac{1}{N}\norm{\bfsigma^{t,j} - \widetilde{\bfsigma}^{t,j}}_0\to 0$ as $N\to\infty$ in probability.
    The assumption that $K^{(T)}$ is non-singular ensures that the probability measure $\Unif[0,1]\times \mathrm{Law}(h^t)$ is absolutely continuous with respect to Lebesgue measure on $[0,1]\times \R$. For $\eps>0$, let 
    \begin{align*}
        A_\eps &= \{i \in B(j) : d((U_i^t, \tilde{\bfh}_i^t),D)<\eps\}, \\
        B_\eps &= \{i\in B(j): |\bfh^t_i - \tilde{\bfh}^t_i|<\eps\} \\
        \Delta &= \{i\in B(j): \bfsigma^t_i \neq \tilde{\bfsigma}^t_i\}.
    \end{align*}
        We have $\Delta \subseteq  A_\eps \cup B_\eps^C$. Hence, for all $\eps,\eps'>0,$ we have
        \begin{align*}
            \Pr\left(\frac{1}{N}\norm{\bfsigma^{t,j} - \widetilde{\bfsigma}^{t,j}}_0 \geq \eps'\right)&\leq \frac{1}{\eps' N}\left(\E |A_\eps| +\E |B_\eps^C| \right).
        \end{align*}
        We have $\frac{1}{N}\E|B_\eps^C|\to 0$ by assumption, because for all $\alpha>0,$
        \begin{align*}
            \frac{1}{N}\E|B_\eps^C| &\leq \alpha  + \Pr[|B_\eps^C| \geq \alpha N] \\
            &\leq \alpha  + \Pr\left[\frac{1}{N}\norm{\bfh^{t,j} - \widetilde{\bfh}^{t,j}}_2^2 \geq \alpha \eps^2\right] \\
            &\to \alpha.
        \end{align*}
        Moreover, we have
        \begin{align*}
            \frac{1}{N}\E |A_\eps| &= \E_{U_i^t, h^t} \1_{D_\eps}.
        \end{align*}
        By continuity of probability, taking $\eps\to 0$ the right hand side converges to $\E_{U_i^t, h^t} \1_{D},$ which is zero by absolute continuity. This concludes the proof.
\end{proof}

\section{Dynamics at a fixed $\delta>0$}\label{sec:finite-delta-dynamics}
This section is organized as follows.
\begin{enumerate}
    \item[(1)] In \Cref{sec:gcond}, we apply the Gaussian conditioning technique to obtain a description of the law of $J$ conditioning on iterates of Algorithm \ref{alg:t-passes}.
    \item[(2)] In \Cref{sec:decoup}, we utilize part (1) to obtain a description of the law of the field contributions $\bfh^{(t),j \rightarrow k}$ from Algorithm \ref{alg:t-passes}. We show by induction that this law is approximately given by a stochastic process which ``reveals'' new randomness at each step, and this randomness is encoded by a simple standard Gaussian vector in $\R^N$. We say that this process is ``approximately decoupled'' because its entries in $\R^N$ are approximately independent. The only source of dependence comes from the low-dimensional coefficient terms $\hat{C}, \hat{Q}$ which are functions of the empirical measure of spins and fields. Hence, one expects them to concentrate for $N$ large, leading to deterministic coefficients in the $N\to\infty$ limit and hence a truly decoupled process. However, this last statement is only proved in step (6) below.
    \item[(3)] In \Cref{sec:field-contributions}, we describe a low-dimensional stochastic process over fields $\{h^{(t),j} \in \mathbb{R}^{t}:t\geq 1, j\in [1/\delta]\}$ and spins $\{\sigma^{(t),j} \in \{\pm 1\}^t: t\geq 1, j\in [1/\delta]\}$ at any finite $\delta>0$. This stochastic process, which we term the Effective Dynamics, or \emph{Effective Process}, plays the role of equations \eqref{eq:h-joint-law-alone} and \eqref{eq:Sigma}-\eqref{eq:A-upper-diag} at a fixed $\delta>0.$
    \item[(4)] In \Cref{sec:lifted}, we define a (still low-dimensional) ``lifted" Effective Process over field contributions $h^{(t),j \rightarrow k} \in \mathbb{R}^{t+1}$ for $j,k \in [1/\delta]$ which subsumes the distribution over $h^{(t),j},\sigma^{(t),j}$ in part (3). We additionally derive some key simplifying properties satisfied by this lifted Effective Process.
    \item[(5)] In \cref{sec:mapdecoup}, we show that if one takes the ``approximately decoupled'' process from part (3) and pretends that the coefficients $\hat{C},\hat{Q}$ concentrate to quantities $C,Q$ which are given by the low-dimensional lifted Effective Process from part (4), then what one gets is exactly a high-dimensional process which is independent across $i\in [N]$, with each of its entries equal in distribution to the lifted Effective Process.
    \item[(6)] Finally, in \cref{sec:mainthm} and \Cref{sec:main_induct}, we utilize the results in (1)-(5) to inductively  establish concentration of the order parameters $\hat{C},\hat{Q}$ to their corresponding asymptotic values $C,Q,$ and hence convergence of the true process in Algorithm \ref{alg:t-passes} to independent draws from the Effective Process.
\end{enumerate}

\subsection{Dynamics of the conditional law of $J$}\label{sec:gcond}

In this section, we utilize the iterative conditioning technique introduced in  \cite{bolthausen2014iterative} to describe the evolution of $J$'s distribution, conditioned on the iterates in Algorithm \ref{alg:t-passes}. The resulting description splits the projection of the spin blocks $\bfsigma^{t|j}$ onto $J$ as follows:
\begin{equation}
    \bfh^{t,j\to k} \approx \text{Independent Gaussian noise}+\text{deterministic contributions from past iterates}.
\end{equation}

To proceed with the proof, we begin by noting the following elementary decomposition.
\begin{lem}\label{lem:decomp}
For any $V \in \mathbb{R}^{N \times p}$ and matrix $J \in \mathbb{R}^{N \times N}$, we have
\begin{align*}
    J = P^\perp_V J P^\perp_V + P_V J P^\perp_V + JP_V.
\end{align*}
\end{lem}
\begin{proof}
    Recall that by definition $P^\perp_V + P_V = I$. The desired statement is then obtained by substituting the identity $I$ with $P^\perp_V + P_V$ as follows:
    \begin{align*}
        J & = JI\\
        &= J(P^\perp_V+P_V)\\
        &= IJ(P^\perp_V) + JP_V\\
        &= (P^\perp_V+P_V)J(P^\perp_V)+JP_V\\
        &=  P^\perp_V J P^\perp_V + P_V J P^\perp_V + JP_V.
    \end{align*}
\end{proof}

Next, we define the relevant filtration over which the conditioning will take place. Throughout, we assume that the pairs $(t,j)$, where $t\in\N$ and $j\in [1/\delta]$, are ordered lexicographically, i.e., we have $(t',j') < (t,j)$ if either $t' < t$ or both $t = t'$ and $j<j'$. Now, using the notation $\mathcal{F}(\{X_\alpha\})$ for the $\sigma$-field generated by the random variables $\{X_\alpha\}$, we define the filtration $(\mathcal{F}_j^t)_{t \in \mathbb{N},j \in [\frac{1}{\delta}]}$ as follows:
\begin{align}\label{eq:calS-defn}
    \mathcal{F}_j^t &= \mathcal{F}\left(\{\bfsigma^0\}\cup  \left(\bigcup_{s=0}^{t-1} \{J\bfsigma^{s|k}: k\in [1/\delta]\} \right) \cup \{J\bfsigma^{t|k} : k<j\}\right).
\end{align}
We further denote by $\mathcal{S}_j^t$ the set of random variables on which we are conditioning
\begin{equation}\label{eq:Stj}
\mathcal{S}_j^t \coloneqq \{\bfsigma^0\}\cup  \left(\bigcup_{s=0}^{t-1} \{J\bfsigma^{s|k}: k\in [1/\delta]\} \right) \cup \{J\bfsigma^{t|k} : k<j\},
\end{equation}
so that $\calF_j^t = \calF(\calS_j^t).$

Unlike the setting of \Cref{lem:GOE-ind}, Algorithm \ref{alg:t-passes} involves projections on random vectors dependent on the disorder matrix $J$. Following \cite{bolthausen2014iterative}, we handle such dependence by inductively applying \Cref{lem:GOE-ind} over the ordering $(t,j);$ see also \cite[Section 6.2]{feng2022unifying}. This results in the following proposition, describing the distribution of $J$ conditioned on $\mathcal{F}_j^t$ for all $t \in \mathbb{N}$, $j \in [1/k]$.

\begin{prop}\label{prop:gen_gauss_cond}
Consider the setting of Algorithm \ref{alg:t-passes} and define $P^{(t)|j} \in \mathbb{R}^{N \times N}$ to be the projector onto the subspace spanned by $\{\bfsigma^{0|j},\cdots, \bfsigma^{t|j}\}$.
   For all $t \in \mathbb{N}$ and $j \in [1/\delta]$, define
   \begin{align*}
     P^{(t)}(j) = \sum_{k < j}P^{(t)|k} +\sum_{k \geq j} P^{(t-1)|k}.
   \end{align*}
Then, conditioned on $\mathcal{F}^t_j$, $(I- P^{(t)}(j))J(I- P^{(t)}(j))$ is equal in law to $(I- P^{(t)}(j) )\wt{J} (I-  P^{(t)}(j) )$, where $\wt{J}$ is a GOE matrix independent of $\calF_j^t$ and $J.$ Equivalently,
   \begin{equation}\label{eq:J_cond}
\begin{split}
    J \Big\vert_{\mathcal{F}^j_t} &\stackrel{d}{=} (I- P^{(t)}(j) )\wt{J} (I-  P^{(t)}(j) ) +   P^{(t)}(j)  J   P^{(t)}(j)  +   P^{(t)}(j) J (I-  P^{(t)}(j)).
\end{split}
\end{equation}
\end{prop}
\begin{proof}
    The proof proceeds by induction over $(t,j)$. First, consider the base case $t=1,j=0$. 
    We have, by definition,
    \begin{align*}
    \mathcal{F}^{1}_0 =  \mathcal{F}\left(\{\bfsigma^0\}\cup \{J\bfsigma^{0|k}: k \in [1/\delta]\} \right).
    \end{align*}
By \Cref{lem:decomp} with $V=U^0 := \operatorname{span}(\bfsigma^{0|k}:k\in [1/\delta])$, we have
    \begin{align*}
        J = (I-P_{U^0})J(I-P_{U^0}) + P_{U^0}JP_{U^0}+P_{U^0}J(I-P_{U^0}).
    \end{align*}
Since $\bfsigma^{0|k} \perp \bfsigma^{0|j}$ for $j \neq k$, we have
\begin{align*}
    P_{U^0} = \sum_{k 
    \in [1/\delta]} P^{(0)|k}
\end{align*}
and $(I-P_{U^0})\bfsigma^{0|k} = \mathbf{0}$ for all $k \in [1/\delta].$ Moreover, since $\{\bfsigma^{0|k}: k \in [1/\delta]\}$ are independent of $J$, by applying \Cref{lem:GOE-ind}, we obtain that $(I-P_{U^0})J(I-P_{U^0})$ is independent of $\mathcal{F}^{1}_0$, which proves \eqref{eq:J_cond} for $t=1, j=0$.
    
    

    Now, suppose the claim holds for some $t \in \mathbb{N}, j \in [1/\delta]$, i.e., the random variable
    \begin{align*}
        J^{t}_\perp(j) \coloneqq (I- P^{(t)}(j))J(I- P^{(t)}(j)),
    \end{align*}
 conditioned on $\mathcal{F}^t_j$, is distributed as 
\begin{equation}\label{eq:ind_eq}
       J^{t}_\perp(j) \Big\vert_{\mathcal{F}^j_t} \stackrel{d}{=} (I- P^{(t)}(j) )\wt{J} (I-  P^{(t)}(j)).
 \end{equation}
Now, let $(\tilde{t},\tilde{j})$ denote the time indices following $(t,j)$ under the lexicographical ordering:
\begin{align*}
    (\tilde{t},\tilde{j}) =\begin{cases}
        (t,j+1), & j < 1/\delta\\
        (t+1,0), & j = 1/\delta.
    \end{cases}
\end{align*}
Now, since
\begin{align*}
    (I- P^{(\tilde{t})}(\tilde{j})) = (I- P^{t|j})(I- P^{(t)}(j)) = (I- P^{(t)}(j))(I- P^{t|j}),
\end{align*}
we obtain
\begin{equation}\label{eq:J_perp}
\begin{split}
J^{\tilde{t}}_\perp(\tilde{j})&=(I- P^{(\tilde{t})}(\tilde{j}))J(I- P^{(\tilde{t})}(\tilde{j}))\\ 
&= (I- P^{t|j})(I- P^{(t)}(j))J(I- P^{(t)}(j))(I- P^{t|j})\\
    &=   (I- P^{t|j})J^{t}_\perp(j)(I- P^{t|j}).
\end{split}
\end{equation}
Now recall that $J^{t}_\perp(j)$ is Gaussian conditioned on $\mathcal{F}^t_j$, while $\sigma^{t|j}$ is measurable with respect to $\mathcal{F}^{t}_{j}$.  Therefore, conditioned on $\mathcal{F}^{t}_{j}$, the random variable $J^{\tilde{t}}_\perp(\tilde{j})$ is Gaussian as well, and independent of $J.$ In what remains of the proof, we will show that, conditional on $\calF_j^t$, the random variables $J\bfsigma^{t|j}$ and $J_\perp^{\tilde{t}}(j)$ are independent. Note that this suffices to conclude the proof, because since $\calF^{\tilde{t}}_{\tilde{j}} = \calF(\calS^t_j \cup\{J\bfsigma^{t|j}\}),$ we will conclude by \eqref{eq:J_perp} that $J_\perp^{\tilde{t}}(j)$ is Gaussian conditioned on $\calF^{\tilde{t}}_{\tilde{j}}$, with 
\[
J^{\tilde{t}}_\perp(\tilde{j}) \Big\vert_{\mathcal{F}^j_t} \stackrel{d}{=} (I- P^{(\tilde{t})}(\tilde{j}) )\wt{J} (I-  P^{(\tilde{t})}(\tilde{j})).
\]
To prove the independence of $J\bfsigma^{t|j}$ and $J_\perp^{\tilde{t}}(j)$ conditional on $\calF_j^t$, note that we have the decomposition
\begin{equation}\label{eq:split}
    J\bfsigma^{t|j} = \underbrace{J^{t}_\perp(j)\bfsigma^{t|j}}_{\bfz^{t|j}}+\underbrace{P^{(t)}(j) J (I-  P^{(t)}(j))\bfsigma^{t|j} +  J   P^{(t)}(j) \bfsigma^{t|j}}_{\bfzeta^{t|j}}.
\end{equation}
Observe that, conditioned on conditioned on $\mathcal{F}^{t}_{j}$, $\bfz^{t|j}$ is Gaussian while $\bfzeta^{t|j}$ is deterministic. 
By equations \eqref{eq:ind_eq}  and \eqref{eq:J_perp}, we obtain the joint equality in law
\begin{align*}
\left(J^{\tilde{t}}_\perp(\tilde{j}), \bfz^{t|j} \right)\Big\vert_{\mathcal{F}^j_t} \stackrel{d}{=} \left( (I- P^{(\tilde{t})}(\tilde{j}))\tilde{J}(I- P^{(\tilde{t})}(\tilde{j})),(I- P^{(t)}(j) )\wt{J} (I-  P^{(t)}(j))\bfsigma^{t|j} \right).
\end{align*}
Moreover, we have $(I- P^{(\tilde{t})}(\tilde{j}))(I-  P^{(t)}(j))\bfsigma^{t|j}=0$, so by \Cref{lem:GOE-ind}, we obtain that conditioned on  $\mathcal{F}^t_j$, $J^{\tilde{t}}_\perp(\tilde{j})$ is independent of $\bfz^{t|j}$, and hence also of $J\bfsigma^{t|j}.$ This concludes the proof.
\end{proof}

\subsection{Approximate law of the field contributions $\bfh^{t, j \rightarrow k}$}\label{sec:decoup}

We now leverage \Cref{prop:gen_gauss_cond} to obtain a coupling between the field contributions $ \bfh^{t,j\rightarrow k}, j, k \in [1/\delta]$ and the vectors $\bfsigma,\bfh$ generated by Algorithm \ref{alg:t-passes}. The resulting coupling is recursively described in the following lemma.

\begin{prop}\label{lem:ind_finite_N}
Consider the setting of Algorithm \ref{alg:t-passes} and recall the definitions of the random variables $\hat{C}^{(t)} \in \mathbb{R}^{(t+1) \times (t+1)}, \hat{Q}^{(t)} \in \mathbb{R}^{(t+1) \times (t+1)}$ from equations \eqref{eq:hatc},  \eqref{eq:hatq}, respectively. Define $\hat{b}^{t,j},  \hat{\psi}^{t,j},   \hat{\theta}^{t,j\rightarrow k}$ for $t \geq 1$ to be the following sequence of random variables:
\begin{align*}
    \hat{b}^{t,j} &= \sqrt{1- \hat{C}^{(t),j}_{t,0:t-1}(\hat{C}^{(t-1),j})^{\dagger} \hat{C}^{(t),j}_{0:t-1,t}} \;\; \in \R\\
   \hat{\psi}^{t,j}&=(\hat{C}^{(t-1),j})^{\dagger} \hat{C}^{(t),j}_{0:t-1,t} \;\; \in \R^t\\
   \hat{\theta}^{t,j\rightarrow k}&= \begin{cases}
       (\hat{C}^{(t),k})^{\dagger} \hat{Q}^{(t),k \rightarrow j}_{0:t,t}-(\hat{C}^{(t),k})^{\dagger} \hat{Q}^{(t),k \rightarrow j}_{0:t,0:t-1}(\hat{C}^{(t-1),j})^{\dagger} \hat{C}^{(t),j}_{0:t-1,t} &\in \R^{t+1}, k<j\\
       (\hat{C}^{(t-1),k})^{\dagger} \hat{Q}^{(t-1),k \rightarrow j}_{0:t-1,t}-(\hat{C}^{(t-1),k})^{\dagger} \hat{Q}^{(t-1),k \rightarrow j}(\hat{C}^{(t-1),j})^{\dagger} \hat{C}^{(t),j}_{0:t-1,t} &\in \R^{t},\;\;\;\, k \geq j.
   \end{cases}
\end{align*}
Let $\bfsigma, \bfh$ be given by Algorithm \ref{alg:t-passes} and let $(\mathcal{F}^t_j)_{t \in \mathbb{N}, j \in [1/\delta]}$ denote the filtration defined by \eqref{eq:calS-defn}. For $t \in \mathbb{N}, j \in [1/\delta]$ , let $W^t_j$ be independent random matrices distributed as $W^t_j \sim \text{GOE}(N)$. 
Define the extended filtration $\tilde{\mathcal{F}}^t_j$ for $t \in \mathbb{N}, j \in [1/\delta]$ as follows:
\begin{equation}\label{eq:ext_filt}
 \tilde{\mathcal{F}}^t_j = \mathcal{F}(\mathcal{S}^t_j \cup (W_{k}^s)_{(s,k)\leq (t,j)}),
\end{equation}
where $\mathcal{S}^t_j$ denotes the set of random variables defined in \eqref{eq:Stj}.

Then, there exists a sequence of Gaussian vectors $\{\bfxi^{t,j\rightarrow k}\}_{t\in \mathbb{N},j,k \in [\frac{1}{\delta}]}$ in $\R^N$ 
with $\bfxi^{t,j\rightarrow k} \sim \mathcal{N}(0,D^k)$ and vectors $\Delta^{t,j\to k}$ in $\R^t$ for $k<j$ and in $\R^{t+1}$ for $k\geq j$ such that the following hold:
\begin{enumerate}
    \item  $\bfxi^{t,j\rightarrow k}$ is measurable with respect to $\tilde{\mathcal{F}}^s_{\ell}$ whenever either $\ell>j$ or $s>t$.
    \item For any $k \in [1/\delta]$, $\bfxi^{t,j\rightarrow k}$ are independent across $j$ and $t$. 
    \item For all $j,k \in [1/\delta]$, we have
    \begin{equation}\label{eq:rech0j}
    \bfh^{0,j\rightarrow k} = \sqrt{\delta}\bfxi^{0,j\rightarrow k} + \bfsigma^{(0)|k}\Delta^{0,j \rightarrow k}.
    \end{equation}

    \item For all $t \geq 1$ and $k, j \in [1/\delta]$, we have
\begin{equation}\label{eq:rechtj}
    \bfh^{t,j\rightarrow k}  = \begin{cases}
        \hat{b}^{t,j} \sqrt{\delta}\bfxi^{t,j\rightarrow k}+ \bfh^{(t-1),j \rightarrow k} \hat{\psi}^{t,j}+\delta\bfsigma^{(t)|k}\hat{\theta}^{t,j\rightarrow k} +   \bfsigma^{(t)|k}\Delta^{t,j \rightarrow k}, & k < j \\
         \hat{b}^{t,j} \sqrt{\delta}\bfxi^{t,j\rightarrow k}+ \bfh^{(t-1),j \rightarrow k} \hat{\psi}^{t,j}+\delta\bfsigma^{(t-1)|k}\hat{\theta}^{t,j\rightarrow k}+   \bfsigma^{(t-1)|k}\Delta^{t,j \rightarrow k}, & k \geq j.
    \end{cases}
\end{equation}

\item For all $t\geq 0$ and $j,k\in [1/\delta],$ we have
\begin{align*}
\norm{\Delta^{(t),j \rightarrow k}}_F \xrightarrow[N \rightarrow \infty]{P} 0.
\end{align*}
\end{enumerate}

\end{prop}
\begin{proof}
We split the proof into the cases $t=0$ and $t\geq 1.$ In each case, the proof begins by explicitly constructing candidate vectors $\bfxi^{t,j\to k}$. Then, in the case $t=0,$ claims (1)-(3) and (5) will be verified, and in the case $t\geq 1,$ claims (1), (2), (4) and (5) will be verified.

\paragraph{Case $t=0$.} We begin by constructing the vectors $\bfxi^{0,j\to k}$. For $j \neq k$, we set
\begin{align*}
\bfxi^{0,j\rightarrow k} \coloneqq
\frac{1}{\sqrt{\delta}}D^kJ\bfsigma^{0|j},
\end{align*}
while for $j=k$, we set
\begin{equation}\label{eq:xi0j}
\bfxi^{0,j\rightarrow j} \coloneqq
 D^j \frac{1}{\sqrt{\delta}}\left(I + \frac{1}{\delta N} \bfsigma^{0|j} (\bfsigma^{0|j})^\top \right)^{-1/2}J\bfsigma^{0|j}.
\end{equation}
We now proceed to checking claims (1)-(3) and (5). Claim (1) is immediately verified by construction. We now verify claim (2). Note that for $j \in [1/\delta]$, $\bfsigma^{0|j}$ is independent of $J$ and $\norm{\bfsigma^{0|j}}^2=\delta N$. If we define $\mathcal{F}^0 \coloneqq  \mathcal{F}\left(\{\bfsigma^0\}\right),$ \Cref{lem:GOE_proj_law} implies that
\begin{equation}\label{eq:xi0}
    J\bfsigma^{0|j} \Big|_{\mathcal{F}^0} \sim \mathcal{N}\left(\mathbf{0},\delta I + \frac{1}{N} \bfsigma^{0|j} (\bfsigma^{0|j})^\top\right).
\end{equation}
Since $\bfh^{0, j \rightarrow k} = D^k J\bfsigma^{0|j}$, we obtain
\begin{align*}
    \bfh^{0, j \rightarrow k} \Big|_{\mathcal{F}^{0}} \sim \mathcal{N}\left(\mathbf{0},\delta D^k+\frac{1}{N} D^k\bfsigma^{0|j} (\bfsigma^{0|j})^\top D^k\right).
\end{align*}
When $j \neq k$, we have $D^k\bfsigma^{0|j} = 
 \mathbf{0},$ so \Cref{lem:GOE-ind} implies that for any $k \in [1/\delta]$,  conditioned on $\mathcal{F}^1_0$, $\bfxi^{0,j\rightarrow k}$ are independent across $j$, with
\begin{align*}
\bfxi^{0,j\rightarrow k} \Big|_{\mathcal{F}^{0}}  \sim \mathcal{N}(\mathbf{0}, D^k).
\end{align*}
We may therefore lift the conditioning to obtain claim (2), as well as the fact that
\begin{align*}
\bfxi^{0,j\rightarrow k}   \sim \mathcal{N}(\mathbf{0},D^k).
\end{align*}
Finally, for claims (3) and (5), we have that, for $j \neq k$,
\begin{align*}
    \bfh^{0,j\rightarrow k} = \sqrt{\delta}\bfxi^{0,j\rightarrow k},
\end{align*}
while for $j=k$, by \eqref{eq:xi0j}, we have
\begin{equation}\label{eq:xidiff0}
    \bfh^{0,j\rightarrow j}- \sqrt{\delta}\bfxi^{0,j\rightarrow j} = 
D^j \left(I-\left(I + \frac{1}{\delta N} \bfsigma^{0|j} (\bfsigma^{0|j})^\top \right)^{-1/2}\right)J\bfsigma^{0|j}.
\end{equation}
It's easy to verify that
\begin{equation}\label{eq:sqrt}
     \left(I + \frac{1}{\delta N} \bfsigma^{0|j} (\bfsigma^{0|j})^\top \right)^{1/2} =  \left(I + \frac{\sqrt{2}-1}{\delta N} \bfsigma^{0|j} (\bfsigma^{0|j})^\top \right).
\end{equation}
Hence applying the Sherman-Morrison formula, we have
\begin{align*}
    \left(I + \frac{1}{\delta N} \bfsigma^{0|j} (\bfsigma^{0|j})^\top \right)^{-1/2} = \left(I - \frac{\sqrt{2}-1}{\delta N \sqrt{2}} \bfsigma^{0|j} (\bfsigma^{0|j})^\top \right).
\end{align*}
Substituting back in \eqref{eq:xidiff0} yields
\begin{align*}
    \bfh^{0,j\rightarrow j}- \sqrt{\delta}\bfxi^{0,j\rightarrow j} &= 
D^j\frac{\sqrt{2}-1}{\sqrt{2} \delta N}\bfsigma^{0|j} (\bfsigma^{0|j})^\top J\bfsigma^{0|j}\\
&=\bfsigma^{0|j}\Delta^{0, j \rightarrow j},
\end{align*}
where
\begin{align*}
    \Delta^{0, j \rightarrow j} = \frac{\sqrt{2}-1}{\sqrt{2}\delta N}(\bfsigma^{0|j})^\top J\bfsigma^{0|j}.
\end{align*}
\Cref{lem:proj_gauss} then implies that
\begin{align*}
    \norm{\Delta^{0,j \rightarrow j}}_2 \xrightarrow[N \rightarrow \infty]{P} 0,
\end{align*}
as desired. 

\paragraph{Case $t\geq 1$.}  We again begin by constructing the vectors $\bfxi^{t
,j\to k}$. By \Cref{lem:decomp}, for all $t \in \mathbb{N}$ and $j,k \in [1/\delta]$, we have
\begin{equation}\label{eq:J-expansion}
    J = (I- P^{(t)}(j) )J (I-  P^{(t)}(j))+ P^{(t)}(j) J (I-  P^{(t)}(j))+ J   P^{(t)}(j),
\end{equation}
where we recall that
\begin{align*}
    P^{(t)}(j) = \sum_{\ell < j}P^{(t)|\ell} +\sum_{\ell \geq j} P^{(t-1)|\ell}.
\end{align*}
Now observe that for any $s, t \in \mathbb{N}$, and  $\ell \neq m \in [1/\delta]$, $\bfsigma^{s|\ell}$ and $\bfsigma^{t|m}$ have non-zero entries in  disjoint subsets of coordinates. Therefore, $\bfsigma^{s|\ell} \perp \bfsigma^{t|m}$ whenever $\ell \neq m$. 
As a consequence, we have
\begin{align*}
    P^{(t)}(j) \bfsigma^{t|j} &= P^{(t-1)|j} \bfsigma^{t|j}
\end{align*}
and
\begin{align*}
     D^k P^{(t)}(j)&= \begin{cases}
         P^{(t)|k} & k < j\\
          P^{(t-1)|k} & k \geq j.
     \end{cases}
\end{align*}
Substituting the above simplifications in \eqref{eq:J-expansion}, for $k<j$ we get
\begin{align*}
     \bfh^{t,j \rightarrow k} = \underbrace{D^k(I-P^{(t)|k})J (I-P^{(t-1)|j}) \bfsigma^{t|j}}_{H^{t,j \rightarrow k}_1} +  \underbrace{D^k J P^{(t-1)|j} \bfsigma^{t|j} +  P^{(t)|k} J (I-P^{(t-1)|j})\bfsigma^{t|j}}_{H^{t,j \rightarrow k}_2},
\end{align*}
while for $k \geq j$ we get
\begin{align*}
     \bfh^{t,j \rightarrow k} = \underbrace{D^k(I-P^{(t-1)|k})J (I-P^{(t-1)|j}) \bfsigma^{t|j}}_{H^{t,j \rightarrow k}_1} +  \underbrace{D^k J P^{(t-1)|j} \bfsigma^{t|j} +  P^{(t-1)|k} J (I-P^{(t-1)|j})\bfsigma^{t|j}}_{H^{t,j \rightarrow k}_2}.
\end{align*}

We first consider the terms $H^{t,j \rightarrow k}_1$. For $j\in [1/\delta],$ let $\bfu^{t,j} = (I-P^{(t-1)|j}) \bfsigma^{t|j}$. Then $\bfu^{t,j}$ is measurable with respect to $\mathcal{F}^t_j$. Therefore, applying \Cref{prop:gen_gauss_cond}, we obtain
\begin{equation}\label{eq:ht1}
\begin{split}
    H^{t,j \rightarrow k}_1 \Big|_{ \tilde{\mathcal{F}}^t_j} &\stackrel{d}{=} \begin{cases}
     D^k(I-P^{(t)}(j)) \tilde{J} (I-P^{(t)}(j))\bfu^{t,j} & k < j\\
         D^k(I-P^{(t)}(j))  \tilde{J} (I-P^{(t)}(j)) \bfu^{t,j} & k \geq j
    \end{cases}   \\
    &= \begin{cases}
     D^k(I-P^{(t)|k}) \tilde{J} \bfu^{t,j} & k < j\\
         D^k(I-P^{(t-1)|k}) \tilde{J} \bfu^{t,j} & k \geq j,
    \end{cases}  
\end{split}
\end{equation}
where $\tilde{J} \sim \text{GOE}(N)$. Now, recall $W^t_j \sim \text{GOE}(N)$ from \eqref{eq:ext_filt}. By \eqref{eq:ht1}, for $k<j$ we have
\begin{equation}\label{eq:ht1_dist1}
    H^{t,j \rightarrow k}_1+ P^{(t)|k} W^{t}_j \bfu^{t,j}  \Big|_{ \tilde{\mathcal{F}}^t_j } \stackrel{d}{=} D^k(I-P^{(t)|k}) \tilde{J} \bfu^{t,j} + P^{(t)|k} W^{t}_j \bfu^{t,j}, 
\end{equation}
and similarly for $k \geq j$,
\begin{equation}\label{eq:ht1_dist2}
    H^{t,j \rightarrow k}_1+ P^{(t-1)|k} W^{t}_j \bfu^{t,j}  \Big|_{ \tilde{\mathcal{F}}^t_j}  \stackrel{d}{=} D^k(I-P^{(t-1)|k}) \tilde{J} \bfu^{t,j} + P^{(t-1)|k} W^{t}_j \bfu^{t,j}.
\end{equation}
Let 
\[
\Sigma_{\bfu^{t,j}} \coloneqq \frac{1}{N}(\norm{\bfu^{t,j}}^2I+(\bfu^{t,j})(\bfu^{t,j})^\top).
\]
By
\Cref{lem:GOE_proj_law}, both $\tilde{J} \bfu^{t,j}$ 
and $W^t_j \bfu^{t,j}$ are distributed as
$\mathcal{N}(\mathbf{0}, \Sigma_{\bfu^{t,j}})$. Hence, note that the right hand sides of equations \eqref{eq:ht1_dist1}, \eqref{eq:ht1_dist2} are Gaussian with covariance given by $ D^k(\Sigma_{\bfu^{t,j}}) D^k$. Indeed, for \eqref{eq:ht1_dist1}, we have 
\begin{align*}
    \E&\left[(D^k (I-P^{(t)|k}) \tilde{J}\bfu^{t,j} + P^{(t)|k}W^{t}_j\bfu^{t,j})( (\tilde{J}\bfu^{t,j})^\top (I-P^{(t)|k}) D^k + (W^{t}_j\bfu^{t,j})^\top P^{(t)|k}) \;|\; \calF^t_j\right] \\
    &\qquad = D^k \Sigma_{\bfu^{t,j}} D^k - D^k \Sigma_{\bfu^{t,j}} P^{(t)|k} -  P^{(t)|k}\Sigma_{\bfu^{t,j}}  D^k + 2P^{(t)|k} \Sigma_{\bfu^{t,j}} P^{(t)|k} \\
    & \qquad = D^k \Sigma_{\bfu^{t,j}} D^k + \frac{\norm{\bfu^{t,j}}^2}{N}\left(-P^{(t)|k}-P^{(t)|k}+2P^{(t)|k}\right) \\
    & \qquad = D^k \Sigma_{\bfu^{t,j}} D^k,
\end{align*}
and similarly for \eqref{eq:ht1_dist2}. Hence, for $k<j$,
\begin{equation}\label{eq:htj_dist1}
     H^{t,j \rightarrow k}_1+ P^{(t)|k} W^{t}_j \bfu^{t,j}  \Big|_{ \tilde{\mathcal{F}}^t_j}  \stackrel{d}{=} \mathcal{N}\left(\mathbf{0}, \frac{1}{N}\left(\norm{\bfu^{t,j}}^2  D^k +D^k(\bfu^{t,j})(\bfu^{t,j})^\top D^k\right)\right)
\end{equation}
as well as for $k \geq j$,
\begin{equation}\label{eq:htj_dist2}
     H^{t,j}_1+ P^{(t-1)|k} W^{t}_j \bfu^{t,j}  \Big|_{ \tilde{\mathcal{F}}^t_j}  \stackrel{d}{=} \mathcal{N}\left(\mathbf{0}, \frac{1}{N}\left(\norm{\bfu^{t,j}}^2  D^k +D^k(\bfu^{t,j})(\bfu^{t,j})^\top D^k\right)\right).
\end{equation}
When $k \neq j$, $D^k\bfu^{t,j} = \mathbf{0}$. Motivated by this, we  set ${\bfxi}^{t,j \rightarrow k}$ as
\begin{equation}\label{eq:xidef}
    {\bfxi}^{t,j \rightarrow k} \coloneqq  \frac{\sqrt{N}}{\norm{\bfu^{t,j}}}\begin{cases}
        H^{t,j \rightarrow k}_1+ P^{(t)|k} W^{t}_j \bfu^{t,j} & k < j\\
         H^{t,j \rightarrow k}_1+ P^{(t-1)|k} W^{t}_j \bfu^{t,j} & k > j.
    \end{cases}
\end{equation}
For the case $j=k$, we must remove the spike  $(\bfu^{t,j})(\bfu^{t,j})^\top $ in the covariance. Because of this, we set
\begin{equation}\label{eq:xidefj}
     {\bfxi}^{t,j \rightarrow j} = \sqrt{N}\Sigma^{-1/2}_{\bfu^{t,j}}(H^{t,j \rightarrow j}_1+ P^{(t-1)|j} W^{t}_j \bfu^{t,j}).
\end{equation}
Having constructed the vectors $\bfxi^{t,j\to k}$, we now proceed to verifying claims (1), (2), (4) and (5) in the statement of the lemma. Claim (1) is once again immediate by construction. Equations \eqref{eq:htj_dist1} and  \eqref{eq:htj_dist2} imply that for any $k \in [1/\delta]$ we have
\begin{align*}
    {\bfxi}^{t,j \rightarrow k} \Big|_{  \tilde{\mathcal{F}}^t_j} \stackrel{d}{=} \mathcal{N}(0,  D^k).
\end{align*}
As a result, this conditional distribution of ${\bfxi}^{t,j}$ is independent of $\tilde{\calF}^t_j$ and $ {\bfxi}^{t,j \rightarrow k} \sim \mathcal{N}(0, D^k)$ unconditionally. In particular, ${\bfxi}^{t,j \rightarrow k}$ is independent of ${\bfxi}^{s,\ell \rightarrow k}$ for $s \leq t, \ell < j$ and  $s < t, \ell  > k$. This verifies claim (2). For the remainder of the proof, we verify claims (4) and (5).

First, we show that $H^{t,j \rightarrow k}_1\approx \sqrt{\delta}\hat{b}^{t,j}{\bfxi}^{t,j}$. For this, we first verify that we have $\hat{b}^{t,j} = \frac{\norm{\bfu^{t,j}}}{\sqrt{\delta N}}$. To see this, recall from \eqref{eq:projection-pseudoinverse} that we have 
\begin{align*}
    P^{(t)|j} = \bfsigma^{(t)|j}((\bfsigma^{(t)|j})^\top\bfsigma^{(t)|j}))^\dagger(\bfsigma^{(t)|j})^\top,
\end{align*}
which can be further simplified via \eqref{eq:hatc} to
\begin{align*}
    P^{(t)|j} = \frac{1}{\delta N}\bfsigma^{(t)|j}(\hat{C}^{(t),j})^{\dagger}(\bfsigma^{(t)|j})^\top.
\end{align*}
Therefore,
\begin{equation}\label{eq:bu}
\begin{split}
    \norm{\bfu^{t,j}}^2 &=
          \norm{(I-P^{(t-1)|j}) \bfsigma^{t|j}}^2\\
          &= \norm{\bfsigma^{t|j}}^2 - \norm{P^{(t-1)|j}\bfsigma^{t|j}}^2 \\
          &= \delta N - \frac{1}{(\delta N)^2} \norm{\bfsigma^{(t-1)|j}(\hat{C}^{(t-1),j})^{\dagger} (\bfsigma^{(t-1)|j})^\top \bfsigma^{t|j}}^2 \\
            &= \delta N - \hat{C}^{(t),j}_{t,0:t-1 }(\hat{C}^{(t-1),j})^{\dagger}(\bfsigma^{(t-1)|j})^\top\bfsigma^{(t-1)|j} (\hat{C}^{(t-1),j})^{\dagger} C^{(t),j}_{0:t-1,t} \\
          &= \delta N- \delta N\hat{C}^{(t),j}_{t,0:t-1 }(\hat{C}^{(t-1),j})^{\dagger} C^{(t),j}_{0:t-1,t}\\
          &= \delta N (\hat{b}^{t,j})^2,
\end{split}
\end{equation}
as desired.
We next establish control over the error $H^{t,j \rightarrow k}_1- \sqrt{\delta} \hat{b}^{t,j} {\bfxi}^{t,j \rightarrow k}$. For the case $j \neq k$, equations \eqref{eq:xidef} and \eqref{eq:bu} imply that 
\begin{align*}
  H^{t,j \rightarrow k}_1- \sqrt{\delta}\hat{b}^{t,j}{\bfxi}^{t,j \rightarrow k}= \begin{cases}
       P^{(t)|k} W^{t}_j \bfu^{t,j}, & k < j\\
      P^{(t-1)|k} W^{t}_j \bfu^{t,j} & k \geq j.
  \end{cases}
\end{align*}
In either case, by the definition of $P^{(t)|k}$ and $P^{(t-1)|k}$, the above difference lies in the span of $\bfsigma^{(t)|k}$:
\begin{align*}
  H^{t,j \rightarrow k}_1- \hat{b}^{t,j}{\bfxi}^{t,j \rightarrow k} =\bfsigma^{(t)|k} \Delta^{t, j \rightarrow k}.
\end{align*}
Next, we show that $\Delta^{t, j \rightarrow k} \xrightarrow[N \rightarrow \infty]{P} 0$. Note that $\Delta^{t, j \rightarrow k}$ is a projection of a GOE matrix onto a subspace of dimension which does not grow with $N$ (either $t$ or $t+1$). Hence, to show that it vanishes in probability as $N\to \infty$ it suffices to check that $\frac{1}{\sqrt{N}}\norm{\bfu^{t,j}}$ remains bounded as $N\to \infty$. But this is immediate from the fact that $\norm{\bfu^{t,j}} \leq\norm{\bfsigma^{t|j}}=\sqrt{\delta N}.$

For the case $j=k$, we additionally require control over the contribution from $\Sigma^{-1/2}_{\bfu^{t,j}}$. Similar to \eqref{eq:sqrt}, we have
\begin{align*}
    (\norm{\bfu^{t,j}}^2I+(\bfu^{t,j})(\bfu^{t,j})^\top)^{1/2}= \left(\norm{\bfu^{t,j}}I+\frac{\sqrt{2}-1}{\norm{\bfu^{t,j}}}(\bfu^{t,j})(\bfu^{t,j})^\top\right).
\end{align*}
The Sherman-Morrison formula then implies:
\begin{align*}
\sqrt{N}\Sigma^{-1/2}_{\bfu^{t,j}}= \sqrt{N}(\norm{\bfu^{t,j}}^2I+(\bfu^{t,j})(\bfu^{t,j})^\top)^{-1/2} = \frac{\sqrt{N}}{\norm{\bfu^{t,j}}}I- \frac{\sqrt{N}(\sqrt{2}-1)}{ \norm{\bfu^{t,j}}^3}(\bfu^{t,j})(\bfu^{t,j})^\top.
\end{align*}
Therefore, \eqref{eq:xidefj} yields
\begin{equation}\label{eq:diff_eq}
\begin{split}
 &H^{t,j \rightarrow j}_1-  \sqrt{\delta} \hat{b}^{t,j}{\bfxi}^{t,j} \\
 &= H^{t,j \rightarrow j}_1-  \frac{\norm{\bfu^{t,j}}}{\sqrt{N}} \left(\frac{\sqrt{N}}{\norm{\bfu^{t,j}}}I- \frac{\sqrt{N}(\sqrt{2}-1)}{ \norm{\bfu^{t,j}}^3}(\bfu^{t,j})(\bfu^{t,j})^\top\right) (H^{t,j \rightarrow k}_1+ P^{(t-1)|j} W^{t}_j \bfu^{t,j})\\
 &= \frac{\sqrt{N} (\sqrt{2}-1)}{\norm{\bfu^{t,j}}^2}(\bfu^{t,j})(\bfu^{t,j})^\top H^{t,j \rightarrow j}_1 - \left(I- \frac{\sqrt{2}-1}{\norm{\bfu^{t,j}}^2}(\bfu^{t,j})(\bfu^{t,j})^\top\right) P^{(t-1)|j} W^{t}_j \bfu^{t,j}.
\end{split}
\end{equation}
By definition of $\bfu^{t,j}$ and $P^{(t-1)|j}$, we note that $H^{t,j \rightarrow j}_1-  \sqrt{\delta} \hat{b}^{t,j}{\bfxi}^{t,j}$ lies in the span of $\bfsigma^{(t)|j}$. Therefore, we may write
\begin{align*}
    H^{t,j \rightarrow j}_1-  \sqrt{\delta} \hat{b}^{t,j}{\bfxi}^{t,j} = \bfsigma^{(t)|j}  \Delta^{t,j \rightarrow j}.
\end{align*}
Now, since $\norm{\bfu^{t,j}}^2 \leq \sqrt{\delta N}$, by \eqref{eq:ht1},  we have
\begin{equation}\label{eq:htorth}
(\bfu^{t,j})^\top H^{t,j \rightarrow j}_1 \xrightarrow[N \rightarrow \infty]{P} 0.
\end{equation}
\Cref{lem:proj_gauss} further implies that
\begin{equation}\label{eq:resid_pt}
     \|\bfsigma^{(t-1)}_j W^t_j \bfu^{t,j}\|  \xrightarrow[N \rightarrow \infty]{P} 0,
\end{equation}
so equations \eqref{eq:htorth} and \eqref{eq:resid_pt} together imply that
\begin{align*}
    \Delta^{t,j \rightarrow k} \xrightarrow[N \rightarrow \infty]{P} 0.
\end{align*}

Finally, to conclude the proof of claims (4) and (5), and hence of the lemma, it remains to control the remaining terms $H^{t,j \rightarrow k}_2$:
\begin{equation}\label{eq:H2jk}
    H^{t,j \rightarrow k}_2 = \begin{cases}
        D^k J P^{(t-1)|j} \bfsigma^{t|j} +  P^{(t)|k} J (I-P^{(t-1)|j})\bfsigma^{t|j}, & k < j\\
        D^k J P^{(t-1)|j} \bfsigma^{t|j} +  P^{(t-1)|k} J (I-P^{(t-1)|j})\bfsigma^{t|j}, & k \geq j.
    \end{cases}  
\end{equation}
We treat the two terms in the right hand side separately. In what follows, we will focus on the case $k < j$, and the case $k \geq j$ will follow similarly. Substituting $P^{(t)|k}, P^{(t)|j}$ in $H^{t,j \rightarrow k}_2$, we get
\begin{equation}\label{eq:H2jk1}
    D^k J P^{(t-1)|j} \bfsigma^{t|j} = \frac{1}{\delta N}D^k J \bfsigma^{(t-1)|j}(\hat{C}^{(t-1),j})^{\dagger}(\bfsigma^{(t-1)|j})^\top \bfsigma^{t|j}, 
\end{equation}
and
\begin{equation}\label{eq:H2jk2}
\begin{split}
    P^{(t)|k} J (I-P^{(t-1)}_j)\bfsigma^{t|j} &= \frac{1}{\delta N}\bfsigma^{(t)|k}(\hat{C}^{(t),k})^{\dagger}(\bfsigma^{(t)|k})^\top J \bfsigma^{t|j}\\ &- \frac{1}{\delta N}\bfsigma^{(t)|k}(\hat{C}^{(t),k})^{\dagger}(\bfsigma^{(t)|k})^\top J \bfsigma^{(t-1)|j}(\hat{C}^{(t-1),j})^{\dagger}(\bfsigma^{(t-1)|j})^\top \bfsigma^{t|j} 
\end{split}
\end{equation}
Recalling the definition of $\hat{C}^{(t)}$ and $\bfh^{(t),j \rightarrow k}$, \eqref{eq:H2jk1} simplifies to
\begin{equation}\label{eq:p1}
\begin{split}
    D^k J P^{(t-1)|j} \bfsigma^{t|j}&=\frac{1}{\delta N} D^k J \bfsigma^{(t-1)|j}(\hat{C}^{(t-1),j})^{\dagger}(\bfsigma^{(t-1)|j})^\top  \bfsigma^{t|j}\\&= \frac{1}{N}\bfh^{(t-1),j \rightarrow k}(\hat{C}^{(t-1),j})^{\dagger} (\bfsigma^{(t-1)|j})^\top  \bfsigma^{t|j}\\
    &=\bfh^{(t-1),j \rightarrow k}(\hat{C}^{(t-1),j})^{\dagger}\hat{C}^{(t),j}_{0:t-1,t}
\end{split}
\end{equation}
Recognizing $(\hat{C}^{(t-1),j})^{\dagger}\hat{C}^{(t),j}_{0:t-1,t}$ as $\hat{\psi}^{t,j}$ yields the second term in the right hand side of \eqref{eq:rechtj}. We similarly proceed with \eqref{eq:H2jk2}. Recall that $D^j J \bfsigma^{(t)|k} = \bfh^{(t)|k \rightarrow j}$. Therefore,
\begin{align*}
     (\bfsigma^{(t)|k})^\top J \bfsigma^{t|j}  &=   (\bfh^{(t),k \rightarrow j})^\top \bfsigma^{t|j}\\
    &= \delta^2 N\hat{Q}^{(t), k \rightarrow j}_{0:t:,t},
\end{align*}
where we used the definition of $\hat{Q}^{(t)}$ in \eqref{eq:hatq}. Substituting the above into the terms appearing in the right hand side of \eqref{eq:H2jk2} leads to the following simplification
\begin{equation}\label{eq:theta_1}
\begin{split}
\frac{1}{\delta N}\bfsigma^{(t)|k}(\hat{C}^{(t),k})^{\dagger}(\bfsigma^{(t)|k})^\top J \bfsigma^{t|j} &= \frac{1}{\delta}\bfsigma^{(t)|k}(\hat{C}^{(t),k})^{\dagger} (\bfh^{(t),k \rightarrow j})^\top \bfsigma^{t|j} \\
&= \delta\bfsigma^{(t)|k}(\hat{C}^{(t),k})^{\dagger} \hat{Q}_{0:t,t}^{(t),k \rightarrow j}.
\end{split}
\end{equation}
Finally, we have
\begin{equation}\label{eq:theta_2}
\begin{split}
    \frac{1}{\delta N}\bfsigma^{(t)|k}(\hat{C}^{(t),k})^{\dagger}(&\bfsigma^{(t)|k})^\top J \bfsigma^{(t-1)|j}(\hat{C}^{(t-1),j})^{\dagger}(\bfsigma^{(t-1)|j})^\top \bfsigma^{t|j}\\ 
    &= \frac{1}{\delta N}\bfsigma^{(t)|k} (\hat{C}^{(t),k})^{\dagger}(\bfh^{(t),k \rightarrow j})^\top \bfsigma^{(t-1)|j}(\hat{C}^{(t-1),j})^{\dagger}(\bfsigma^{(t-1)|j})^\top \bfsigma^{t|j}\\
    &= \delta\bfsigma^{(t)|k}(\hat{C}^{(t),k})^{\dagger} \hat{Q}^{(t),k \rightarrow j}_{0:t, 0:t-1}(\hat{C}^{(t-1),j})^{\dagger} \hat{C}^{(t),j}_{0:t-1,t}.
\end{split}
\end{equation}
Substituting equations \eqref{eq:theta_1} and \eqref{eq:theta_2} into \eqref{eq:H2jk2} yields
\begin{equation}\label{eq:p2}
    P^{(t)|k} J (I-P^{(t-1)|j})\bfsigma^{t|j} = \delta \bfsigma^{(t)|k} \left((\hat{C}^{(t),k})^{\dagger} \hat{Q}^{(t),k \rightarrow j}_{0:t,t}-(\hat{C}^{(t),k})^{\dagger} \hat{Q}^{(t),k \rightarrow j}_{0:t,0:t-1}(\hat{C}^{(t-1),j})^{\dagger} \hat{C}^{(t),j}_{0:t-1,t}\right).
\end{equation}
We recognize the right hand side as $ \delta \bfsigma^{(t)|k}\hat{\theta}^{t,j}$. Substituting equations \eqref{eq:p1}, \eqref{eq:p2} back into \eqref{eq:H2jk}, we obtain
\begin{align*}
    H^{t,\ell \rightarrow j}_2 = \bfh^{(t-1),j \rightarrow k} \hat{\psi}^{t,j}+\delta\bfsigma^{(t)|k}\hat{\theta}^{t,j\rightarrow k}.
\end{align*}
Proceeding similarly for $k \geq j$ concludes the proof of the lemma.
\end{proof}

\subsection{Asymptotics at a fixed $\delta>0$}\label{sec:field-contributions}

In this section, we describe the limiting dynamics at a fixed $\delta>0$, which we call the Effective Dynamics, and which are the analogue of the equations \eqref{eq:h-joint-law-alone} and \eqref{eq:Sigma}-\eqref{eq:A-upper-diag} at a fixed $\delta>0.$ As discussed at the beginning of \Cref{sec:finite-delta-dynamics}, it will turn out that the process $\bfh, \bfsigma$ from Algorithm \ref{alg:t-passes} is approximately given by i.i.d. draws from this Effective Process. 

For any $j \in [\frac{1}{\delta}]$, we define the following Effective Dynamics stochastic processes over fields $ \{h^{t,j} \in \mathbb{R}\}_{t=1}^T$ and spins $ \{\sigma^{t,j} \in \{\pm 1\}\}_{t=0}^T$. We let $\sigma^{0,j}\sim \Unif\{\pm1\}$, and for $t=1,\dots,T,$
\begin{equation}\label{eq:h-joint-law-del}
\begin{cases}
        h^{t,j} = G^{t,j} + \brac{f^{t,j}, \sigma^{(t-1),j}}\\
         \sigma^{t,j} = c(h^{t,j}).
\end{cases}
\end{equation}
Above, $G^{1,j},\cdots, G^{T,j}$ are zero-mean Gaussian random variables with covariance $\Sigma^{(T),j}$, independent of $\sigma^{0,j}$, and $\Sigma^{(T),j}\in \R^{T\times T}, f^{t,j}\in \R^t,t=1,\dots,T$ are specified by the difference equations
\begin{equation}\label{eq:finite_del-dim}
\begin{split}
        &\Sigma^{(T),j} = \delta\sum_{k < j} C_{1:T, 1:T}^{(T),k}  + \delta\sum_{k \geq  j} C^{(T-1),k}\\
    &f^{t,j} =\delta \sum_{k < j} A_{0:t-1, t}^{(t),k} + \delta\sum_{k > j}\begin{bmatrix}
        0 \\
        A_{0:t-2, t-1}^{(t-1),k}
    \end{bmatrix} + \delta A_{0:t-1, t-1}^{(t-1),j} \qquad\qquad\;\;\;\;\, t=1,\dots,T\\ 
    &C_{st}^{(T),j} = \E[\sigma^{s,j} \sigma^{t,j}] \qquad\qquad\qquad\qquad \qquad\qquad\qquad\qquad \qquad\qquad \, s,t=0,\dots,T\\\\
    &R_{st}^{(T),j} =  \E[G^{s,j} \sigma^{t,j}]\qquad\qquad\qquad\qquad \qquad\qquad\qquad\qquad \qquad\qquad s=1,\dots,T,t=0,\dots,T\\\\
    &A_{0:t-1, t}^{(T),j} = (\Sigma^{(T),j})^{-1} R_{1:t, t}^{(T),j} +\left( A^{(T-1),j} - (\Sigma^{(T),j})^{-1} R_{1:t,0:t-1}^{(T),j} \right) (C^{(T-1),j})^{-1} C_{0:t-1,t}^{(T),j},
\end{split}
\end{equation}
where $A$ is upper triangular with zeros on the diagonal, and where we define $\Sigma^{(1),1}:=1$ and $A_{-1,0}:=0.$
We recall that it is proved in \Cref{prop:finite-delta-regularity} that the above matrix inverses are well-defined along the solution trajectory. Finally, we define three more quantities which will be used in our proofs below. These are
\begin{align}
    \psi^{t,j} &= (C^{(t-1),j})^\dagger C_{0:t-1,t}^{(t),j}\in \mathbb{R}^{t}\label{eq:psi-defn}  \\
    \theta^{t,j\rightarrow k}&= \begin{cases}
    (C^{(t),k})^{-1} Q^{(t),k \rightarrow j}_{0:t,t}-(C^{(t),k})^{-1} Q^{(t),k \rightarrow j}_{0:t,0:t-1}(C^{(t-1),j})^{-1} C^{(t),j}_{0:t-1,t}  &\in \R^{t+1}, k<j\\
       (C^{(t-1),k})^{-1} Q^{(t),k \rightarrow j}_{0:t-1,t}-(C^{(t-1),k})^{-1} Q^{(t-1),k \rightarrow j}(C^{(t-1),j})^{-1} C^{(t),j}_{0:t-1,t} &\in \R^{t},\;\;\;\, k \geq j.
   \end{cases} \label{eq:theta-defn} \\
   b^{t,j} &= \sqrt{1- C^{(t),j}_{t,0:t-1}(C^{(t-1),j})^{-1} C^{(t),j}_{0:t-1,t}} \in \R. \label{eq:b-defn}
\end{align}
We stress that the quantities defined in equations \eqref{eq:finite_del-dim}-\eqref{eq:b-defn} above are \emph{deterministic}, as opposed to the stochastic variants like $\hat{C}, \hat{Q}, \hat{\psi},\hat{\theta}$ and $\hat{b}$ defined in \Cref{sec:useful-lemmas} and \Cref{lem:ind_finite_N}.

\subsection{Lifted process over pairwise contributions}\label{sec:lifted}

An intermediate step towards obtaining \eqref{eq:finite_del-dim} as the limiting dynamics over spins and fields is to first define an Effective Process on the lifted space of field contributions which we can more directly compare with \Cref{lem:ind_finite_N}. We next describe this lifted Effective Dynamics and relate it to the Effective Dynamics specified by \eqref{eq:finite_del-dim}.

\begin{defn}[Lifted Effective Process]\label{def:joint_eff_proc}
  Let $\Sigma^{(T)}, f^{(T)}, C^{(T)}, A^{(T)}$ be given by the solution to the difference equations specified by \eqref{eq:finite_del-dim}.
  For $t \in [T], j,k \in [1/\delta]$, let $g^{(t), k \rightarrow j} \in \mathbb{R}^{t+1}$ be Gaussian random row vectors, independent across $j,k$ with covariance $\delta C^{(t),k}$, i.e
  \begin{align*}
      g^{(t), k \rightarrow j} \sim \mathcal{N}(0, \delta C^{(t),k}).
  \end{align*}
Let $\sigma^{0,j} \sim \Unif\{\pm1\}$, independently across $j$ and independent of the $g^{(t),k\to j}$, and for $t=1,\dots,T$, let
\begin{equation}\label{eq:law-lifted}
\begin{split}
\begin{cases}
        h^{t,j} = \sum_{k <j} h^{t,k \rightarrow j}+\sum_{k \geq j}h^{t-1,k \rightarrow j}\\
   \sigma^{t,j}  =  c(h^{t,j}),
\end{cases}
\end{split}
\end{equation}
where the $h^{t,k\to j}$ are defined as follows. 
\begin{align}\label{eq:kj-contribution-effective-process}
     h^{t,k \rightarrow j} &=\begin{cases}
         g^{t,k\to j} + \delta\sigma^{(t-1),j}A_{0:t-1,t}^{(t),k} & t=1,\dots,T \text{ and }k< j \\
             g^{t,k\to j} + \delta \sigma^{(t),j} A_{0:t,t}^{(t),k} & t=0,\dots,T-1 \text{ and } k=j\\
          g^{t,k\to j} + \delta \sigma^{(t),j} \begin{bmatrix}
        0 \\
        A_{0:t-1,t}^{(t),k}
    \end{bmatrix} & t=0,\dots,T-1  \text{ and } k > j.\\
     \end{cases} 
\end{align}
\end{defn}
The Effective Process defined above specifies a distribution for $\{\sigma^{t,j}\}_{t=0}^T, \{h^{t,j}\}_{t=1}^T$ which matches that of \eqref{eq:h-joint-law-del}.
\begin{prop}
The marginal law of $\{\sigma^{t,j}\}_{t=0}^T, \{h^{t,j}\}_{t=1}^T$ obtained from \Cref{def:joint_eff_proc} agrees with the law defined by \eqref{eq:h-joint-law-del}.
\end{prop}
\begin{proof}
For $t=1,\dots,T,$ let $\tilde{G}^{t,j}$ denote the random variable
\begin{align*}
   \wt{G}^{t,j}= \sum_{k <j} g^{t,k \rightarrow j}+\sum_{k \geq j}g^{t-1,k \rightarrow j}.
\end{align*}
Note that all terms in the sum above are independent, so for $1\leq s,t\leq T$, we have
\begin{align*}
    \E \wt{G}^{s,j} \wt{G}^{t,j} &= \sum_{k <j} \E g^{s,k \rightarrow j}g^{t,k \rightarrow j}+\sum_{k \geq j}\E g^{s-1,k \rightarrow j}g^{t-1,k \rightarrow j} \\
    &= \delta \sum_{k<j}C_{st}^{(T),k} + \delta \sum_{k\geq j}C_{s-1,t-1}^{(T-1),k} \\
    &= \Sigma_{s,t}^{(T),j},
\end{align*}
so $\wt{G}^{(t),j}$ agrees in law with $G^{(t),j}$ from \eqref{eq:h-joint-law-del}. Now according to \eqref{eq:law-lifted}, for $t=1,\dots,T,$ we have
\begin{align*}
        h^{t,j} &=  \tilde{G}^{t,j}+  \sigma^{(t-1),j}\left(\delta\sum_{k<j}A_{0:t-1,t}^{(t),k}  +\delta \sum_{k> j}\begin{bmatrix}
        0 \\
        A_{0:t-2, t-1}^{(t-1),k} 
    \end{bmatrix} + \delta A_{0:t-1,t-1}^{(t-1),j}
    \right) \\
    &= \tilde{G}^{t,j}+  \brac{f^{t,j}, \sigma^{(t-1),j}}.
\end{align*}
Since $\wt{G}^{(t),j}$ agrees in law with $G^{(t),j}$, and moreover the $\{\sigma^{t,j}\}_{t=0}^T, \{h^{t,j}\}_{t=1}^T$ from \Cref{def:joint_eff_proc} satisfy the same recursion as those from \eqref{eq:h-joint-law-del} with $\wt{G}$ in place of $G,$ and with the same distribution of $\sigma^{0,j}$, the proposition follows.
\end{proof}
We now derive a result that links certain parameters defined through the Effective Process in \Cref{def:joint_eff_proc} to the recursion on $A$ given by \eqref{eq:finite_del-dim}.  Define matrices $Q^{(T),k\to j}\in \R^{(T+1)\times(T+1)}$ for $k,j\in [1/\delta]$ as follows
\begin{equation}\label{eq:defQt}
   Q^{(T),k \rightarrow j}_{st}= \frac{1}{\delta}\mathbb{E}[h^{s, k \rightarrow j}\sigma^{t,j}].
\end{equation}
This matrix will be useful in the proof of the following lemma, as well as later when we have to compare the Effective Process to the vectors $\bfsigma^{(t)|j}, \bfh^{(t),k\to j}.$

\begin{lem}\label{lem:spin-term-simplifications}
Let $\psi^{t,j}$ and $\theta^{t,j\to k}$ be defined as in equations \eqref{eq:psi-defn} and \eqref{eq:theta-defn}, respectively. For $k < j$, we have
 \begin{equation}\label{eq:lemma7-k<j}
        \begin{bmatrix}
            0 \\A^{(t-1),j}\psi^{t,j}\end{bmatrix} +\theta^{t,j\rightarrow k} =  \begin{bmatrix}
            0\\ A_{0:t-1,t}^{(t),j}
            \end{bmatrix}, 
\end{equation}
and for $k \geq j$, we have
 \begin{equation}\label{eq:lemma7-k>j}
       A^{(t-1),j}\psi^{t,j}+\theta^{t,j\rightarrow k} =  A^{(t),j}_{0:t-1,t}.
\end{equation}
\end{lem}
\begin{proof}
The proof is by induction in $t.$ For $t=0,$ equations \eqref{eq:lemma7-k<j}, \eqref{eq:lemma7-k>j} are immediate since in that case $Q^{(t),k\to j}_{0,0}=0$ by independence of $\sigma^0$ and $h^{0,k\to j}$ and hence $\theta^{0,j\to k} = 0$. In the rest of the proof we assume equations \eqref{eq:lemma7-k<j}, \eqref{eq:lemma7-k>j} at time $t-1$ and verify them at time $t.$ We split into the cases $j<k,j>k$ and $j=k.$

    We begin with the case $j<k.$ From \eqref{eq:theta-defn}, we have
    \begin{align*}
        \theta^{t, j\to k} &= (C^{(t-1),k})^{-1} Q^{(t),k \rightarrow j}_{0:t-1,t}-(C^{(t-1),k})^{-1} Q^{(t-1),k \rightarrow j}(C^{(t-1),j})^{-1} C^{(t),j}_{0:t-1,t}.
    \end{align*}
    Now, since $j< k,$ by \eqref{eq:kj-contribution-effective-process}, we have 
        \begin{align*}
        \delta Q^{(t-1), k\to j} &= \E \left( g^{(t-1),k\to j} + \delta \sigma^{(t-1),j} \begin{bmatrix}
        0 \\
        A_{0:t-2,0:t-1}^{(t-1),k}
    \end{bmatrix}\right)^\top\sigma^{(t-1),j} \\
    &= \E ( g^{(t-1),k\to j})^\top \sigma^{(t-1),j} + \delta \begin{bmatrix}
        0 \\
        A_{0:t-2,0:t-1}^{(t-1),k}
    \end{bmatrix}^\top C^{(t-1),j}.
    \end{align*}
    Now, recalling that $\partial_v$ for a vector $v$ to denotes the (weak) gradient with respect to $v$, by Stein's Lemma (\Cref{lem:Stein}), we have
    \begin{align*}
    \delta Q^{(t-1), k\to j}&= \delta C^{(t-1),k}\E[\partial_{g^{(t-1), k\to j}} \sigma^{(t-1),j}]+ \delta \begin{bmatrix}
        0 \\
        A_{0:t-2,0:t-1}^{(t-1),k}
    \end{bmatrix}^\top C^{(t-1),j},
    \end{align*}
    where we have used that $c$, the update rule, has a distributional derivative since it's bounded (see \Cref{lem:distributional-derivative}). Now, for each $s,s'\in \{0,\dots,t-1\}$ and $r\in \{1,\dots,t\}$, note from \Cref{def:joint_eff_proc} that, for $G^{r,j}:= \sum_{k <j} g^{r,k \rightarrow j}+\sum_{k \geq j}g^{r-1,k \rightarrow j},$ since $k>j$ we have $\partial_{g^{s',k\to j}} G^{r,j} = \delta_{r,s'+1}$. Hence, using the chain rule, we get
    \begin{align*}
        \partial_{g^{s',k\to j}} \sigma^{s,j} &= \sum_r (\partial_{g^{s',k\to j}} G^{r,j} )\partial_{G^{r,j}} \sigma^{s,j},\\
        &= \partial_{G^{s'+1,j}} \sigma^{s,j},
    \end{align*}
    so recalling that $G^{(t),j}=(G^{1,j},\dots, G^{t,j})$, we have
    \begin{align*}
    \delta Q^{(t-1), k\to j}&= \delta C^{(t-1),k}\E[\partial_{G^{(t),j}} \sigma^{(t-1),j}]+ \delta \begin{bmatrix}
    0 \\
    A_{0:t-2,0:t-1}^{(t-1),k}
\end{bmatrix}^\top C^{(t-1),j}.
\end{align*}
Now again by Stein's lemma, we get
\begin{align*}
    \delta Q^{(t-1), k\to j}&=  \delta C^{(t-1),k}(\Sigma^{(t),j})^{-1}\E[(G^{(t),j})^\top \sigma^{(t-1),j}]+ \delta \begin{bmatrix}
    0 \\
    A_{0:t-2,0:t-1}^{(t-1),k}
\end{bmatrix}^\top C^{(t-1),j} \\
&=  \delta C^{(t-1),k}(\Sigma^{(t),j})^{-1}R^{(t),j}_{1:t, 0:t-1}+ \delta \begin{bmatrix}
    0 \\
    A_{0:t-2,0:t-1}^{(t-1),k}
\end{bmatrix}^\top C^{(t-1),j}.
    \end{align*}
Similarly, we have
\begin{align*}
    Q^{k\to j}_{0:t-1,t} &= C^{(t-1),k}(\Sigma^{(t),j})^{-1} R^{(t),j}_{1:t,t} + \begin{bmatrix}
        0 \\
        A^{(t-1),k}_{0:t-2,0:t-1}
    \end{bmatrix} C^{(t),j}_{0:t-1,t}.
\end{align*}
This gives
\begin{align*}
    \theta^{t, j\to k} &= (\Sigma^{(t),j})^{-1} R_{1:t,t}^{(t),j} - (\Sigma^{(t),j})^{-1} R_{1:t,0:t-1}^{(t),j} (C^{(t-1),j})^{-1}C_{0:t-1,t}^{(t),j},
\end{align*}
and hence
\begin{align*}
    A^{(t-1)}\psi^{t,j} + \theta^{t,j\to k} &= (\Sigma^{(t),j})^{-1}  R_{1:t,t}^{(t),j} +\left( A^{(t-1)}- (\Sigma^{(t),j})^{-1}R_{1:t,0:t-1}^{(t),j} \right)(C^{(t-1),j})^{-1}C_{0:t-1,t}^{(t),j} \\
    &= A^{(t),j}_{0:t-1,t},
\end{align*}
as desired. 

Next we consider the case $j>k.$ From \eqref{eq:theta-defn}, we have
\begin{align*}
    \theta^{t,j\rightarrow k}&= 
    (C^{(t),k})^{-1} Q^{(t),k \rightarrow j}_{0:t,t}-(C^{(t),k})^{-1} Q^{(t),k \rightarrow j}_{0:t,0:t-1}(C^{(t-1),j})^{-1} C^{(t),j}_{0:t-1,t}.
\end{align*}
Since $j> k,$ by the same arguments as above, we have
    \begin{align*}
        \delta Q_{0:t,0:t-1}^{k\to j} &= \E (g^{(t), k\to j} + \delta \sigma^{(t-1),j}A_{0:t-1,0:t}^{(t),k})^\top\sigma^{(t-1),j} \\
        &= \E (g^{(t), k\to j})^\top \sigma^{(t-1),j} + \delta (A^\top)^{(t),k}_{0:t, 0:t-1}C^{(t-1),j} \\
        &= \delta C^{(t),k} \E[\partial_{g^{(t), k\to j}} \sigma^{(t-1),j}]+ \delta (A^\top)^{(t),k}_{0:t, 0:t-1}C^{(t-1),j} \\
        &= \delta C^{(t),k} \begin{bmatrix}
            0 \\
            \E[\partial_{G^{(t),j}} \sigma^{(t-1),j}]
        \end{bmatrix}+ \delta (A^\top)^{(t),k}_{0:t, 0:t-1}C^{(t-1),j} \\
         &= \delta C^{(t),k} \begin{bmatrix}
            0 \\
            (\Sigma^{(t),j})^{-1}\E[(G^{(t),j})^\top\sigma^{(t-1),j}]
        \end{bmatrix}+ \delta (A^\top)^{(t),k}_{0:t, 0:t-1}C^{(t-1),j}\\
        &= \delta C^{(t),k} \begin{bmatrix}
            0 \\
            (\Sigma^{(t),j})^{-1}R_{1:t,0:t-1}^{(t),j}
        \end{bmatrix}+ \delta (A^\top)^{(t),k}_{0:t, 0:t-1}C^{(t-1),j}
    \end{align*}
    and similarly
    \begin{align*}
        Q_{0:t,t}^{(t),k\to j} &= \delta C^{(t),k} \begin{bmatrix}
            0 \\
            (\Sigma^{(t),j})^{-1}R_{1:t,t}^{(t),j}
        \end{bmatrix}+ \delta (A^\top)^{(t),k}_{0:t, 0:t-1}C^{(t),j}_{0:t-1,t}.
    \end{align*}
    Hence
    \begin{align*}
        \theta^{t,j\rightarrow k}
    &=  \begin{bmatrix}
            0 \\
            (\Sigma^{(t),j})^{-1}R_{1:t,t}^{(t),j}
        \end{bmatrix} - \begin{bmatrix}
            0 \\
            (\Sigma^{(t),j})^{-1}R_{1:t,0:t-1}^{(t),j}
        \end{bmatrix}(C^{(t-1),j})^{-1} C^{(t),j}_{0:t-1,t} \\
    &= \begin{bmatrix}
            0 \\
            (\Sigma^{(t),j})^{-1}R_{1:t,t}^{(t),j} - (\Sigma^{(t),j})^{-1}R_{1:t,0:t-1}^{(t),j}(C^{(t-1),j})^{-1} C^{(t),j}_{0:t-1,t}
        \end{bmatrix},
    \end{align*}
    and we conclude
    \begin{align*}
        \begin{bmatrix}
            0 \\A^{(t-1),j}\psi^{t,j}\end{bmatrix} +\theta^{t,j\rightarrow k} =  \begin{bmatrix}
            0\\ A_{0:t-1,t}^{(t),j}
            \end{bmatrix}, 
    \end{align*}
    as desired.

    The only remaining case is $j=k.$ Here we have 
    \begin{align*}
          \theta^{t,j\rightarrow k}&= 
       (C^{(t-1),k})^{-1} Q^{(t),k \rightarrow j}_{0:t-1,t}-(C^{(t-1),k})^{-1} Q^{(t-1),k \rightarrow j}(C^{(t-1),j})^{-1} C^{(t),j}_{0:t-1,t},
    \end{align*}
    but since $k\geq j,$ we get
    \begin{align*}
     Q^{(t-1), k\to j} &= 
 C^{(t-1),k}(\Sigma^{(t),j})^{-1}R^{(t),j}_{1:t, 0:t-1}+  (A^{(t-1),k})^\top C^{(t-1),j},
    \end{align*}
    and
    \begin{align*}
        Q^{k\to j}_{0:t-1,t} &= C^{(t-1),k}(\Sigma^{(t),j})^{-1} R^{(t),j}_{1:t,t} + (A^{(t-1),k})^\top C^{(t),j}_{0:t-1,t}.
    \end{align*}
    so 
    \begin{align*}
        \theta^{t,j\to k} &= (\Sigma^{(t),j})^{-1} R^{(t),j}_{1:t,t} -  (\Sigma^{(t),j})^{-1} R^{(t),j}_{1:t,0:t-1} (C^{(t-1),j})^{-1}C^{(t),j}_{0:t-1,t},
    \end{align*}
    so from here onward, the situation is the same as in the case $j<k.$ This concludes the proof.
\end{proof}

\bigskip

\subsection{Mapping to the Effective Process}\label{sec:mapdecoup}
In what follows, our goal will be to prove inductively that the quantities $ \hat{b}^{t,j},  \hat{\psi}^{t,j}, \hat{\theta}^{t, j \rightarrow k}$ from \Cref{lem:ind_finite_N} converge to deterministic the limits $b^{t,j}, \psi^{t,j},\theta^{t,j\to k}$ corresponding to the Effective Process from \Cref{def:joint_eff_proc}. Assuming such convergence, we show that the ``approximately decoupled'' process from \Cref{lem:ind_finite_N} converges to i.i.d draws from the Effective Process from \Cref{def:joint_eff_proc}. 

Consider an instance of Algorithm \ref{alg:t-passes} with $U^t_i \in [0,1]$ denoting the random seeds utilized in the update $c$. Recall the vectors $\bfxi^{t,j\to k}$ defined in \Cref{lem:ind_finite_N}, and let $\tilde{\bfh}^{t,j \rightarrow k} , \tilde{\bfh}^{t,j}\in \mathbb{R}^N$ and $\tilde{\bfsigma}^{t|j}$ be defined recursively as follows. For $t=0,$ we define
\begin{align}
\tilde{\bfh}^{0,j\rightarrow k} &= \bfxi^{0,j\rightarrow k}, &&\forall j,k \in [1/\delta]\label{eq:tilh-initial-condition}, \\
\tilde{\bfsigma}^0 &= \bfsigma^0.
\end{align}
Then, recalling the definitions of $\psi^{t,j}, \theta^{t,j\to k}$ and $b^{t,j}$ from equations \eqref{eq:psi-defn}-\eqref{eq:b-defn}, for $t \geq 1$, we inductively define
\begin{equation} \label{eq:tilh-recursion}
\begin{split}
    \tilde{\bfh}^{t,j\rightarrow k}  &= \begin{cases}
        b^{t,j} \sqrt{\delta}\bfxi^{t,j\rightarrow k}+ \tilde{\bfh}^{(t-1),j \rightarrow k} \psi^{t,j}+\delta\tilde{\bfsigma}^{(t)|k}\theta^{t,j\rightarrow k}, & k < j \\
         b^{t,j} \sqrt{\delta}\bfxi^{t,j\rightarrow k}+ \tilde{\bfh}^{(t-1),j \rightarrow k} \psi^{t,j}+\delta\tilde{\bfsigma}^{(t-1)|k}\theta^{t,j\rightarrow k} & k \geq j
    \end{cases}\\
    \tilde{\bfh}^{t,k} &= \sum_{j<k} \tilde{\bfh}^{t,j\rightarrow k} +  \sum_{j\geq k} \tilde{\bfh}^{t-1,j\rightarrow k}\\
    \tilde{\bfsigma}^{t}_i&= 
        c(U^t_i, \tilde{\bfh}^{t,k}_i)\qquad \qquad \forall i\in B(k).
\end{split}
\end{equation}
The following proposition shows that the process defined above is equal in distribution to i.i.d. draws from the lifted Effective Process.

\begin{prop}\label{prop:decop}
 
Suppose for each $k\in [1/\delta]$, the random variables $\{h^{(T),j\to k},h^{(T),j},\sigma^{t,j}\}_{j\in [1/\delta]}$ are given as in \Cref{def:joint_eff_proc}. Then the entries of $\{\tilde{\bfh}_i^{(T),j\to k},\tilde{\bfh}_i^{(T),j},\tilde{\bfsigma}_i^{t|j}\}_{j\in [1/\delta]}$ are i.i.d. across $i \in B(k)$ and we have
\begin{equation}\label{eq:tilh-iid-h}
    \{\tilde{\bfh}_i^{(T),j\to k},\tilde{\bfh}_i^{(T),j},\tilde{\bfsigma}_i^{t|j}\}_{j\in [1/\delta]}\stackrel{d}{=}  \{h^{(T),j\to k},h^{(T),j},\sigma^{t,j}\}_{j\in [1/\delta]} 
\end{equation}
for all $i \in B(k).$
\end{prop}
\begin{proof}
The fact that $\{\tilde{\bfh}_i^{(T),j\to k},\tilde{\bfh}_i^{(T),j},\tilde{\bfsigma}^{t|j}\}_{j\in [1/\delta]}$ is i.i.d. across $i \in B(k)$ is immediate by induction: the initial condition is i.i.d. across $i \in B(k)$ and the recursion \eqref{eq:tilh-recursion} acts independently on each coordinate, since we know from \Cref{lem:ind_finite_N} that the $\bfxi^{t,j\to k}$ are i.i.d. 

Thus we focus on verifying \eqref{eq:tilh-iid-h}. It's clear that it suffices to show that $\tilde{\bfh}_i^{(T),j\to k} \stackrel{d}{=}  h^{(T),j\to k}$ for all $i \in B(k).$ Let $\bfg^{t,j \rightarrow k}$  be defined recursively by
    \begin{align*}
       \tilde{\bfg}^{t,j \rightarrow k} = \sqrt{\delta}\bfxi^{t,j\to k}b^{t,j} + \tilde{\bfg}^{(t-1),j \rightarrow k}\psi^{t,j}.
    \end{align*}
    We now prove simultaneously by induction over $t=1,\dots,T$ the following two claims, from which the lemma follows.
    \begin{itemize}
        \item For $k<j$, we have
        \begin{align*}
                \tilde{\bfh}^{t,j\rightarrow k}  &=
          \tilde{\bfg}^{t,j\to k}+\delta\tilde{\bfsigma}^{(t)|k}\begin{bmatrix}
              0 \\
              A^{(t),j}_{0:t-1,t}
          \end{bmatrix}, 
          \end{align*}
        for $k> j$, we have
          \begin{align*}
         \tilde{\bfh}^{t,j\rightarrow k}  &=\tilde{\bfg}^{t,j\to k}+\delta\tilde{\bfsigma}^{(t-1)|k} A^{(t),j}_{0:t-1,t},
        \end{align*}
        and for $j=k,$ we have 
              \begin{align*}
         \tilde{\bfh}^{t,j\rightarrow k}  &=\tilde{\bfg}^{t,j\to k}+\delta\tilde{\bfsigma}^{(t)|k} A^{(t),j}_{0:t,t}.
        \end{align*}
        \item The vector $\tilde{\bfg}^{t,j\to k}_i$ is i.i.d. across $i\in B(k)$ and equal in distribution to $g^{t,j\to k}$, as defined in \Cref{def:joint_eff_proc}.
    \end{itemize}
    The case $t=0$ is immediate from \eqref{eq:tilh-initial-condition}.
    Assume inductively that this holds for all times less than $t.$ Note that by induction, for $k<j,$ we have
\begin{align*}
    \tilde{\bfh}^{t,j\rightarrow k} &= b^{t,j}\sqrt{\delta}\bfxi^{t,j\to k} + \tilde{\bfg}^{(t-1),j\to k} \psi^{t,j} + \delta  \tilde{\bfsigma}^{(t)|k}\left( \begin{bmatrix}
    0 \\
    A^{(t-1),j}\psi^{t,j} 
\end{bmatrix} + \theta^{t,j\to k}\right) \\
    &= \tilde{\bfg}^{t,j\to k} +  \delta\tilde{\bfsigma}^{(t)|k}\left( \begin{bmatrix}
    0 \\
    A^{(t-1),j}\psi^{t,j} 
\end{bmatrix} + \theta^{t,j\to k}\right) \\
&= \tilde{\bfg}^{t,j\to k} +  \delta\tilde{\bfsigma}^{(t)|k}\begin{bmatrix}
    0 \\
    A^{(t),j}_{0:t-1,t}
\end{bmatrix},
\end{align*}
where the last equality is from \Cref{lem:spin-term-simplifications}. For $k> j,$ we similarly have
\begin{align*}
    \tilde{\bfh}^{t,j\rightarrow k} &= \tilde{\bfg}^{t,j\to k} +  \delta \tilde{\bfsigma}^{(t-1)|k}\left( 
    A^{(t-1),j}\psi^{t,j} 
 + \theta^{t,j\to k}\right) \\
&= \tilde{\bfg}^{t,j\to k} +  \delta\tilde{\bfsigma}^{(t-1)|k}
    A^{(t),j}_{0:t-1,t},
\end{align*}
where we again have used \Cref{lem:spin-term-simplifications}. Finally for $j=k,$ we have 
\begin{align*}
    \tilde{\bfh}^{t,j\rightarrow k} &=  \tilde{\bfg}^{t,j\to k} +  \delta \tilde{\bfsigma}^{(t-1)|k}\left( 
    A^{(t-1),j}\psi^{t,j} 
 + \theta^{t,j\to k}\right) \\
&= \tilde{\bfg}^{t,j\to k} +  \delta\tilde{\bfsigma}^{(t-1)|k}
    A^{(t),j}_{0:t-1,t} \\
    &= \tilde{\bfg}^{t,j\to k} +  \delta\tilde{\bfsigma}^{(t)|k}
    A^{(t),j}_{0:t,t},
\end{align*}
where in the last step we have used that $A$ has zeros on the diagonal. To conclude the inductive step, it remains to check that, for each $i\in B(k),$ $\tilde{\bfg}^{(t),j \rightarrow k}_i $ is equal in distribution to $g^{(t),j \rightarrow k}.$ The joint gaussianity of $\tilde{\bfg}_i^{(t), j\to k}$ is immediate from the definition, so it suffices to check that the covariance of $\tilde{\bfg}^{(t),j \rightarrow k}_i$ is $C^{(t),j}$.  By inductive hypothesis, for $s=0,\dots,t-1,$ we have
    \begin{align*}
        \E[\tilde{\bfg}^{s,j \rightarrow k}_i \tilde{\bfg}^{t,j \rightarrow k}_i] &= \E [\tilde{\bfg}_i^{s,j\to k} \tilde{\bfg}_i^{(t-1),j\to k}](C^{(t-1),j})^{-1} C_{0:t-1,t}^{(t),j} \\
        &= C_{s, 0:t-1}^{(t),j} (C^{(t-1),j})^{-1}C_{0:t-1,t}^{(t),j} \\
        &= C_{s,t}^{(t),j},
    \end{align*}
    and moreover 
    \begin{align*}
        \Var(\tilde{\bfg}_i^{t,j\to k}) &= (b^{t,j})^2 +  C_{t, 0:t-1}^{(t),j}(C^{(t-1),j})^{-1}\E[(\tilde{\bfg}_i^{(t-1),j\to k})^\top\tilde{\bfg}_i^{(t-1),j\to k}](C^{(t-1),j})^{-1}C_{0:t-1, t}^{(t),j} \\
        &= (b^{t,j})^2 +  C_{t, 0:t-1}^{(t),j}(C^{(t-1),j})^{-1}C_{0:t-1, t}^{(t),j} \\
        &= 1,
    \end{align*}
    as desired. This completes the proof.
\end{proof}

\subsection{Concentration of order parameters}\label{sec:mainthm}

Next, we present our main technical result which couples $ \bfh^{(t),j \rightarrow k},\bfh^{(t),j}$ with independent draws from the Effective Process defined by \Cref{def:joint_eff_proc}. The purpose of such a coupling is two-fold. First, it implies the asymptotic equivalence of the empirical measures, and second, the concentration of all relevant order parameters.

\begin{thm}\label{thm:finite_del_ind}
 For $t \in \mathbb{N}, j, k \in [1/\delta]$, let $\tilde{\bfh}^{t,j\rightarrow k}, \tilde{\bfh}^{(t)},\tilde{\bfsigma}^{t|j}:= D^j\tilde{\bfsigma}^{t}$ be defined  as in \Cref{prop:decop}. Then, for any $\delta \in (0,1]$, there exist constants $B_\ell>0$ such that the following holds for all $t \in \mathbb{N}$, $j \in [1/\delta]$:
    \begin{enumerate}
    \item
We have
\begin{equation}\label{eq:h_eq}
       \bfh^{(t),j}\stackrel{P}{\simeq}\tilde{\bfh}^{(t),j}
    \end{equation}
    and
    \begin{align*}
        \frac{1}{\delta N} \norm{{\bfsigma^{t|j}}-\tilde{\bfsigma}^{t|j}}_0 \xrightarrow[N \rightarrow \infty]{P} 0,
    \end{align*}
    where $\|\|_0$ denotes the Hamming distance.
\item 
For all $\ell\geq 1$, with high probability,
\begin{align*}
    \frac{1}{N}\norm{\bfh^{t,j}}_\ell^\ell \leq B_\ell.
\end{align*}
\item 
     Recall $C^{(t),j},Q^{(t), j \rightarrow k}, \hat{C}^{(t),j},\hat{Q}^{(t)} $, defined by equations \eqref{eq:finite_del-dim}, 
     \eqref{eq:defQt}, \eqref{eq:hatc},   \eqref{eq:hatq}, respectively. We have
    \begin{align*}
        \norm{\hat{C}^{(t),j}-C^{(t),j}}_F \xrightarrow[N \rightarrow \infty]{P} 0
    \end{align*}
    and, for all $ k < j$,
    \begin{align*}
 \norm{\hat{Q}^{(t), k \rightarrow j}- Q^{(t), k \rightarrow j}}_F \xrightarrow[N \rightarrow \infty]{P} 0,
    \end{align*}
    while for all $k \geq j$,
    \begin{align*}
 \norm{\hat{Q}^{(t),k \rightarrow j}_{0:t-1,0:t}- Q^{(t),k \rightarrow j}_{0:t-1,0:t}}_F \xrightarrow[N \rightarrow \infty]{P} 0.
    \end{align*}
    As a consequence, $\hat{b},\hat{\psi}$ and $\hat{\theta}$ from \Cref{lem:ind_finite_N} converge to $b,\psi$ and $\theta$ from \Cref{sec:field-contributions}, respectively.
\item 
For all $k \in [1/\delta]$,
\begin{align*}
       \bfh^{(t),j \rightarrow k} \stackrel{P}{\simeq}   \tilde{\bfh}^{(t),j \rightarrow k},
    \end{align*}
Furthermore, for all $\ell \geq 1$, with high probability,
\begin{align*}
    \frac{1}{N}\norm{ \bfh^{t,j \rightarrow k} }_\ell^\ell \leq B_\ell.
\end{align*}
\end{enumerate}
\end{thm}


\subsection{Proof of \Cref{thm:finite_del_ind}}\label{sec:main_induct}

In what follows, it will be convenient to define the following sequence of parameters:
\begin{align}
    \tilde{C}^{(t),j}&= \frac{1}{\delta N} (\tilde{\bfsigma}^{(t)\vert j})^\top \tilde{\bfsigma}^{(t)\vert j} \in \mathbb{R}^{(t+1)\times (t+1)}\label{eq:tildec} \\
    \tilde{Q}^{(t),k \rightarrow j} &= \frac{1}{\delta^2 N} (\tilde{\bfh}^{(t),k \rightarrow j})^\top \tilde{\bfsigma}^{(t)\vert j}\in \mathbb{R}^{(t+1)\times (t+1)}.\label{eq:tildeq}
\end{align}

\subsubsection{Initial conditions $t=0$}

Before proceeding with the first pass $t=1$, we show the convergence of order-parameters at $t=0$. Note that for all $j \in [1/\delta]$, we have trivially that $\hat{C}^{(0),j}=1$. 
 We next consider the terms $  \hat{Q}^{(0),k \rightarrow j}$ for $k<j$. Since $\tilde{\bfsigma}^{0|j} ={\bfsigma}^{0|j}$, by Cauchy-Schwarz we have
\begin{align*}
    |\hat{Q}^{(0),k \rightarrow j}-\tilde{Q}^{(0),k\rightarrow j}| \leq \frac{1}{\delta^2 N}\norm{\bfh^{0,k\rightarrow j}-\tilde{\bfh}^{0,k\rightarrow j}}_2 \norm{\bfsigma^{0|j}}_2.
\end{align*}
Now, \Cref{lem:ind_finite_N} implies that $\frac{1}{\sqrt{N}}\norm{\bfh^{0,j\rightarrow k}-\tilde{\bfh}^{0,j\rightarrow k}}_2 \xrightarrow[N \rightarrow \infty]{P} 0$ while $\frac{1}{\sqrt{N}} \norm{\bfsigma^{0|k}}_2$ is bounded by definition. Hence
\begin{equation}\label{eq:0coup}
    |\hat{Q}^{(0),k \rightarrow j} - \tilde{Q}^{(0),k \rightarrow j}|\xrightarrow[N \rightarrow \infty]{P} 0,
\end{equation}
and by the law of large numbers, we have
\begin{equation}\label{eq:0conc}
    \tilde{Q}^{(0),k \rightarrow j} \xrightarrow[N \rightarrow \infty]{P} Q^{(0), k \rightarrow j} = 0,
\end{equation}
where $Q^{(0), k \rightarrow j} = 0$ follows by the definition in \eqref{eq:defQt}. Equations \eqref{eq:0coup}, \eqref{eq:0conc} together imply that
\begin{equation}\label{eq:q_conv_init}
     \hat{Q}^{(0),k \rightarrow j} \xrightarrow[N \rightarrow \infty]{P} Q^{(0), k \rightarrow j} = 0.
\end{equation}

\subsubsection{Base case $t=1,j=1$}

We begin by proving claims $(1)$ and $(2)$ of the inductive statement. By \eqref{eq:rech0j} in \Cref{lem:ind_finite_N}, we have that
\begin{equation}\label{eq:field_app}
\bfh^{0,k\rightarrow 1} = \sqrt{\delta}\bfxi^{0,k\rightarrow 1} + \bfsigma^{0|1}\Delta^{0,k\to 1}
\end{equation}
where $\Delta^{0,k \rightarrow 1} \in \mathbb{R}$ is defined as in \Cref{lem:ind_finite_N}. The field decomposition in \eqref{eq:field_decomp} gives
\begin{align*}
    \bfh^{1,1} = \sum_{k \in [1/\delta]} \bfh^{0,k\rightarrow 1},
\end{align*}
and from \Cref{prop:decop}, we have
\begin{align*}
    \tilde{\bfh}^{1,1} = \sum_{k \in [1/\delta]} \tilde{\bfh}^{0,k\rightarrow 1}.
\end{align*}
Combining the above yields
    \begin{align*}
    \bfh^{1,1}=\tilde{\bfh}^{1,1} + \bfsigma^{0|1}\sum_{k\in [1/\delta]}\Delta^{0,k \rightarrow 1}.
\end{align*}
Since $\frac{1}{\sqrt{N}}\norm{\bfsigma^{0|1}}$ is bounded while $\Delta^{0,k \rightarrow 1} \xrightarrow[N \rightarrow \infty]{P} 0$, we obtain
\begin{equation}\label{eq:eqh0}
     \bfh^{1,1} \stackrel{P}{\simeq} \tilde{\bfh}^{1,1}.
\end{equation}
Then \Cref{lem:spins} implies that
\begin{equation}\label{eq:spin0}
        \frac{1}{\delta N} \norm{\bfsigma^{1|1}-\tilde{\bfsigma}^{1|1}}_0 \xrightarrow[N \rightarrow \infty]{P} 0.
    \end{equation}
Similarly, for any $\ell \in \mathbb{N}$,
\begin{equation}\label{eq:mom_bound_0}
    \frac{1}{N}\norm{\bfh^{1,1}}_\ell^\ell \leq  \frac{2^{\ell-1}}{N }\norm{\tilde{\bfh}^{1,1}}_\ell^\ell+ \frac{2^{\ell-1}}{N}\norm{\bfsigma^{0|1}}_\ell^\ell\left(\sum_{k\in [1/\delta]}\Delta^{0,k\rightarrow 1}\right)^\ell,
\end{equation}
so boundedness of entries of $\bfsigma^{0|1}$ and the fact that $\norm{\Delta^{0,k \rightarrow 1}} \xrightarrow[N \rightarrow \infty]{P} 0$ imply that $\frac{1}{N}\norm{\bfh^{1,1}}_\ell^\ell $ is with high probability bounded by a constant by the law of large numbers. This proves claims $(1)$ and $(2)$ of the inductive statement.

We next move to claim $(3)$, i.e., convergence of $\hat{C}^{(1),1}$. The entries $\hat{C}^{(1),1}_{0,0}$ and $\hat{C}^{(1),1}_{1,1}$ are fixed to $1$ by definition. Moreover, we have by Cauchy-Schwarz
\begin{align*}
|\hat{C}^{(1),1}_{0,1} - \tilde{C}^{(1),1}_{0,1}| \leq \frac{1}{\delta N}\norm{\bfsigma^{1|1}-\tilde{\bfsigma}^{1|1}}_2 \norm{\bfsigma^{0|1}}_2, 
\end{align*}
so \eqref{eq:spin0} and the concentration of $\tilde{C}^{(1),1}$ to $C^{(1),1}$ then implies
\begin{align*}
        \norm{\hat{C}^{(1),1}-C^{(1),1}}_F \xrightarrow[N \rightarrow \infty]{P} 0.
\end{align*}
Then for $k\geq 1,$ we have
\begin{align*}
    |\hat{Q}^{(1),k \rightarrow 1}_{0,1}- \tilde{Q}^{(1),k \rightarrow 1}_{0,1}| &\leq \frac{1}{\delta^2 N} \norm{\tilde{\bfsigma}^{1|1}}_2 \norm{\tilde{\bfh}^{0, k\to 1} - \hat{\bfh}^{0, k\to 1}} + \frac{1}{\delta^2 N} \norm{\tilde{\bfsigma}^{1|1} - \hat{\bfsigma}^{1|1}}_2 \norm{\hat{\bfh}^{0, k\to 1}}  \\
    &\xrightarrow[N \rightarrow \infty]{P} 0
\end{align*}
by \Cref{lem:ind_finite_N} and \eqref{eq:spin0}. Moreover, 
\[
|\hat{Q}^{(1),k \rightarrow 1}_{0,0}- \tilde{Q}^{(1),k \rightarrow 1}_{0,0}| \xrightarrow[N \rightarrow \infty]{P} 0
\]
by induction. Finally, to prove claim $(4)$, recall that, $\hat{b}^{1,0},\hat{\psi}^{1,0}, \hat{\theta}^{1,0 \rightarrow k}$ are given by
\begin{align*}
    \hat{b}^{1,1} &= \sqrt{1- \hat{C}^{(1),1}_{0,1}(C^{(0),1})^{-1} C^{(1),1}_{0,1}}\\
   \hat{\psi}^{1,1}&=(\hat{C}^{(0),1})^{-1} \hat{C}^{(1),1}_{0,1}\\
   \hat{\theta}^{1,1\rightarrow k}&= 
(\hat{C}^{(1),1})^{\dagger} \hat{Q}^{(0),1 \rightarrow 0}-(\hat{C}^{(1),1})^{\dagger} \hat{Q}^{(0),k \rightarrow 1}(\hat{C}^{(0),1})^{\dagger} \hat{C}^{(0),1}_{0,1}.
\end{align*}
By claim $(3)$, \eqref{eq:q_conv_init}  and \Cref{lem:add_mult}, we obtain:
\begin{equation}\label{eq:conv0param}
    \begin{split}
    \hat{b}^{1,1} &\xrightarrow[N \rightarrow \infty]{P} b^{1,1}\\
    \hat{\psi}^{1,1}  &\xrightarrow[N \rightarrow \infty]{P} \psi^{1,1}\\
    \hat{\theta}^{1,1 \rightarrow k} &\xrightarrow[N \rightarrow \infty]{P}\theta^{1,1 \rightarrow k}.
    \end{split}
\end{equation}
Now, \Cref{lem:decomp} implies that
\begin{equation}\label{eq:rechtj0}
\bfh^{1,1\rightarrow k} 
=\hat{b}^{1,1}\sqrt{\delta}\bfxi^{1,1\rightarrow k}+\bfh^{(0), 1 \rightarrow k}\hat{\psi}^{1,0}+\delta \bfsigma^{(0)|k} \hat{\theta}^{1,1\to k}+\bfsigma^{(0)|k}(\Delta^{1, 1 \rightarrow k}).
\end{equation}
Therefore, the convergence of $\hat{b}^{1,1},\hat{\psi}^{1,1},\hat{\theta}^{1,1\to k}$ in \eqref{eq:conv0param} along with $\bfh^{(0), 1 \rightarrow k}  \stackrel{P}{\simeq}\tilde{\bfh}^{(0), 1 \rightarrow k}$ and $\bfsigma^{(0)|k}=\tilde{\bfsigma}^{(0)|k}$ imply
\begin{align*}
\bfh^{1,1\rightarrow k} \stackrel{P}{\simeq} \tilde{\bfh}^{1,1\rightarrow k}. 
\end{align*}
To additionally bound the higher-moments of $\bfh^{1,1\rightarrow k}$, we note that by \eqref{eq:rech0j} $\bfh^{1,1\rightarrow k}-\tilde{\bfh}^{1,1\rightarrow k}$ can be explicitly expressed as
\begin{equation}\label{eq:h0decomp}
\begin{split}
\bfh^{1,1\rightarrow k}-\tilde{\bfh}^{1,1\rightarrow k} &= (\hat{b}^{1,1}-b^{1,1})\sqrt{\delta}\bfxi^{1,1\rightarrow k}+\bfh^{(0), 1 \rightarrow k}(\hat{\psi}^{1,1}-\psi^{1,1}) + (\bfh^{(0), 1 \rightarrow k}-\tilde{\bfh}^{(0), 1 \rightarrow k})\psi^{1,1}\\ 
&\qquad +\delta \bfsigma^{(0)|k} (\hat{\theta}^{1,1}-\theta^{1,1})+\bfsigma^{(0)|k}(\Delta^{1, 1 \rightarrow k}).
\end{split}
\end{equation}
Similar to \eqref{eq:mom_bound_0}, applying the convexity of $x \rightarrow x^\ell$ to the above decomposition yields bounds on the higher-moments of $\bfh^{1,1\rightarrow k}$.

\subsubsection{Induction}
Suppose claims $(1),(2),(3),(4)$ hold for all $(\tilde{t},\tilde{j}) < (t,j)$ under the lexicographical ordering. By inductive claim $(4)$, we have for all $k<j$,
\begin{equation}\label{eq:rechtkj}
\bfh^{t,k\rightarrow j}  \stackrel{P}{\simeq} \tilde{\bfh}^{t,k\rightarrow j}
\end{equation}
and for $k \geq j$,
\begin{equation}\label{eq:rechtjk}
    \bfh^{t-1,k\rightarrow j} 
    \stackrel{P}{\simeq}\tilde{\bfh}^{t-1,k\rightarrow j}.
\end{equation}
Recall the field decomposition given by \eqref{eq:field_decomp},
\begin{align*}
    \bfh^{t,j} = \sum_{k<j} \bfh^{t,k \rightarrow j}  + \sum_{k\geq j} \bfh^{t-1,k \rightarrow j}.
\end{align*}
Substituting equations \eqref{eq:rechtkj}, \eqref{eq:rechtjk} into the above field decomposition yields
\begin{align*}
  \bfh^{t,j} \stackrel{P}{\simeq} \tilde{\bfh}^{t,j}.
\end{align*}
\Cref{lem:spins} then implies that
\begin{equation}\label{eq:spin0j}
        \frac{1}{\delta N} \norm{\bfsigma^{t|j}-\tilde{\bfsigma}^{t|j}}_0 \xrightarrow[N \rightarrow \infty]{P} 0,
    \end{equation}
proving claim $(1)$. Claim $(2)$ follows similarly, noting that
\begin{align*}
     \frac{1}{N}\norm{\bfh^{t,j}}_\ell^\ell \leq \frac{1}{N\delta^{\ell-1}}  \sum_{k<j} \norm{\bfh^{t,k \rightarrow j}}_\ell^\ell + \frac{1}{N\delta^{\ell-1}} \sum_{k\geq j} \norm{\bfh^{t-1,k \rightarrow j}}_\ell^\ell.
\end{align*}
By Cauchy-Schwarz inequality, we have, for any $s < t$,
\begin{equation}\label{eq:hatcconv}
|\hat{C}^{(t),j}_{t,s} - \tilde{C}^{(t),j}_{t,s}| \leq \frac{1}{\delta N}\norm{\bfsigma^{s|j}-\tilde{\bfsigma}^{s|j}}_2 \norm{\bfsigma^{t|j}}+ \frac{1}{\delta N}\norm{\bfsigma^{t|j}-\tilde{\bfsigma}^{t|j}}_2 \norm{\bfsigma^{s|j}}.
\end{equation}
By \eqref{eq:spin0j} and induction claim $(1)$, we obtain that the right hand side converges to $0$ in probability as $N \rightarrow \infty$. By the law of large numbers, $\tilde{C}^{(t),j}_{t,s} \xrightarrow[N \rightarrow \infty]{P} C^{(t),j}_{s,t}$. Combined with \eqref{eq:hatcconv}, and an identical argument for $\hat{Q}$, we obtain claim $(3)$. To obtain claim $(4)$, recall that, according to \Cref{lem:ind_finite_N}, we have
\begin{equation}
\bfh^{t,j\rightarrow k}  = \begin{cases}
        \hat{b}^{t,j} \sqrt{\delta}\bfxi^{t,j\rightarrow k}+ \bfh^{(t-1),j \rightarrow k} \hat{\psi}^{t,j}+\delta\bfsigma^{(t)|k}\hat{\theta}^{t,j\rightarrow k} +   \bfsigma^{(t)|k}\Delta^{t,j \rightarrow k}, & k < j \\
         \hat{b}^{t,j} \sqrt{\delta}\bfxi^{t,j\rightarrow k}+ \bfh^{(t-1),j \rightarrow k} \hat{\psi}^{t,j}+\delta\bfsigma^{(t-1)|k}\hat{\theta}^{t,j\rightarrow k}+   \bfsigma^{(t-1)|k}\Delta^{t,j \rightarrow k}, & k \geq j.
    \end{cases}
\end{equation}
Claim $(3)$ further implies
\begin{align*}
\hat{C}^{(t),j}\xrightarrow[N \rightarrow \infty]{P} C^{(t),j},
\end{align*}
 Since $C^{(t),j} \succ 0$ by \Cref{prop:finite-delta-regularity}, we obtain that there exists a constant $\kappa > 0$ such that
\begin{align*}
    \Pr[\lambda_{\rm min}(\hat{C}^{(t),j}) > \kappa] \rightarrow 1.
\end{align*}
Since $\hat{b}^{t,j},  \hat{\psi}^{t,j},\hat{\theta}^{t,j \rightarrow k}$ are continuous in $\hat{C}^{(t),j},\hat{Q}^{(t),j \rightarrow k}$ over $\hat{C}^{t,j} \succ \kappa I$, we obtain, by the continuous mapping theorem, for all $k \in [1/\delta]$,
\begin{align*}
\hat{b}^{t,j}&\stackrel{P}{\simeq} b^{t,j}\\
    \hat{\psi}^{t,j} &\stackrel{P}{\simeq } \psi^{t,j} \\
     \hat{\theta}^{t,j \rightarrow k} &\stackrel{P}{\simeq } \theta^{t,j \rightarrow k}.
\end{align*}
Therefore, applying \Cref{lem:add_mult}, \eqref{eq:rechtj}, and the induction claim $(1)$, we have that, for any $k < j$,
\begin{align*}
    \hat{b}^{t,j} \sqrt{\delta}\bfxi^{t,j\rightarrow k}+ \bfh^{(t-1),j \rightarrow k} \hat{\psi}^{t,j}+\delta\bfsigma^{(t)|k}\hat{\theta}^{t,j\rightarrow k} \stackrel{P}{\simeq} b^{t,j} \sqrt{\delta}\bfxi^{t,j\rightarrow k}+ \tilde{\bfh}^{(t-1),j \rightarrow k} \psi^{t,j}+\delta\tilde{\bfsigma}^{(t)|k}\theta^{t,j\rightarrow k},
\end{align*}
and for $k \geq j$,
\begin{align*}
    \hat{b}^{t,j} \sqrt{\delta}\bfxi^{t,j\rightarrow k}+ \bfh^{(t-1),j \rightarrow k} \hat{\psi}^{t,j}+\delta\bfsigma^{(t-1)|k}\hat{\theta}^{t,j\rightarrow k} \stackrel{P}{\simeq} b^{t,j} \sqrt{\delta}\bfxi^{t,j\rightarrow k}+ \tilde{\bfh}^{(t-1),j \rightarrow k} \psi^{t,j}+\delta\tilde{\bfsigma}^{(t-1)|k}\theta^{t,j\rightarrow k},
\end{align*}
\Cref{prop:decop} and inductive claim $(1)$ then imply that, in either case,
\begin{align*}
    \bfh^{(t),j \rightarrow k} \simeq \tilde{\bfh}^{(t),j \rightarrow k}.
\end{align*}
The higher-moments are subsequently bounded using the decomposition similar to \eqref{eq:h0decomp}.

\section{Convergence as $\delta\to 0$}\label{sec:delta-to-0-limit}
In this section we show that, as $\delta \to 0,$ the equations \eqref{eq:finite_del-dim} converge, in a sense which we now specify, to the system \eqref{eq:Sigma}-\eqref{eq:A-upper-diag}. To formulate the convergence, we begin by defining, for each $\delta > 0,$ the functions 
\begin{align}\label{eq:finite-delta-interpolations}
\begin{split}
    \Sigma^{(T), \delta}&:[0,1]\to \R^{T\times T} \\
    f^{t, \delta}&:[0,1]\to \R^{t} \qquad \qquad \qquad \qquad  t=1,\dots,T 
\end{split}
\end{align}
as the linear interpolation in $[0,1]$ of the functions given by \eqref{eq:finite_del-dim}. Specifically, for each $j=0,\dots,\lfloor \frac{1}{\delta}\rfloor$, we define
\begin{align*}
    \Sigma^{(T), \delta}(j\delta) = \Sigma^{(T),j},
\end{align*}
and if $x =  \alpha j\delta + (1-\alpha) (j+1)\delta$ for some $j\in \{0,\dots,\lfloor \frac{1}{\delta}\rfloor\}$ and $\alpha \in [0,1],$ we define $\Sigma^{(T), \delta}(x) = \alpha \Sigma^{(T),j} + (1-\alpha)\Sigma^{(T),j+1}$, where we set $\Sigma^{(T),\lfloor \frac{1}{\delta}\rfloor + 1}:= \Sigma^{(T),\lfloor \frac{1}{\delta}\rfloor}$. We define $f^{t,\delta}$ similarly.

The goal of this section is to prove that the functions \eqref{eq:finite-delta-interpolations} converge uniformly as $\delta \to 0$ to the unique solution of the system \eqref{eq:Sigma}-\eqref{eq:A-upper-diag}. The proof will proceed in two steps. First, we will show via the Arzelà-Ascoli Theorem that \eqref{eq:finite-delta-interpolations}, viewed as function families parameterized by $\delta \in (0,1),$ are pre-compact. Then, we will verify that any limit point (in the uniform topology) along a subsequence $\delta_n\to 0$ must satisfy \eqref{eq:Sigma}-\eqref{eq:A-upper-diag}, which uniquely determines it. This will conclude the proof.

We begin by recalling the Arzelà-Ascoli Theorem.
\begin{thm}[Arzelà-Ascoli]\label{thm:arzela-ascoli}
    Let $\mathscr{C}[0,1]$ be the space of continuous real-valued functions in $[0,1]$ equipped with the supremum norm topology, and let $\Phi \subseteq \mathscr{C}[0,1].$ Suppose $\Phi$ satisfies both of the following:
    \begin{enumerate}
        \item \emph{Equicontinuity.} For all $\eps>0,$ there exists $\eps'>0$ such that for all $\phi \in \Phi$ and $x,y\in [0,1]$ such that $|x-y|<\eps',$ we have $|\phi(x)-\phi(y)|<\eps.$
        \item \emph{Uniform boundedness.} There exists $B>0$ such that $|\phi(x)|\leq B$ for all $x\in [0,1]$ and $\phi \in \Phi.$
    \end{enumerate}
    Then $\Phi$ is pre-compact, i.e., its closure in $\mathscr{C}[0,1]$ is compact.
\end{thm}
Let us verify that our function family \eqref{eq:finite-delta-interpolations} is equicontinuous and uniformly-bounded. We now invoke the main technical result of this section, whose proof we postpone to the appendix.
\begin{restatable}{prop}{deltaRegularityProp}
\label{prop:finite-delta-regularity}
    For each $\delta>0,$ consider the solution to the difference equations \eqref{eq:finite_del-dim}. Define
    \begin{align*}
    \lambda^t(j) &:= \lambda_{\min}\left(\Sigma^{(t),j}\right) \\
    \eta^t(j) &:= \lambda_{\min}\left(C^{(t),j}\right) \\
    a^t(j) &:= \norm{A_{0:t-1,t}^{(t),j}}_2.
\end{align*}
There exist constants $B_t > 0, t=1,\dots,T$, independent of $\delta,$ such that 
\begin{align*}
    \max_{j\in [1/\delta]} \norm{f^{t,j}}_2 &\leq B_t \\
    \max_{j\in [1/\delta]} a^t(j) &\leq B_t \\
    \min_{j\in [1/\delta]} \lambda^t(j) &\geq 1/B_t \\
    \min_{j\in [1/\delta]} \eta^t(j) &\geq 1/B_t.
\end{align*}
\end{restatable}

We immediately conclude the following.

\begin{cor}\label{cor:precompact}
    The function family \eqref{eq:finite-delta-interpolations} is coordinatewise equicontinuous and uniformly bounded.
\end{cor}
\begin{proof}
   It suffices to show that there exists a constant $B=B(T)>0$ independent of $\delta$ such that, in the system \eqref{eq:finite_del-dim}, for all $j\in [1/\delta]$ we have
\begin{align}
    \norm{\Sigma^{(T),j} - \Sigma^{(T),j-1}}_F &\leq B\delta \label{eq:Sigma-bound} \\
    \norm{f^{t,j} - f^{t,j-1}}_2 &\leq B\delta && t=1,\dots,T.\label{eq:f-diff-bound} \\
    \norm{f^{t,j}}_2 &\leq B && t=1,\dots,T.\label{eq:f-bound-needed}
\end{align}
Indeed, if we had this, then uniform boundedness for $\Sigma^{(T),\delta}$ follows from the fact that
\begin{align*}
    \sup_x \norm{\Sigma^{(T),\delta}(x)}_F &= \sup_j \norm{\Sigma^{(T),j}}_F \\
    &\leq \delta \sum_{k<j} \norm{C_{1:T,1:T}^{(T),k}}_F + \delta \sum_{k\geq j}\norm{ C_{0:T-1,0:T-1}^{(T-1),k}}_F \\
    &\leq T \max_{s,t}|C^{(T),j-1}_{s,t}| \\
    &\leq T\max_{s,t}|\E \sigma^{s,j-1}\sigma^{t,j-1}| \\
    &\leq T.
\end{align*}
Similarly, equicontinuity for $f^{t,\delta}$ is proved as follows
\begin{align*}
    \sup_{x\neq y \in [0,1]} \frac{\norm{f^{t,\delta}(x) - f^{t,\delta}(y)}_2}{|x-y|} &= \sup_{j> k, \; j,k\in [1/\delta]} \frac{\norm{f^{t,j} - f^{t,k}}_2}{\delta(j-k)} \\
    &= \sup_{j> k, \; j,k\in [1/\delta]}  \frac{\norm{\sum_{\ell = k+1}^j f^{t,\ell} - f^{t,\ell-1}}_2}{\delta(j-k)} \\
    &\leq \frac{B\delta (j-k)}{\delta(j-k)} \\
    &\leq B,
\end{align*}
and the same argument shows equicontinuity of $\Sigma^{(T),\delta}$ as well. Hence, it remains to check \eqref{eq:Sigma-bound}-\eqref{eq:f-bound-needed}.
Note first that \eqref{eq:Sigma-bound} is immediate since
\begin{align*}
    \norm{\Sigma^{(T),j} - \Sigma^{(T),j-1}}_F &= \delta \norm{C_{1:T,1:T}^{(T),j-1} - C_{0:T-1,0:T-1}^{(T-1),j-1}}_F \\
    &\leq 2\delta \sqrt{T} \max_{s,t}|C^{(T),j-1}_{s,t}| \\
    &\leq 2\delta \sqrt{T}.
\end{align*}
For \eqref{eq:f-diff-bound}, we have
\begin{align*}
    \norm{f^{t,j} - f^{t,j-1}}_2 &\leq \delta \left(\norm{A_{0:t-1,t}^{(t),j-1}}_2 + \norm{A_{0:t-2,t-1}^{(t-1),j-1}}_2\right),
\end{align*}
so the result follows from \Cref{prop:finite-delta-regularity}, as does \eqref{eq:f-bound-needed}. This concludes the proof.
\end{proof}
From this we conclude the main result of this section.
\begin{thm}\label{thm:convergence-of-delta-to-0}
    The system \eqref{eq:Sigma}-\eqref{eq:A-upper-diag} has a unique global smooth solution $\Sigma^{(T)}, f^t,t=1,\dots,T$ in $[0,1].$ Moreover, the functions \eqref{eq:finite-delta-interpolations} converge uniformly as $\delta \to 0$ to $\Sigma^{(T)}, f^t,t=1,\dots,T$.
\end{thm}
\begin{proof}
    By \Cref{thm:arzela-ascoli} and \Cref{cor:precompact}, the function family \eqref{eq:finite-delta-interpolations} parameterized by $\delta$ is pre-compact. Let $\delta_k$ be a sequence such that $\delta_k\to 0$ and such that $\Sigma^{(T), \delta_k}, f^{t,\delta_k}, t=1,\dots,T,$ converge uniformly in $[0,1]$ as $k\to\infty$ to a limit  $\Sigma^{(T)}, f^{t}, t=1,\dots,T.$ It is immediate from uniform convergence that this limit must be continuous and a global solution in $[0,1]$ to \eqref{eq:Sigma}-\eqref{eq:A-upper-diag}. To see that it is the unique continuous solution, suppose there was an alternate continuous solution $\tilde{\Sigma}^{(T)}, \tilde{f}^t,t=1,\dots,T.$ The Cauchy-Lipschitz Theorem implies that the set $\{x\in [0,1]:\tilde{\Sigma}^{(T)}(x)=\Sigma^{(T)}, \tilde{f}^t(x)=f^t(x),t=1,\dots,T\}$ is open and nonempty, but continuity implies that it's closed; hence it must be equal to $[0,1].$ Since the solution of \eqref{eq:Sigma}-\eqref{eq:A-upper-diag} exists and is unique, it must also be smooth since the right-hand sides of \eqref{eq:Sigma}-\eqref{eq:A-upper-diag} are smooth functions of $\Sigma^{(T)}, f^t,t=1,\dots,T$ in a neighborhood of the solution (since $\Sigma$ and $C$ are invertible in this neighborhood by \Cref{prop:finite-delta-regularity}).
    
    To conclude the proof, let $\delta_k$ be any sequence such that $\delta_k\to 0$ as $k\to \infty.$ We must show that $\Sigma^{(T), \delta_k}, f^{t,\delta_k}, t=1,\dots,T,$ converge uniformly in $[0,1]$ as $k\to\infty$ to $\Sigma^{(T)}, f^{t}, t=1,\dots,T.$ We have shown above that if such a sequence converges uniformly at all, then it must converge to $\Sigma^{(T)}, f^{t}, t=1,\dots,T.$ Now suppose for contradiction that $\Sigma^{(T), \delta_k}, f^{t,\delta_k}, t=1,\dots,T$ did not converge uniformly to $\Sigma^{(T)}, f^{t}, t=1,\dots,T.$ By pre-compactness, $\Sigma^{(T), \delta_k}, f^{t,\delta_k}, t=1,\dots,T$ must have a uniformly convergent subsequence $\Sigma^{(T), \delta_{k_n}}, f^{t,\delta_{k_n}}, t=1,\dots,T$ with a limit different from $\Sigma^{(T)}, f^{t}, t=1,\dots,T.$ But this is a contradiction, concluding the proof.
\end{proof}

\section{Proof of \Cref{thm:main}}\label{sec:main-thm-proof}
\begin{proof}[Proof of \Cref{thm:main}]
    In \Cref{thm:finite_del_ind} we have shown that $\bfh^{(t),j} \pequiv \wt{\bfh}^{(t),j}$ and 
    \[
    \frac{1}{\delta N}\norm{\bfsigma^{t|j} - \tilde{\bfsigma}^{t|j}} \xrightarrow[N \rightarrow \infty]{P}0.
    \]
    By \Cref{lem:p-conv-empirical-conv}, this implies that for every $j \in [1/\delta]$ and pseudo-Lipschitz function $\varphi,$ we have 
    \begin{align*}
        \frac{1}{\delta N}\sum_{i=j\delta N}^{(j+1)\delta N - 1}\varphi(\bfsigma^{(T)}_i, \bfh^{(T)}_i)   \xrightarrow[N \rightarrow \infty]{P}\E[\varphi(\sigma^{(T),j}, h^{(T),j})],
    \end{align*}
    where the expectation on the right hand side is with respect to the law specified by the difference equations \eqref{eq:h-joint-law-del}, \eqref{eq:finite_del-dim}. We have shown in \Cref{thm:convergence-of-delta-to-0} that the linear interpolation of the solution of the equations \eqref{eq:h-joint-law-del}, \eqref{eq:finite_del-dim} in $[0,1]$ converges uniformly to the unique global smooth solution to the system \eqref{eq:Sigma}-\eqref{eq:A-upper-diag}. Hence, if $\delta j \to x$ as $\delta \to 0,$ the solution to the difference equations \eqref{eq:finite_del-dim} evaluated at $j$ converges to the solution of \eqref{eq:Sigma}-\eqref{eq:A-upper-diag} evaluated at $x.$ Now observe that, recalling \eqref{eq:h-joint-law-alone}, the map
    \[
    (K^{(T)}, v^{(T)})\mapsto \E_{K^{(T)}, v^{(T)}} [\varphi(\sigma^{(T),j}, h^{(T),j})]
    \]
    is continuous. Hence, since $\Sigma^{(T),j}\to\Sigma^{(T)}(x)$ and $f^{t,j}\to f^t(x)$ for $t=1,\dots,T$ as $\delta \to 0,$ by continuity we have 
    \begin{align*}
        \E[\varphi(\sigma^{(T),j}, h^{(T),j})] \to \E[\varphi(\sigma^{(T)}_x, h^{(T)}_x)]
    \end{align*}
    as $\delta\to 0$, proving the theorem.
\end{proof}


\section{Numerical solution and comparison to experiment}\label{sec:numerics}
In this section, we numerically solve our equations, compare to experiment, and make some observations. We remark at the outset that, while so far in this paper we have used $x\in [0,1]$ to refer to the position within each pass and $t\in \N$ to refer to the pass number, we will now also use $t'\in \R_+$ to refer to the overall ``time'' as measured across passes. Hence, if we run for $T$ passes, $t'$ will range in the interval $[0,T]$, and e.g. $t'=2.3$ will correspond to position $x=0.3$ in pass $t=2.$ 

We begin by plotting the energy as a function of time as predicted by our equations (see \Cref{sec:computing-the-energy}) compared to simulations for a selection of temperatures in \Cref{fig:energy}.
\begin{figure}[H]
    \centering
    \includegraphics[width=0.7\linewidth]{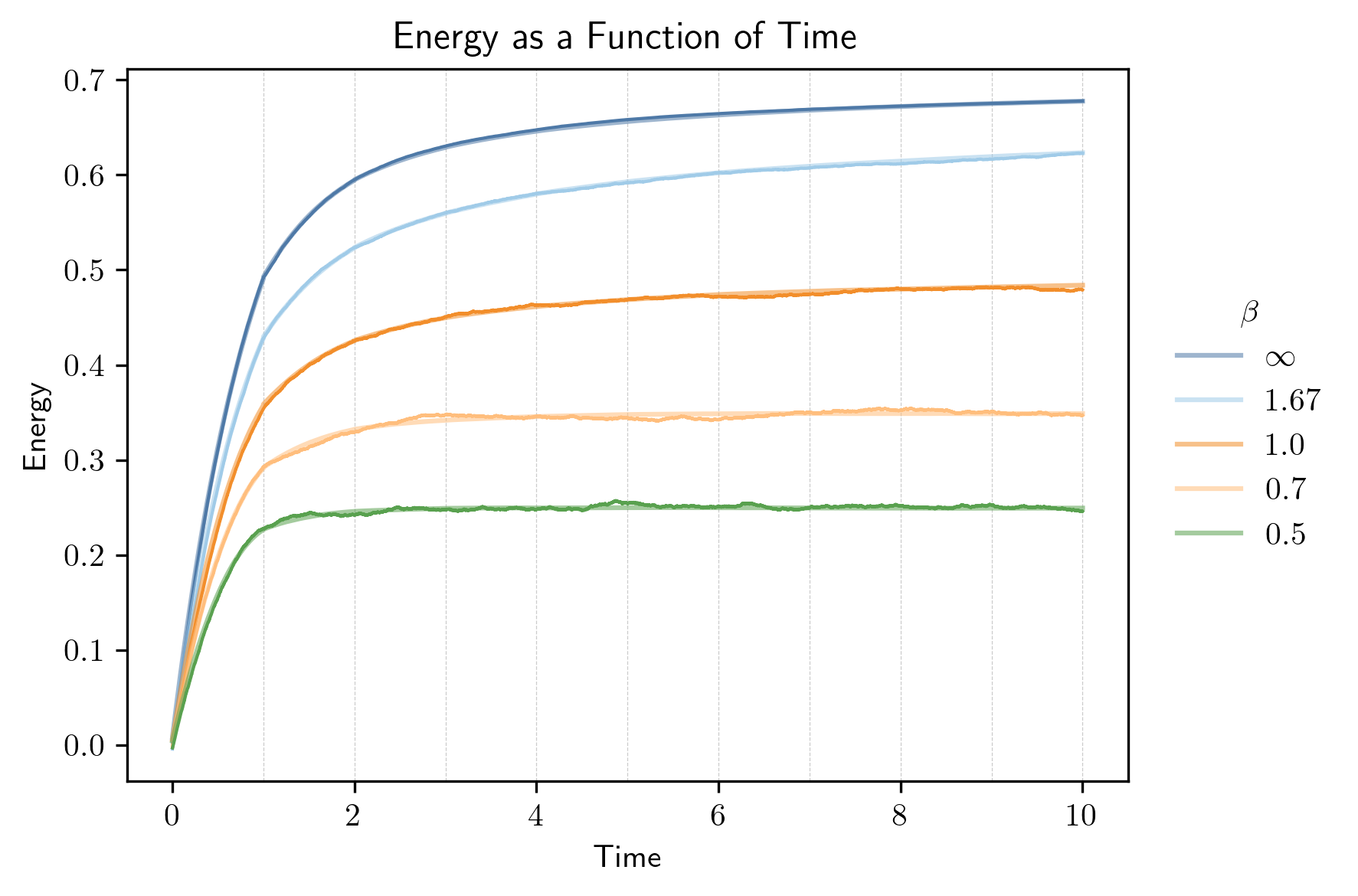}
    \caption{Energy as a function of $t'.$
     The faint colored curves are obtained by numerically solving our equations, while the thinner, darker curves overlaid on top are obtained by averaging $50$ simulation samples with $N=1500.$ Dotted vertical lines mark the transitions between passes. At higher temperatures the simulation curves are noisier; at lower temperatures the two curves almost perfectly overlap.
    }
    \label{fig:energy}
\end{figure}

Next, in \Cref{fig:energy-inverted} we plot only the numerical solution of our equations, but for a larger range of temperatures. This time, we plot the energy as a function of $1/(1+t').$ Separately, for a selection of low temperatures, we plot their difference in energy as a function of $1/(1+t')$ compared to zero temperature. Perhaps surprisingly, our results suggest that there are temperatures that beat the asymptotic value of the energy at zero temperature in the timescales where our equations apply.
\begin{figure}
  \centering
  
  \begin{subfigure}[b]{0.49\linewidth}
    \includegraphics[width=\linewidth]{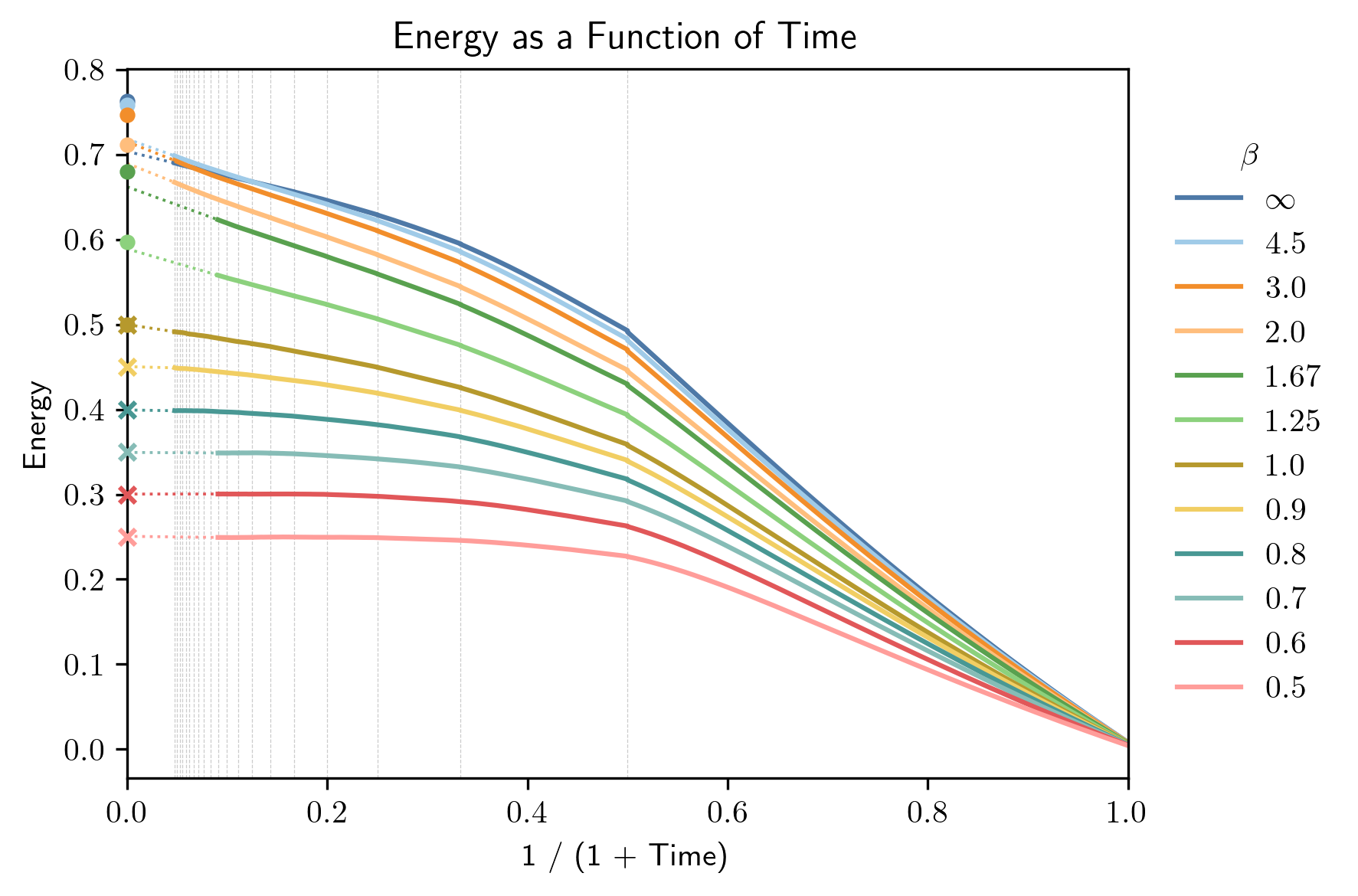}
    \caption{Curves are plotted for 10 or 20 passes. The equilibrium energy values at the given temperature are represented as crosses (at high temperature) or dots (at low temperature) in the $y$-axis. The linear extrapolation of each curve onto the $y$-axis is dotted.}
    \label{fig:energy-inverted-left}
  \end{subfigure}
  \hfill                 
  \begin{subfigure}[b]{0.49\linewidth}
    \includegraphics[width=\linewidth]{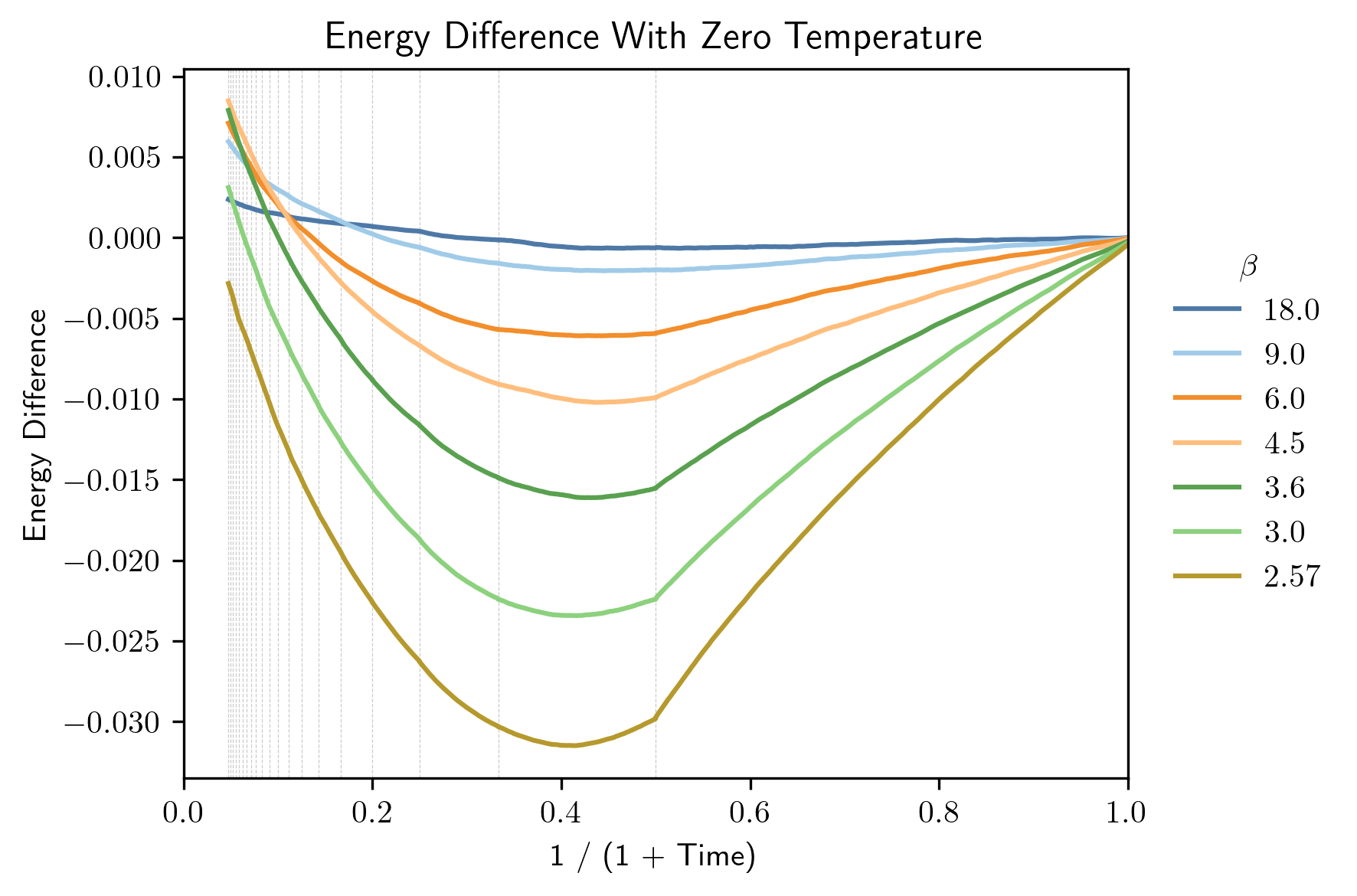}
    \caption{All curves are plotted for 20 passes. If $E_\beta(t')$ is the energy at time $t'\in \R_+$ and inverse temperature $\beta,$ we plot $E_\beta(t') - E_\infty(t')$ as a function of $1/(1+t')$ for a selection of values of $\beta.$  Several values of $\beta$ exceed the energy attained by $\beta=\infty$ at large times.}
    \label{fig:energy-inverted-right}
  \end{subfigure}
  
  \caption{The dotted vertical lines in the background mark transitions between passes. As expected, the curves have non-differentiabilities at the transition points between passes. The equilibrium energies were computed using the method from \cite{rizzo-numerics} ---we thank Tommaso Rizzo for facilitating code to compute the energy values.}
  \label{fig:energy-inverted}
\end{figure}
Note that, in line with the expectation that above the critical temperature ($\beta \leq 1$) the Glauber dynamics equilibrates fast, an extrapolation of the energy curves matches the equilibrium energy value, which is $\beta/2$ at high temperature. On the other hand, at low temperature  ($\beta > 1$), the extrapolated value differs from the equilibrium value, and the gap gets wider the lower the temperature. This is because for $\beta > 1$ the equilibrium of the model is in the spin glass phase, where equilibration requires a larger time. 

We now plot the magnetization as a function of time started from the all-ones vector. Specifically, we initialize $\bfsigma^0$ equal to the all-ones vector, and we track how the mean magnetization $\frac{1}{N}\sum_{i=1}^N \bfsigma^{t,j}_i$ evolves. Due to symmetry, we can instead initialize $\bfsigma^0$ uniformly at random and track the quantity $\frac{1}{N}\brac{\bfsigma^0, \bfsigma^{t,j}}$, which if we take $j$ such that $\delta j \to x$, will converge to
\begin{align}
m(t,x)=m(t'=t+x)=\int_0^x C_{0,t}(y)dy + \int_{1-x}^1 C_{0,t-1}(y)dy. \label{eq:magnetization}
\end{align}
In \Cref{fig:magnetization}, we plot $m(t')$ as a function of $1/(1+t')$, as computed by numerically solving our equations, for a selection of temperatures. As expected in the physics literature, the remnant magnetization $\lim_{t'\to \infty}m(t')$ goes zero quickly at high temperature ($\beta \leq 1$) and remains nonzero at the times considered at low temperature ($\beta >1$). The ``bumps'' observed within passes have a simple explanation. Note that, for the very first spin to be updated in the first pass, the local field is independent of the initial condition, so $C_{0,1}(0)=0.$ Moreover, since all later spins are affected by $\bfsigma^0_0$ only through $\bfsigma^1_0$, we have $C_{0,t}(0)=0$ for all $t.$ Thus, for any given $t$, the curve $x\mapsto C_{0,t}(x)$ starts at zero and grows, so \eqref{eq:magnetization} is integrating in a sliding interval what looks like a ``sawtooth wave'' whose peaks decay with time, resulting in the bumps.
\begin{figure}[H]
    \centering
    \includegraphics[width=0.65\linewidth]{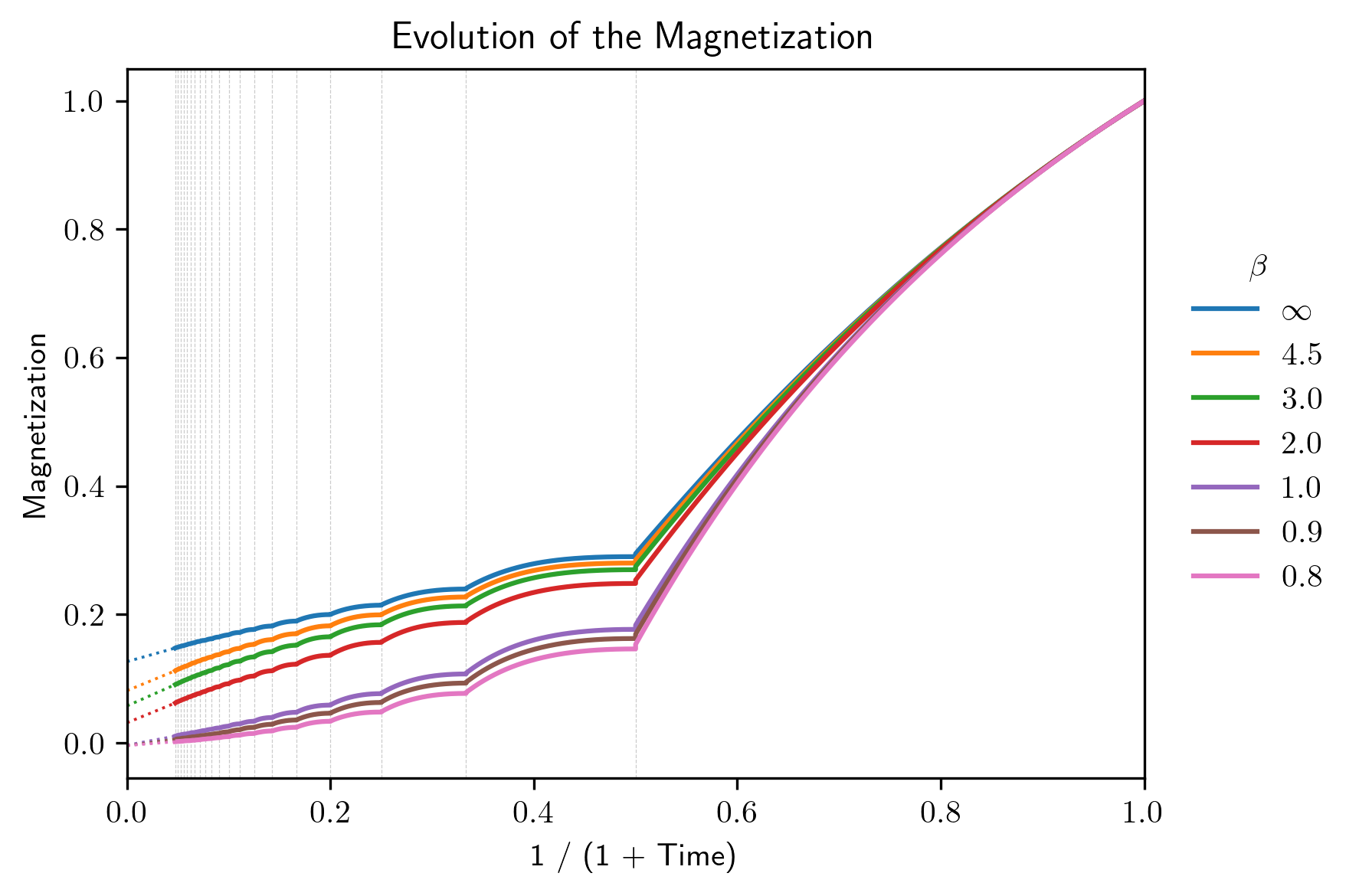}
    \caption{All curves are plotted for 20 passes. The linear extrapolation of each curve onto the $y$-axis is dotted. The dotted vertical lines in the background mark transitions between passes.}
    \label{fig:magnetization}
\end{figure}

Finally, in \Cref{fig:C-prev-together} we plot the correlation with the previous pass $C_{t,t-1}(x)$ at zero temperature ($\beta=\infty$) as a function of the position $x\in [0,1]$ within each pass $t\in \N$. Since $C_{t,t-1}(x)$ will converge to 1 as $t\to\infty$, we plot $\log (1-C_{t,t-1}(x))$ in the $y$-axis. The convergence rate of $C_{t,t-1}(x)$ to 1 appears to be polynomial, so we plot $\int_0^1 \log (1-C_{t,t-1}(y))dy$ as a function of $\log t.$ The fact that this convergence rate appears to be polynomial was surprising to us. We consider it an intriguing open question to verify this phenomenon, either theoretically or by further numerical investigation.

\begin{figure}[H]
  \centering
  \begin{subfigure}[b]{0.49\linewidth}
\includegraphics[width=\linewidth]{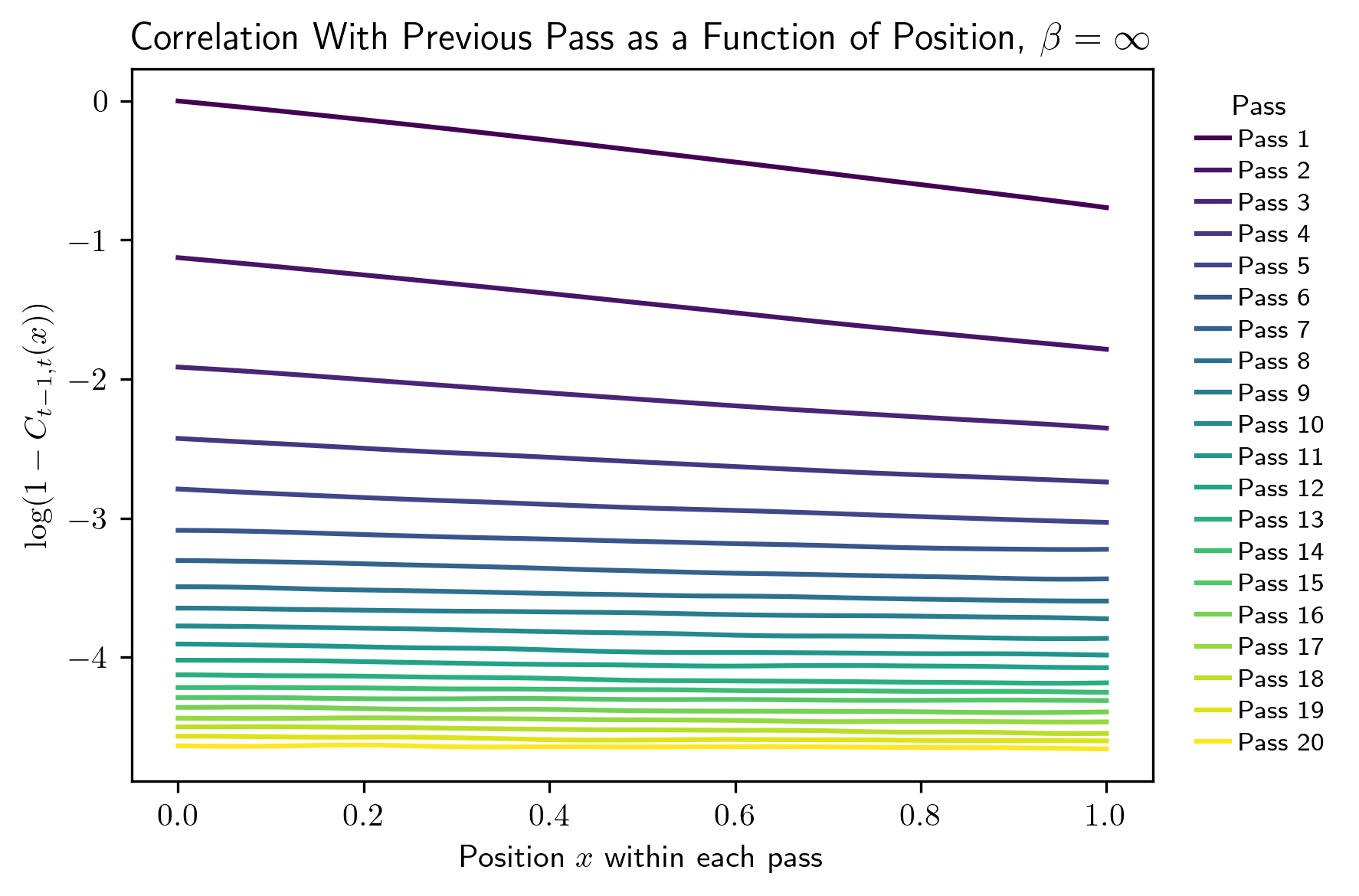}
    \caption{Each curve plots $x\mapsto \log(1-C_{t,t-1}(x))$ for a given value of $t\in \{1,\dots,20\}.$ It appears that the curves diverge to $-\infty$ at a slower than linear rate, suggesting that the convergence of $C_{t,t-1}(x)$ to 1 may be polynomial.}
    \label{fig:C-prev-left}
  \end{subfigure}
  \hfill                 
  \begin{subfigure}[b]{0.49\linewidth}
    \includegraphics[width=0.9\linewidth]{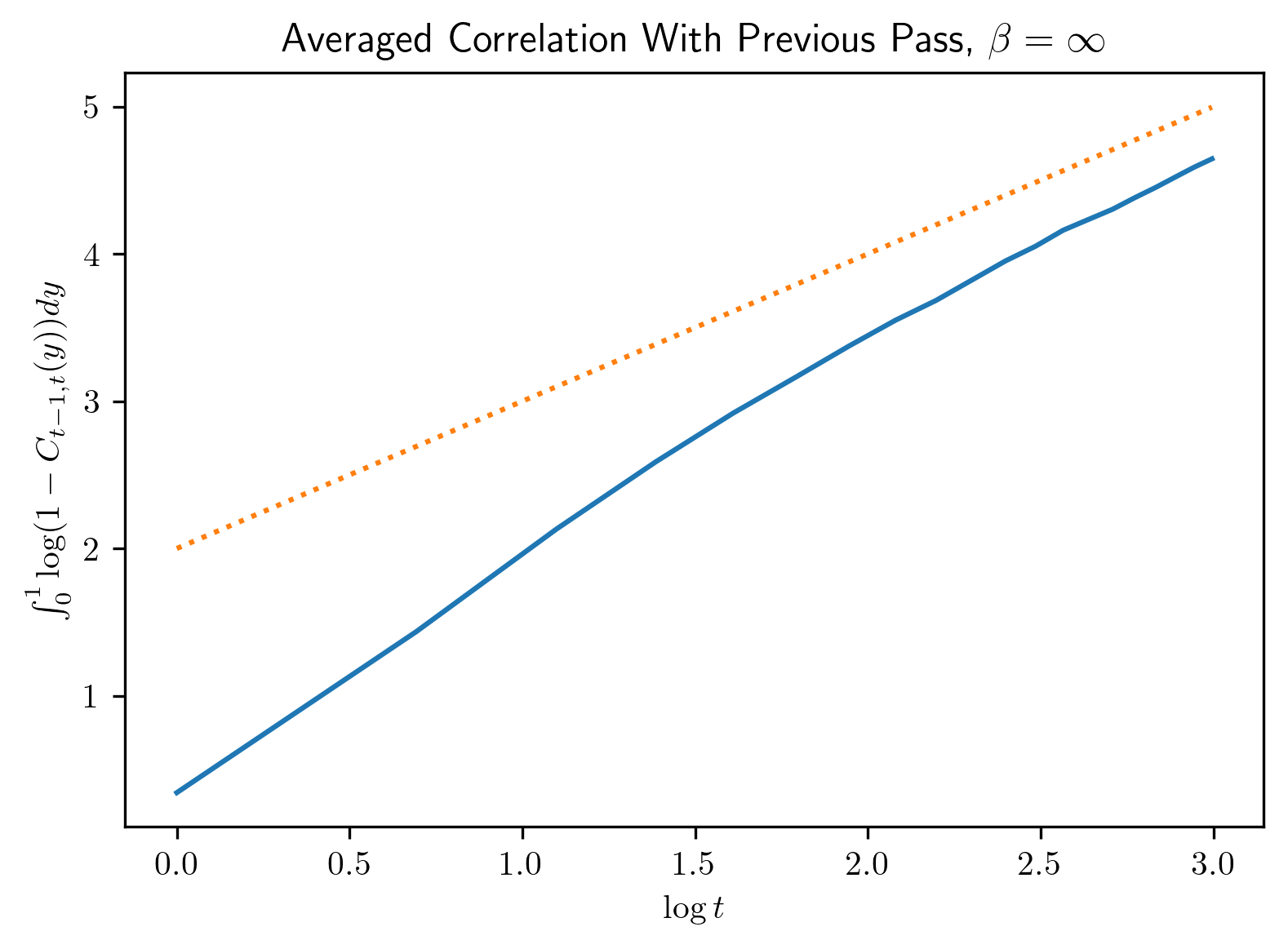}
    \caption{In blue, we plot the average $\int_0^1 \log (1-C_{t,t-1}(y))dy$ of the curves on figure (a) as a function of $\log t$, where $t\in \{1,\dots,20\}$ is the pass number. As a guide, we also include a dotted curve of slope 1, colored orange.}
    \label{fig:C-prev-right}
  \end{subfigure}
  
  \caption{The correlation with the previous spin value at zero temperature.}
  \label{fig:C-prev-together}
\end{figure}

\section{Acknowledgments}
D. Gamarnik and F. Pernice thank Ilias Zadik for many helpful discussions in the early stages of this project. We thank Tommaso Rizzo for facilitating the code to compute the equilibrium energy values that appear in our plots.

This work was carried out while the author Y. Dandi was visiting the Massachusetts Institute of Technology (MIT), with support from the EPFL Doc.Mobility grant, whose financial assistance is gratefully acknowledged.

Several developments in this work took place while the authors Y. Dandi, D. Gamarnik, and L. Zdeborová were visiting the Simons Laufer Mathematical Sciences Institute (SLMath) in Berkeley, California, during the 2025 program Probability and Statistics of Discrete Structures. The authors are grateful to SLMath for its hospitality and support.

David Gamarnik greatfully acknowledges the support of NSF Grant  CISE  2233897.

\bibliography{refs}

\appendix
\section{Appendix}
\subsection{Uniform regularity of difference equations}
In this section we prove the following.
\deltaRegularityProp*
\begin{proof}
The proof is divided in two parts. In the first part, we establish some basic estimates that relate $\norm{f^{t,j}}_2,a^t(j), \lambda^t(j)$ and $\eta^t(j)$ above. In the second part, we use these estimates to inductively show the desired bounds, proving the proposition.

We begin with the first part. Note that, letting $\norm{\cdot}_{\op}$ and $\norm{\cdot}_F$ denote the operator and Frobenius norms, respectively, we have
\begin{align}
    a^t(j) &\leq \norm{(\Sigma^{(t),j})^{-1}}_{\op} \norm{R_{1:t,t}^{(t),j}}_2 + \Bigg(\norm{A^{(t-1),j}}_{\op} \nonumber\\
    &\qquad \qquad+\norm{(\Sigma^{(t),j})^{-1}}_{\op} \norm{R_{1:t,0:t-1}^{(t),j}}_{\op}\Bigg) \norm{(C^{(t-1),j})^{-1}}_{\op}\norm{C_{0:t-1,t}^{(t),j}}_2 \nonumber\\
    &\leq \norm{(\Sigma^{(t),j})^{-1}}_{\op} \norm{R_{1:t,t}^{(t),j}}_2 + \Bigg(\norm{A^{(t-1),j}}_{F} \nonumber\\
    &\qquad \qquad+\norm{(\Sigma^{(t),j})^{-1}}_{\op} \norm{R_{1:t,0:t-1}^{(t),j}}_{F}\Bigg) \norm{(C^{(t-1),j})^{-1}}_{\op}\norm{C_{0:t-1,t}^{(t),j}}_2 \nonumber\\
    &\leq \frac{\sqrt{t}}{\lambda^t(j)} + \left(\sum_{s=1}^{t-1}a^s(j) + \frac{t}{\lambda^t(j)}\right) \frac{\sqrt{t}}{\eta^{t-1}(j)}, \label{eq:at-bound}
\end{align}
where we have used that $R_{s,t}^{(T),j} =\E G^{s,j}\sigma^{t,j} \leq 1$ by Cauchy-Schwarz, since in \eqref{eq:finite_del-dim} the diagonal entries of $\Sigma^{(T),j}$ are always equal to 1 so $\E (G^{s,j})^2=1$. Note that we have by definition
\begin{align}
    \norm{f^{t,j}}_2 &\leq \delta \sum_{k<j} a^t(k) + \delta \sum_{k\geq j}a^{t-1}(k). \label{eq:f-bound}
\end{align}
Moreover, letting $\sign:\R\to \{\pm 1\}$ act entrywise on vectors, we have 
\begin{align*}
    \eta^t(j) &= \min_{\norm{u}_2=1} \E \brac{u, \sigma^{(t),j}}^2 \\
    &\geq \min_{\norm{u}_2=1} \Pr[\sigma^{(t),j}=\sign(u)]\norm{u}_1^2 \\
    &\geq \min_{\tau \in \{\pm 1\}^{t+1}}\Pr[\sigma^{(t),j}=\tau].
\end{align*}
To lower bound $\Pr[\sigma^{(t),j}=\tau]$, consider the filtration $\calG^s(j)$ for $s=0,\dots,T$ generated by the random variables $\sigma^{0,j},G^{1,j},\dots,G^{s,j}$. Note that, for all $s=1,\dots,T,$ conditional on $\calG^{s-1}(j),$ the variable $G^{s,j}$ is centered Gaussian with variance $v = \Var(G^{s,j}|\calG^{s-1}(j))$. Since $\lambda^s(j)\preceq \Sigma^{(s)} \preceq \norm{\Sigma^{(s)}}_F=t$, we have $v\in [\lambda^s(j), t].$ Hence
\begin{align*}
    \Pr[\sigma^s(j)=1|\calG^{s,j}]&\geq \int_{-\infty}^\infty dx \Pr[c(x)=1]\frac{1}{\sqrt{2\pi v}}\exp\left(-\left(x-\brac{f^{s,j}, \sigma^{(s-1),j}}\right)^2/(2v)\right) \\
    &\geq \inf_{|y|\leq \norm{f^{s,j}}_1} \int_{-\infty}^\infty dx \Pr[c(x)=1]\frac{1}{\sqrt{2\pi v}}\exp\left(-\left(x- y\right)^2/(2v)\right) \\
    &\geq \inf_{\substack{|y|\leq \norm{f^{s,j}}_2\sqrt{T} \\ \lambda^{s,j}\leq v'\leq T}} \int_{-\infty}^\infty dx \Pr[c(x)=1]\frac{1}{\sqrt{2\pi v'}}\exp\left(-\left(x- y\right)^2/(2v')\right),
\end{align*}
and similarly for $\Pr[\sigma^{s,j}=-1|\calG^s(j)].$ Hence, since we are assuming non-degeneracy of $c$ (see \Cref{sec:technical-overview}), there exists a smooth function $g:\R_{\geq 0}^2\to \R_{\geq 0}$  depending implicitly on $T$ that satisfies the following:
\begin{itemize}
    
    \item $g(\mu,\lambda) \in [0,1]$ for all $\mu,\lambda \geq 0$ and $g(\mu,\lambda)>0$ whenever $\lambda>0.$
    \item $g$ is monotone decreasing in the first argument and increasing in the second argument.
    \item We have the estimate
    \begin{align*}
    \min\{\Pr[\sigma^{s,j}=1|\calG^s(j)], \Pr[\sigma^{s,j}=1|\calG^s(j)]\}&\geq g(\norm{f^{s,j}}_2, \lambda^s(j)).
\end{align*}
\end{itemize}
Hence 
\begin{align}
    \eta^t(j) &\geq \frac{1}{2}\prod_{s=1}^t g(\norm{f^{s,j}}_2, \lambda^s(j)). \label{eq:eta-bound}
\end{align}
Finally, we have 
\begin{align}
    \lambda^t(j) &\geq \delta \sum_{k<j} \lambda_{\min}\left(C_{1:t,1:t}^{(t),k}\right) +\delta \sum_{k\geq j} \lambda_{\min}\left(C_{0:t-1,0:t-1}^{(t-1),k}\right) \nonumber\\
    &\geq \delta \sum_{k<j}\eta^t(k) +\delta \sum_{k\geq j} \eta^{t-1}(k). \label{eq:lambda-bound}
\end{align}

Now we proceed to the second part of the proof, where we show inductively in $t=1,\dots,T$ that there exists $B_t$ independent of $\delta$ such that $\max_j \norm{f^{t,j}}_2\leq B_t$ and  $\min_j \lambda^t(j) \geq 1/B_t$. Note that this suffices to prove the proposition by \eqref{eq:eta-bound}, \eqref{eq:at-bound}. For the base case $t=1,$ we have $\lambda^1(j)=1$ for all $j$ and
\begin{align*}
    |f^{1,j}| &\leq \delta\sum_{k<j}a^1(k) \\
    &= \delta\sum_{k<j}R_{1,1}^{(1),k} - R_{1,0}^{(1),k}C_{0,1}^{(1),k} \\
    &= \delta\sum_{k<j}R_{1,1}^{(1),k} \\
    &\leq 1,
\end{align*}
so we may take $B_1=1.$ Now assuming inductively that $\max_j \norm{f^{t-1,j}}_2\leq B_{t-1}$ and  $\min_j \lambda^{t-1}(j) \geq 1/B_{t-1},$ combining \eqref{eq:at-bound}-\eqref{eq:lambda-bound}, we have
\begin{align*}
    \lambda^t(j) &\geq \delta \sum_{k<j} \frac{1}{2}\prod_{s=1}^t g(\norm{f^{s,k}}_2,\lambda^s(k)) + \delta \sum_{k\geq j} \frac{1}{2} \prod_{s=1}^{t-1}g(\norm{f^{s,k}}
    _2,\lambda^s(k)) \\
    \norm{f^{t,j}}_2 &\leq \delta \sum_{k<j} \left[\frac{\sqrt{t}}{\lambda^t(k)} + \left(\sum_{s=1}^{t-1}a^s(k) + \frac{t}{\lambda^t(k)}\right)\frac{\sqrt{t}}{\eta^{t-1}(k)}\right] \\
    &\qquad \qquad + \delta\sum_{k\geq j}\left[\frac{\sqrt{t}}{\lambda^{t-1}(k)} + \left(\sum_{s=1}^{t-2}a^s(k) + \frac{t}{\lambda^{t-1}(k)}\right)\frac{\sqrt{t}}{\eta^{t-2}(k)}\right].
\end{align*}
Hence, by our inductive hypothesis, there exist constants $B'_{t-1}$ and $B''_{t-1}$ independent of $\delta$ such that
\begin{align*}
    \lambda^t(j) &\geq \delta \sum_{k<j}g(\norm{f^{t,k}}_2, \lambda^t(k)) B'_{t-1} +  (1-\delta j) B'_{t-1} \\
    \norm{f^{t,j}}_2 &\leq \delta\sum_{k<j} \frac{B''_{t-1}}{\lambda^t(k)} + B''_{t-1}.
\end{align*}
We will now show that there exist functions $h_1,h_2:[0,1]\to \R_{>0}$ independent of $\delta$ such that, for all $j=0,\dots,\lfloor \frac{1}{\delta}\rfloor,$ we have
\begin{align}
    \norm{f^{t,j}}_2 &\leq h_1(\delta j)\label{eq:h1-bound} \\
    \frac{1}{\lambda^t(j)} &\leq h_2(\delta j).\label{eq:h2-bound}
\end{align}
Indeed, suppose that we could find smooth functions $h_1,h_2$ satisfying the integral inequalities
\begin{align}
    h_1(x) &\geq \frac{1}{2}\int_0^x B''_{t-1}h_2(y)dy + B''_{t-1}\label{eq:h1-ineq}\\
    \frac{1}{h_2(x)} &\leq \frac{1}{2}\int_0^x g(h_1(x),1/ h_2(x))B_{t-1}' + (1-x)B_{t-1}' \label{eq:h2-ineq}
\end{align}
for all $x\in [0,1].$ Let us now check that if $h_1,h_2$ satisfy \eqref{eq:h1-ineq}, \eqref{eq:h2-ineq}, then they also satisfy \eqref{eq:h1-bound}, \eqref{eq:h2-bound}. Indeed, we have $h_1(0) \geq B''_{t-1} \geq \norm{f^{t,0}}_2$ and $\frac{1}{h_2(0)} \leq B'_{t-1} \leq \lambda^t(0)$ by assumption, and if we inductively assume that  \eqref{eq:h1-bound}, \eqref{eq:h2-bound} hold for all $k<j,$ we have
\begin{align*}
    \norm{f^{t,j}}_2 &\leq \delta \sum_{k<j} \frac{B''_{t-1}}{\lambda^t(k)} + B''_{t-1} \\
    &\leq \delta \sum_{k<j} B''_{t-1} h_2(\delta k) + + B''_{t-1} \\
    &\leq \frac{1}{2} \int_0^{\delta j} B''_{t-1}h_2(y)dy + B''_{t-1} \\
    &\leq h_1(\delta j)
\end{align*}
and
\begin{align*}
    \lambda^t(j) &\geq \delta\sum_{k<j}g(\norm{f^{t,k}}_2, \lambda^t(k)) B'_{t-1} +  (1-\delta j) B'_{t-1}  \\
    &\geq \delta\sum_{k<j}g(h_1(\delta k), 1/h_2(\delta k)) B'_{t-1} +  (1-\delta j) B'_{t-1} \\
    &\geq \frac{1}{2}\int_0^{\delta j} g(h_1(y),1/ h_2(y)) B'_{t-1}dy +  (1-\delta j) B'_{t-1} \\
    &\geq  \frac{1}{h_2(\delta j)}
\end{align*}
where we have used the monotonicity of $g.$

Hence, to verify that $f^{t,j}, \frac{1}{\lambda^t(j)}$ remain bounded independently of $\delta$, it suffices to check that there exist positive functions $h_1,h_2$ in $[0,1]$ satisfying \eqref{eq:h1-ineq}, \eqref{eq:h2-ineq}. For this, note that the initial value problem
\begin{align}\label{eq:h1h2-system}
\begin{split}
        \frac{d}{dx}h_1(x) &= \frac{B''_{t-1}}{2}h_2(x) \\
    \frac{d}{dx}h_2(x) &= h_2(x)^2B'_{t-1} \left(1 -\frac{1}{2}g(h_1(x),1/h_2(x)) \right)\\
    h_1(0)&=B''_{t-1} \\
    h_2(0) &= \frac{1}{B'_{t-1}}
\end{split}
\end{align}
has a global smooth solution in $[0,1].$ Indeed, local existence and uniqueness follows directly from the Cauchy-Lipschitz Theorem, and singularities in finite time are ruled out as follows. First, since $|\frac{d}{dx}h_2(x)| \leq h_2(x)^2 B'_{t-1}$, we have
\begin{align*}
    |h_2(x)| &\leq \frac{1}{\frac{1}{h_2(0)} - xB'_{t-1}} \\
    &= \frac{1}{B'_{t-1}(1-x)}
\end{align*}
and $h_1(x) \leq B''_{t-1} + \int_0^x \frac{1}{B'_{t-1}(1-y)}dy = B_{t-1}'' -\frac{1}{B'_{t-1}}\log(1-x)$. Hence, by monotonicity, letting 
\[
\eps(x):=\frac{1}{2}g\left(B''_{t-1}- \frac{1}{B'_{t-1}}\log(1-x), \frac{1}{B'_{t-1}(1-x)}\right),
\]
we have $|\frac{d}{dx}h_2(x)| \leq h_2(x)^2 B_{t-1}'(1 - \eps(x))$ so $\int_0^x h_2'(y)/h_2^2(y)dy \leq B'_{t-1}(x- \int_0^x \eps(y)dy)$, which gives
\begin{align*}
    h_2(x) &\leq \frac{1}{B_{t-1}'(1-x+ \int_0^x \eps(y)dy)},
\end{align*}
which is uniformly bounded in $[0,1]$ since $\eps(x)$ is non-zero for all $x\in (0,1).$ This proves that $h_1,h_2$ are smooth and globally defined in $[0,1],$ as desired. From this we conclude the proposition.
\end{proof}

\end{document}